\documentclass{lmcs} 
\pdfoutput=1

\usepackage{lastpage}
\lmcsdoi{18}{2}{8}
\lmcsheading{}{\pageref{LastPage}}{}{}%
{Aug.~07,~2019}{May~10,~2022}{}

\usepackage{hyperref}
\theoremstyle{plain}

\usepackage{enumitem}
\usepackage{listings}
\usepackage{graphicx}
\usepackage{color}
\usepackage{stmaryrd}
\usepackage{subfigure}
\usepackage{tabu}
\usepackage{amssymb}
\usepackage{verbatim}
\usepackage[utf8]{inputenc}
\usepackage[english]{babel}
\usepackage{environ}
\usepackage{mdframed}
\usepackage{latexsym}
\usepackage{tikz}
\usetikzlibrary{arrows}
\usepackage{proof}
\usepackage{amsthm}
\usepackage{soul}
\usepackage[T1]{fontenc}
\usepackage[utf8]{inputenc}
\usepackage{tikz-cd}
\usepackage[english]{babel}
\newcommand{\envalias}[2]{\newenvironment{#1}{\begin{#2}}{\end{#2}}}
\theoremstyle{definition}
\envalias{theorem}{thm}
\envalias{lemma}{lem}
\envalias{proposition}{prop}
\envalias{corollary}{cor}
\envalias{definition}{defi}
\envalias{examplex}{exa}
\envalias{remark}{rem}

\newenvironment{example}
  {\pushQED{\qed}\examplex}
   {\popQED\endexamplex}

\newcommand{\msc}[1]{\mbox{\sc #1}}

\newcommand{\Ra}{\Rightarrow}
	
\begin{document}
\title[Local Validity of Session-Typed Processes]{Circular Proofs as Session-Typed Processes:\texorpdfstring{\\}{} A Local Validity Condition}

\author[F.~Derakhshan]{Farzaneh Derakhshan\lmcsorcid{0000-0002-2156-2606}}[a]	
\address{Philosophy Department, Carnegie Mellon University, Pittsburgh, PA, 15213 USA}	
\email{fderakhs@andrew.cmu.edu}  

\author[F.~Pfenning]{Frank Pfenning\lmcsorcid{0000-0002-8279-5817}}[b]	
\address{Computer Science Department, Carnegie Mellon University, Pittsburgh, PA, 15213 USA}	
\email{fp@cs.cmu.edu}  

\begin{abstract}
  \noindent  Proof theory provides a foundation for studying and reasoning about programming languages, most directly based on the well-known Curry-Howard isomorphism between intuitionistic logic and the typed lambda-calculus. More recently, a correspondence between intuitionistic linear logic and the session-typed pi-calculus has been discovered. In this paper, we establish an extension of the latter correspondence for a fragment of substructural logic with least and greatest fixed points. We describe the computational interpretation of the resulting infinitary proof system as session-typed processes, and provide an effectively decidable local criterion to recognize mutually recursive processes corresponding to valid circular proofs as introduced by Fortier and Santocanale. We show that our algorithm imposes a stricter requirement than Fortier and Santocanale's guard condition, but is local and compositional and therefore more suitable as the basis for a programming language.
\end{abstract}
\maketitle

\section*{Introduction}\label{intro}
Proof theory provides a solid ground for studying and reasoning about programming languages. This logical foundation is mostly based on the well-known Curry-Howard isomorphism \cite{Howard69} that establishes a correspondence between natural deduction and the typed $\lambda$-calculus by mapping propositions to types, proofs to well-typed programs, and proof reduction to computation. More recently, Caires et al. \cite{Caires10concur,Caires16mscs} introduced a correspondence between intuitionistic linear logic \cite{Girard87tapsoft} and the session-typed $\pi$-calculus \cite{Wadler12icfp} that relates linear propositions to session types, proofs in the sequent calculus to concurrent processes, and cut reduction to computation. In this paper, we expand the latter for a fragment of intuitionistic linear logic called \emph {subsingleton logic} in which the antecedent of each sequent consists of at most one formula. We consider the sequent calculus of subsingleton logic with least and greatest fixed points and their corresponding rules \cite{DeYoung16aplas,DeYoung20phd}.  We closely follow Fortier and Santocanale's \cite{Fortier13csl} development in singleton logic, where the antecedent consists of exactly one formula.

Fortier and Santoconale \cite{Fortier13csl,Santocanale02ita} extend the sequent calculus for singleton logic with rules for least and greatest fixed points. A naive extension, however, loses the cut elimination property so they call derivations \emph{pre-proofs}.  \emph{Circular pre-proofs} are distinguished as a subset of derivations which are {\it regular} in the sense that they can be represented as finite trees with loops. They then impose a validity condition (which we call the \emph{FS guard condition}) on pre-proofs to single out a class of pre-proofs that satisfy cut elimination.  Moreover, they provide a cut elimination algorithm and show that it locally terminates on derivations that satisfy the guard condition. In addition, Santocanale and Fortier \cite{Fortier13csl,Santocanale02ita,Santocanale02fossacs, santocanale02JPAA} introduced categorical and game semantics for interpreting cut elimination in singleton logic.

In a related line of research, Baelde et al. \cite{Baelde16csl,Baelde12tocl} add least and greatest fixed points to the sequent calculus for the multiplicative additive fragment of linear logic ($\textit{MALL}$) that results in the loss of the cut elimination property. They also introduced a validity condition to distinguish circular proofs from infinite pre-proofs. Using B\"{u}chi automata, Doumane \cite{Doumane17phd}  showed that the validity condition for identifying circular proofs in $\textit{MALL}$ with fixed points is \textsc{PSPACE} decidable. 
Nollet et al. \cite{Nollet18csl} introduced a polynomial time algorithm for locally identifying a stricter version of Baelde's condition in $\textit{MALL}$ with fixed points.

In this paper,  we study (mutually) recursive session-typed processes and their correspondence with circular pre-proofs in subsingleton logic with fixed points. We introduce an algorithm to check a stricter version of the FS guard condition. Our algorithm is local in the sense that we check validity of each process definition separately, and it is stricter in the sense that it accepts a proper subset of the proofs recognized by the FS guard condition.
We further introduce a synchronous computational semantics of cut reduction in subsingleton logic with fixed points in the context of session types, based on a key step in Fortier and Santocanale's cut elimination algorithm, which is compatible with prior operational interpretations of session-typed programming languages~\cite{Toninho13esop}. We show preservation and a strong version of the progress property that ensures that each valid process communicates along its left or right interface in a finite number of steps. A key aspect of our type system is that validity is a compositional property (as we generally expect from type systems) so that the composition of valid programs defined over the same signature are also valid and therefore also satisfy strong progress. In other words, we identify a set of processes such that their corresponding derivations are not only closed under cut elimination, but also closed under cut introduction (i.e. strong progress is preserved when processes are joined by cut).

In the session type system, a singleton judgment $A \vdash B$ is annotated as $x:A \vdash \mathtt{P} :: (y:B)$ which is understood as: process $\mathtt{P}$ uses a service of type $A$ offered along channel $x$ by a process on its left and provides a service of type $B$ along channel $y$ to a process on its right \cite{DeYoung16aplas}. The left and right interfaces of a process in the session type system inherit the symmetry of the left and right rules in the sequent calculus. Each process interacts with other processes along its pair of left and right interfaces, which correspond to the left and right sides of a sequent. For example, two processes $\mathtt{P}$ and $\mathtt{Q}$ with the typing $ x:A \vdash \mathtt{P} :: (y:B)$ and $y:B \vdash \mathtt{Q} :: (z:C)$ can be composed so they interact with each other using channel $y$. Process $\mathtt{P}$ provides a service of type $B$ and offers it along channel $y$ and process $\mathtt{Q}$ uses this service to provide its own service of type $C$. This interaction along channel $y$ can be of two forms: (i) process $\mathtt{P}$ sends a message to the right, and process $\mathtt{Q}$  receives it from the left, or (ii) process $\mathtt{Q}$ sends a message to the left and process $\mathtt{P}$ receives it from the right. In the first case, the session type $B$ is a \emph{positive} type, and in the second case it is a \emph{negative} type. Least fixed points have a positive polarity while greatest fixed points are negative \cite{Lindley16icfp}. As we will see in Sections~\ref{alg1} and~\ref{priority}, due to the interactive nature of computation some types that would be considered ``empty'' (that is, have no closed values) may still be of interest here.

DeYoung and Pfenning \cite{DeYoung16aplas,Pfenning16lectures} provide a representation of Turing machines in the session-typed rule system of subsingleton logic with general equirecursive types. This shows that cut reduction on circular pre-proofs in subsingleton logic with equirecursive types has the computational power of Turing machines. Using this encoding on isorecursive types, we show that recognizing all programs that satisfy a strong progress property is undecidable, since this property can be encoded as termination of a Turing machine on a given input. However, with our algorithm, we can decide validity of a subset of Turing machines represented as session-typed processes in subsingleton logic with isorecursive fixed points.

In summary, the principal contribution of our paper is to extend the Curry-Howard interpretation of proofs in subsingleton logic as communicating processes to include least and greatest fixed points. A circular proof is thus represented as a collection of mutually recursive process definitions. We develop a compositional criterion for validity of such programs, which is local in the sense that each process definition can be checked independently.  Local validity in turn implies a strong progress property on programs and cut elimination on the circular proofs they correspond to.

The structure of the remainder of the paper is as follows.  In Section~\ref{lfixed} we introduce subsingleton logic with fixed points, and then examine it in the context of session-typed communication (Section~\ref{pre}).  We provide a process notation with a synchronous operational semantics in Section~\ref{operat} and a range of examples in Section~\ref{session}.  We then develop a local validity condition through a sequence of refinements in Sections~\ref{alg1}--\ref{modcut}.  We capture this condition on infinitary proofs in Section~\ref{rules} and reduce it to a finitary algorithm in Section~\ref{alg}.  We prove that local validity implies Fortier and Santocanale's guard condition (Section~\ref{Guard}) and therefore cut elimination.  In Section~\ref{semantics} we explore the computational consequences of this, including the strong progress property, which states that every valid configuration of processes will either be empty or attempt to communicate along external channels after a finite number of steps.  We conclude by illustrating some limitations of our algorithm (Section~\ref{negative}) and pointing to some additional related and
future work (Section~\ref{conclusion}).

\setlength{\inferLineSkip}{5pt} 
\setlength{\inferLabelSkip}{1pt}

\section{Subsingleton Logic with Fixed Points} \label{lfixed}

Subsingleton logic is a fragment of intuitionistic linear logic \cite{Girard87tapsoft,Chang03tr} in which the antecedent
and succedent of each judgment consist of at most one proposition.  This reduces consideration to the additive connectives and multiplicative units, because the left or right rules of other connectives would violate this restriction.  The expressive power of pure subsingleton logic is rather limited, among other things due to the absence of the exponential $!A$.  However, we can recover significant expressive power by adding least and greatest fixed points, which can be done without violating the subsingleton restriction.  We think of subsingleton logic as a laboratory in which to study the properties and behaviors of least and greatest fixed points in their simplest nontrivial form, following the seminal work of Fortier and Santocanale~\cite{Fortier13csl}.

The syntax of propositions follows the grammar
\[
A,B ::= A \oplus B \mid 0 \mid A \& B \mid \top \mid 1 \mid \bot \mid t
\]%
where $t$ ranges over a set of propositional variables denoting least or greatest fixed points.  Rather than including these directly as $\mu t.\, A$ and $\nu t.\, A$, we define them in a \emph{signature} $\Sigma$ which records some important additional information, namely their relative \emph{priority} ($i \in \mathbb{N}$).
\[\Sigma ::= \cdot \mid \Sigma, t=^{i}_{\mu} A \mid \Sigma, t=^{i}_\nu A,\]%
with the conditions that
\begin{itemize}
    \item if $  t=^{i}_{a} A \in \Sigma$ and $t'=^{i}_b B \in \Sigma$, then $a=b$, and
    \item if $  t=^{i}_{a} A \in \Sigma$ and $t=^{j}_b B \in \Sigma$, then $i=j$ and $A=B$.
\end{itemize}

For a fixed point $t$ defined as $ t=^{i}_{a} A$ in $\Sigma $ the subscript $a$ is the {\it polarity} of $t$: if $a=\mu$, then $t$ is a fixed point with {\it positive} polarity and if $a=\nu$, then it is of {\it negative} polarity.  Finitely representable least fixed points (e.g., natural numbers and lists) can be represented in this system as defined propositional variables with positive polarity, while the potentially infinite greatest fixed points (e.g., streams and infinite depth trees) are represented as those with negative polarity.

The superscript $i$ is the {\it priority} of $t$. Fortier and Santocanale interpreted the priority of fixed points in their system as the order in which the least and greatest fixed point equations are solved in the semantics \cite{Fortier13csl}.  We use them syntactically as central information to determine local validity of circular proofs. We write $p(t)=i$ for the priority of $t$, and $\epsilon(i)= a$ for the polarity of propositional variable $t$ with priority $i$. The condition on $\Sigma$ ensures that $\epsilon$ is a well-defined function.

The basic judgment of the subsingleton sequent calculus has the form $\omega \vdash_\Sigma \gamma$, where $\omega$ and $\gamma$ are either empty or a single proposition $A$ and $\Sigma$ is a signature.  Since the signature never changes in the rules, we omit it from the turnstile symbol. The rules of subsingleton  logic with fixed points are summarized in Figure~\ref{fig:rule}, constituting a slight generalization of Fortier and Santocanale's.  When the fixed points in the last row are included, this set of rules must be interpreted as \emph{infinitary}, meaning that a judgment may have an infinite derivation in this system.

\begin{figure}
  \[
\begin{tabular}{c c c c} 
      \multicolumn{2}{c}{\infer[\msc{Id}_A]{A \vdash A}{}} & 
   \multicolumn{2}{c}{ \infer[\msc{Cut}_A]  {\omega \vdash \gamma}{\omega \vdash A & A \vdash \gamma}}  \\[6pt]
    \infer[\oplus R_1]{\omega \vdash A \oplus B}{\omega \vdash A} & \infer[\oplus R_2]{\omega \vdash A \oplus B}{\omega \vdash B} &  \multicolumn{2}{c}{\infer[\oplus L]{A \oplus B \vdash \gamma}{ A \vdash \gamma & B \vdash \gamma}} \\[6pt]
    \multicolumn{2}{c}{\infer[\&R]{\omega \vdash A \& B}{\omega \vdash A & \omega \vdash B}} & \infer[\&L_1]{A \& B \vdash \gamma}{A  \vdash \gamma} & \infer[\&L_2]{A \& B \vdash \gamma}{B  \vdash \gamma}  \\[6pt]
    \infer[1R]{\cdot \vdash 1}{} & \infer[1L]{1 \vdash \gamma}{\cdot \vdash \gamma} &
    \infer[{\bot}R]{\omega \vdash \bot}{\omega \vdash \cdot } & \infer[{\bot}L]{\bot \vdash \cdot}{} \\[8pt]
    \infer[\mu R]{\omega \vdash t}{\omega \vdash A \;\;\, t{=_{\mu}^i} A \in\Sigma}& \infer[\mu L]{t \vdash \gamma}{A \vdash \gamma \;\;\, t{=_{\mu}^i} A \in \Sigma}&
     \infer[\nu R]{\omega \vdash t}{\omega \vdash A \;\;\, t{=_{\nu}^i} A \in \Sigma} & \infer[\nu L]{t \vdash \gamma}{A \vdash \gamma \;\;\, t{=_{\nu}^i} A \in \Sigma}
\end{tabular}
\]
  \caption{Infinitary sequent calculus for subsingleton logic with fixed points.}
  \label{fig:rule}
\end{figure}

Even a cut-free derivation may be of infinite length since each defined propositional variable may be unfolded infinitely many times. Also, cut elimination no longer holds for the derivations after adding fixed point rules. What the rules define then are the so-called \emph{pre-proofs}. In particular, we are interested in \emph{circular pre-proofs}, which are the pre-proofs that can be illustrated as finite trees with loops \cite{Doumane17phd}.

Fortier and Santocanale \cite{Fortier13csl} introduced a {\it guard condition} for identifying a subset of circular \emph{proofs} among all infinite pre-proofs in singleton logic with fixed points. Their guard condition states that every cycle should be supported by the unfolding of a positive (least) fixed point on the antecedent or a negative (greatest) fixed point on the succedent. Since they allow mutual dependency of least and greatest fixed points, they need to consider the priority of each fixed point as well. The supporting fixed point for each cycle has to be of the highest priority among all fixed points that are unfolded infinitely in the cycle. 
They proved that the guarded subset of derivations enjoys the cut elimination property; in particular, a cut composing any two guarded derivations can be eliminated effectively.

As an example, the following circular pre-proof defined on the signature $\mathsf{nat}=^{1}_{\mu} \mathsf{1} \oplus \mathsf{nat}$ depicts an infinite pre-proof that consists of repetitive application of $\mu R$ followed by $\oplus R$:
\begin{center}
\begin{tikzpicture}
\node[](2){$\infer[\mu R]{\cdot \vdash \mathsf{nat}}{\infer[\oplus R]{\cdot \vdash \mathsf{1} \oplus \mathsf{nat}}{\cdot \vdash \mathsf{nat}}}$};
\draw[->] (-1,0.6).. controls (-3,.6) and (-3,-.7) .. (-1,-.7);
\end{tikzpicture}
\end{center}
It is not guarded and turns out not to be locally valid either. On the other hand, on the signature $\mathsf{conat}=_{\nu}^{1} \mathsf{1}\ \& \ \mathsf{conat}$, we can define a circular pre-proof using greatest fixed points that is guarded and locally valid:
\begin{center}
\begin{tikzpicture}
\node[](2){$\infer[\nu R]{\cdot \vdash \mathsf{conat}}{\infer[\& R]{\cdot \vdash \mathsf{1} \&\ \mathsf{conat}}{\infer[\mathsf{1}R]{\cdot \vdash \mathsf{1}}{} & \cdot \vdash \mathsf{conat}}}$};
\draw[->] (1.3,0.5).. controls (3.1,0.5) and (3.2,-0.9) .. (1,-0.9);
\end{tikzpicture}
\end{center}

\section{Fixed Points in the Context of Session Types}\label{pre} 

Session types~\cite{Honda93concur,Honda98esop} describe the communication behavior of interacting processes.  \emph{Binary session types}, where each channel has two endpoints, have been recognized as arising from linear logic (either in its intuitionistic~\cite{Caires10concur,Caires16mscs} or classical~\cite{Wadler12icfp} formulation) by a Curry-Howard interpretation of propositions as types, proofs as programs, and cut reduction as communication.  In the context of programming, recursive session types have also been considered~\cite{Toninho13esop,Lindley16icfp}, and they seem to fit smoothly, just as recursive types fit well into functional programming languages. However, they come at a price, since we abandon the Curry-Howard correspondence.

In this paper we show that this is not necessarily the case: we can remain on a sound logical footing as long as we (a) refine general recursive session types into least and greatest fixed points, (b) are prepared to accept circular proofs, and (c) impose conditions under which recursively defined processes correspond to valid circular proofs. General (nonlinear) type theory has followed a similar path, isolating inductive and coinductive types with a variety of conditions to ensure validity of proofs.  In the setting of subsingleton logic, however, we find many more symmetries than typically present in traditional type theories, which appear to be naturally biased towards least fixed points and inductive reasoning.

Under the Curry-Howard interpretation a subsingleton judgment $A \vdash_\Sigma B$ is annotated as 
\[ x:A \vdash_\Sigma \mathtt{P} :: (y:B)\]
where $x$ and $y$ are two different channels and $A$ and $B$ are their corresponding session types. One can understand this judgment as: process $\mathtt{P}$ provides a service of type $B$ along channel $y$ while using channel $x$ of type $A$, a service that is provided by another process along channel $x$ \cite{DeYoung16aplas}. We can form a chain of processes $\mathtt{P}_0, \mathtt{P}_1, \cdots, \mathtt{P}_{n}$ with the typing
\[\cdot \vdash \mathtt{P}_0 :: (x_0:A_0),\quad x_0:A_0 \vdash \mathtt{P}_1 :: (x_1:A_1),\quad \cdots\quad x_{n-1}:A_{n-1} \vdash \mathtt{P}_{n} :: (x_{n}:A_{n})\]
which we write as
\[
\mathtt{P}_0 \mid_{x_0} \mathtt{P}_1 \mid_{x_1} \cdots\; \mid_{x_{n-1}} \mathtt{P}_n
\]
in analogy with the notation for parallel composition for processes $P \mid Q$, although here it is not commutative.  In such a chain, process $\mathtt{P}_{i+1}$ uses a service of type $A_i$ provided by the process $\mathtt{P}_{i}$ along the channel $x_i$, and provides its own service of type $A_{i+1}$ along the channel $x_{i+1}$. Process $\mathtt{P}_0$ provides a service of type $A_0$ along channel $x_0$ without using any services. So, a process in the session type system, instead of being reduced to a value as in functional programming, interacts with both its left and right interfaces by sending and receiving messages. Processes $\mathtt{P}_i$ and $\mathtt{P}_{i+1}$, for example, communicate with each other along the channel $x_i$ of type $A_i$: if process $\mathtt{P}_{i}$ sends a message along channel $x_i$ to the right and process $\mathtt{P}_{i+1}$ receives it from the left (along the same channel), session type $A_i$ is called a \emph{positive} type. Conversely, if process $\mathtt{P}_{i+1}$ sends a message along channel $x_i$ to the left and process $\mathtt{P}_{i}$ receives it from the right (along the same channel), session type $A_i$ is called a \emph{negative} type. In Section~\ref{session} we show in detail that this symmetric behavior of left and right session types results in a symmetric behaviour of least and greatest fixed point types.

In general, in a chain of processes, the leftmost type may not be empty.  Also, strictly speaking, the names of the channels are redundant since every process has two distinguished ports: one to the left and one to the right, either one of which may be empty.  Because of this, we may sometimes omit the channel name, but in the theory we present in this paper it is convenient to always refer to communication channels by unique names.

For programming examples, it is helpful to allow not just two, but any finite number of alternatives for internal ($\oplus$) and external ($\&$) choice. Such finitary choices can equally well be interpreted as propositions, so this is not a departure from the proofs as programs interpretation.

\begin{definition}\label{signature} We define {\it session types} with the following grammar, where $L$ ranges over
  finite sets of labels denoted by $\ell$ and $k$.
\[A::= {\oplus}\{\ell : A_\ell\}_{\ell \in L} \mid {\&}\{\ell : A_\ell\}_{\ell \in L} \mid 1 \mid \bot \mid t\]
where $t$ ranges over a denumerable set of type variables whose definition is given in a signature $\Sigma$ as before.
The binary disjunction and conjunction are defined as $A \oplus B =\oplus\{\pi_1:A, \pi_2:B\}$ and $A \& B =\& \{\pi_1:A, \pi_2:B\}$, respectively. Similarly, we define $0=\oplus\{\}$ and $\top=\& \{\}$.
\end{definition}

As a first programming-related example, consider natural numbers in unary form ($\mathsf{nat}$) and  a type to demand access to a number if desired ($\mathsf{ctrl}$).

\begin{example}[Natural numbers on demand]\label{natstream}
 \begin{align*}
& \mathsf{nat}=^{1}_{\mu} \oplus\{ \mathit{z}:\mathsf{1}, \mathit{s}: \mathsf{nat}\}\\
& \mathsf{ctrl}=^{2}_{\nu} \& \{\mathit{now} : \mathsf{nat}, \ \ \mathit{notyet}: \mathsf{ctrl}\}
 \end{align*}
 In this example, $\Sigma$ consists of an inductive and a coinductive type; these are, respectively: (i) the type of natural numbers ($\mathsf{nat}$) built using two constructors for zero (${\mathit{z}}$) and successor ($\mathit{s}$), and (ii) a type to demand access to a number if desired ($\mathsf{ctrl}$) defined using two destructors for $\mathit{now}$ to obtain the number and $\mathit{notyet}$ to postpone access, possibly indefinitely. Here, the priorities of $\mathsf{nat}$ and $\mathsf{ctrl}$ are, respectively, $1$ and $2$, understood as ``\emph{$\mathsf{nat}$ has higher priority than $\mathsf{ctrl}$}''.
\end{example}

  \begin{example}[Binary numbers in standard form]
  As another example consider the signature with two types with positive polarity and the same priority: $\mathsf{std}$ and $\mathsf{pos}$. Here, $\mathsf{std}$ is the type of standard bit strings, i.e., bit strings terminated with $\$$ without any leading $0$ bits, and $\mathsf{pos}$ is the type of positive standard bit strings, i.e., all standard bit strings except $\$$.  Note that in our representation the least significant bit is sent first.
   \begin{align*}
            & \mathsf{std}=^{1}_{\mu} \oplus\{ \mathit{b0}:\mathsf{pos}, \mathit{b1}:\mathsf{std},  \$:\mathsf{1}\}\\
            & \mathsf{pos}=^{1}_{\mu} \oplus \{\mathit{b0}:\mathsf{pos}, \mathit{b1}:\mathsf{std}\}\qedhere
 \end{align*}

\end{example}
\begin{example}[Bits and cobits]
\label{bittype}
\begin{align*}
             &\mathsf{bits}=^{1}_{\mu} \oplus\{ \mathit{b0}:\mathsf{bits},\ \mathit{b1}:\mathsf{bits}\} \\
             &\mathsf{cobits}=^{2}_{\nu} \&\{ \mathit{b0}:\mathsf{cobits},\ \mathit{b1}:\mathsf{cobits}\}
\end{align*}
 In a functional language, the type $\mathsf{cobits}$ would be a greatest fixed point (an infinite stream of bits), while $\mathsf{bits}$ is recognized as an empty type. However, in the session type system, we treat them in a symmetric way.  $\mathsf{bits}$ is an infinite sequence of bits with positive polarity. And its dual type, $\mathsf{cobits}$, is an infinite stream of bits with negative polarity. In Examples~\ref{bits} and~\ref{cobits}, in Section~\ref{alg1}, we further illustrate the symmetry of these types by providing two recursive processes having them as their interfaces.  Even though we can, for example, write transducers of type $\mathsf{bits} \vdash \mathsf{bits}$ inside the language, we cannot write a \emph{valid} process of type $\cdot \vdash \mathsf{bits}$ that \emph{produces} an infinite stream of bits.
\end{example}

\section{A Synchronous Operational Semantics}
\label{operat}

The operational semantics for process expressions under the proofs-as-programs interpretation of linear logic has been treated exhaustively elsewhere~\cite{Caires10concur,Caires16mscs,Toninho13esop,Griffith16phd}.  We therefore only briefly sketch the operational semantics here.  Communication is \emph{synchronous}, which means both sender and receiver block until they synchronize. Asynchronous communication can be modeled using a process with just one output action followed by forwarding~\cite{Gay10jfp,DeYoung12csl}.  However, a significant difference to much prior work is that we treat types in an isorecursive way, that is, a message is sent to unfold the definition of a type $t$.  This message is written as $\mu_t$ for a least fixed point and $\nu_t$ for a greatest fixed point. The language of process expressions and their operational semantics presented in this section is suitable for general isorecursive types, if those are desired in an implementation. The resulting language then satisfies a weaker progress property sometimes called \emph{global progress} (see, for example, \cite{Caires10concur}).

\begin{definition}\label{process} Processes are defined as follows over the signature $\Sigma$:
\[\begin{aligned}
P,Q ::=\ \   
&   { y \leftarrow x} & identity\\& \mid  ( x \leftarrow P_x ; Q_x) & cut \\
& \mid Lx.k; P  \mid \mathbf{case}\, Rx\ (\ell \Ra Q_{\ell})_{\ell \in L}& \& \{\ell:A_\ell\}_{\ell \in L}\\
& \mid Rx.k; P  \mid \mathbf{case}\, Lx\ (\ell \Ra Q_{\ell})_{\ell \in L} & \phantom{something} \oplus \{\ell:A_\ell\}_{\ell \in L}  \\
& \mid \mathbf{close}\, Rx  \mid \mathbf{wait}\, Lx; Q  & 1\\
& \mid \mathbf{close}\, Lx \mid \mathbf{wait} \,Rx;Q & \bot \\
& \mid  Rx.\mu_{t}; P \mid \mathbf{case}\, Lx\ (\mu_{t} \Ra P) & t=_{\mu}A\\
& \mid Lx. \nu_{t}; P \mid \mathbf{case}\, Rx\ (\nu_{t} \Ra P) & t=_{\nu} A\\
& \mid { \bar{y} \leftarrow X \leftarrow \bar{x}}  & call
\end{aligned}\]
\\ where $X,Y, \ldots$ are {\it  process variables}, $x,y, \ldots$ are channel names, and $\bar{x}$ ($\bar{y}$) is either empty or $x$ ($y$). In $( x \leftarrow P_x ; Q_x)$, the variable $x$ is bound and represents a channel connecting the two processes. Throughout the paper, we may subscript processes with bound variables if they are allowed to occur in them. In the programming examples, we may write $ y\leftarrow X \leftarrow x ; Q_y$ instead of $(y  \leftarrow  (y\leftarrow X \leftarrow x) ; Q_y)$ when a new process executing $X$ is spawned.
The left and right session types a process interacts with are uniquely labelled  with channel names: \[ x:A \vdash P :: (y:B).\]
We read this as 
\begin{center} {\it Process $P$ uses channel $x$ of type $A$ and provides a service of type $B$ along channel $y$.}
\end{center}
However, since a process might not use any service provided along its left channel, e.g. $\cdot \vdash P :: (y: B),$ or it might not provide any service along its right channel, e.g. $ x:A \vdash Q:: (\cdot) $, the labelling of processes is generalized to be of the form:
\[\bar{x}:\omega \vdash P::(\bar{y}:\gamma),\] where $\bar{x}$ ($\bar{y}$) is either empty or $x$ ($y$), and $\omega$ ($\gamma$) is empty given that $\bar{x}$ ($\bar{y}$) is empty.

Process definitions are of the form $\bar{x}:\omega \vdash X=P_{\bar{x},\bar{y}} ::(\bar{y}:\gamma)$ representing that variable $X$ is defined as process $P$. A {\it program} $\mathcal{P}$ is defined as a pair $ \langle V,S \rangle$, where $V$ is a finite set of process definitions, and $S$ is the {\it main} process variable.\footnote{For the sake of simplicity, we may only refer to process variables in $V$ when their definition is clear from the context.} Figure~\ref{fig:annotrule} shows the logical rules annotated with processes in the context of session types.  This set of rules inherits the full symmetry of its underlying sequent calculus. They interpret \emph{pre-proofs}: as can be seen in the rule $\msc{Def}$, the typing rules inherit the infinitary nature of deductions from the logical rules in Figure~\ref{fig:rule} and are therefore not directly useful for type checking.\footnote{The rule $\msc{Def}$ corresponds to forming cycles in the circular derivations of the system of Figure~\ref{fig:rule}.} We obtain a finitary system to check \emph{circular} pre-proofs by removing the first premise from the $\msc{Def}$ rule and checking each process definition in $V$ separately, under the hypothesis that all process definitions are well-typed. Since the system is entirely syntax-directed we may sometimes equate (well-typed) programs with their typing derivations.  This system rules out communication mismatches without forcing processes to actually communicate along their external channels.  In order to also enforce communication the rules need to track additional information (see rules in Figures~\ref{fig:stp-order} (infinitary) and~\ref{fig:validity} (finitary) in Sections~\ref{rules} and~\ref{alg}).

All processes we consider in this paper provide a service along their right channel so in the remainder of the paper we restrict the sequents to be of the form $\bar{x}:\omega \vdash P::(y:A)$. We therefore do not need to consider the rules for type $\perp$ anymore, but the results of this paper easily generalize to the fully symmetric calculus.
\end{definition}

\begin{figure}
\[
\begin{tabular}{c} 
\infer[\msc{Id}]{x : A \vdash  y  \leftarrow x  :: (y : A)}{}\\[6pt] \infer[\msc{Cut}^{w}]{ \bar{x} : \omega \vdash   (w \leftarrow P_{w} ; Q_{w}) :: ( \bar{y} : \gamma)}{ \bar{x} :  \omega \vdash P_{w} ::(w:A) & w: A \vdash Q_{w} :: (\bar{y}: \gamma)} \\[6pt]
\infer[\oplus R]{\bar{x}:\omega \vdash Ry.k; P :: (y: \oplus\{\ell:A_\ell\}_{\ell \in L})}{\bar{x}: \omega \vdash P :: (y: A_{k}) \quad (k \in L)} \\[6pt]
\infer[\oplus L]{x:\oplus\{ \ell:A_\ell \}_{ \ell \in L} \vdash \mathbf{case}\, Lx \ (\ell\Ra P_{\ell})_{\ell \in L}:: (\bar{y}: \gamma)}{\forall \ell\in L \quad x:A_{\ell} \vdash P_\ell :: (\bar{y}:\gamma)}\\[6pt]
\infer[\& R]{\bar{x} : \omega \vdash \mathbf{case}\,  Ry \ (\ell \Ra P_\ell)_{\ell \in L} :: (y : \& \{\ell:A_\ell\}_{\ell \in L})}{\bar{x} : \omega \vdash  P_\ell :: (y :A_{\ell}) \quad \forall \ell\in L}   \\[6pt]
\infer[\& L]{x : \&\{ \ell:A_\ell \}_{ \ell \in L} \vdash  Lx .k; P :: ( \bar{y} : \gamma)}{k\in L \quad x :  A_{k} \vdash   P :: ( \bar{y} : \gamma)}
\end{tabular}
\]
\[
\begin{tabular}{c c} 
\infer[1R]{\cdot \vdash  \mathbf{close}\, Ry :: (y : 1)}{} &
\infer[1L]{x : 1 \vdash  \mathbf{wait}\, Lx;Q :: (\bar{y} : \gamma)}{  . \vdash  Q :: (\bar{y} : \gamma) }  \\[6pt]
\infer[ {\bot} R]{\bar{x}:A \vdash  \mathbf{wait} \, Ry; Q :: (y : \bot)}{\bar{x}:A \vdash Q:: \cdot} &
\infer[{\perp}L]{x : \bot \vdash  \mathbf{close} \, Lx :: \cdot}{}  \\[6pt]
\infer[\mu R]{ \bar{x} : \omega \vdash  Ry .\mu_t; P_{y} :: (y:t)}{  \bar{x} : \omega \vdash P_{y} :: (y :A) & t=_{\mu}^i A } & \infer[\mu L]{x : t \vdash  \mathbf{case}\, Lx\ (\mu_{t} \Ra Q_{x }):: ( \bar{y} : \gamma)}{x : A \vdash Q_{x } :: ( \bar{y} : \gamma)  &  t=_{\mu}^i A }\\[6pt]
\infer[\nu R]{\bar{x} : \omega \vdash  \mathbf{case}\, Ry  \ (\nu_t \Ra P_{y }) :: (y : t)}{\bar{x} : \omega \vdash P_{y} :: (y : A)  & t=_{\nu}^i A }& \infer[\nu L]{x : t \vdash  Lx .\nu_{t}; Q_{x }:: ( \bar{y} : \gamma)}{x : A \vdash Q_{x } :: ( \bar{y} : \gamma) & t=_{\nu}^i A }
\end{tabular}
\]\\
\[ \infer[\msc{Def}(X)]{\bar{x} : \omega \vdash \bar{y} \leftarrow X \leftarrow \bar{x} :: ( \bar{y} : \gamma)}{\bar{x} : \omega \vdash  P_{\bar{x},\bar{y}} :: ( \bar{y} : \gamma) & \bar{u}:\omega \vdash X=P_{\bar{u},\bar{w}} :: (\bar{w}:\gamma) \in V } \]
  \caption{Process assignment for subsingleton logic with fixed points (infinitary).}
  \label{fig:annotrule}
\end{figure}

The computational semantics is defined on configurations
\[
P_0 \mid_{\, x_1} \cdots \mid_{\, x_n} P_n
\]
where $\mid$ is associative and has unit $(\cdot)$ but is not commutative.  The following transitions can be applied anywhere in a configuration:

\newcommand{\semi}{\mathrel{;}}
\[\renewcommand{\arraystretch}{1}
  \begin{array}{lcll}
    P_x \mid_x (y \leftarrow x) \mid_y Q_y & \mapsto & P_z \mid_z Q_z & \mbox{($z$ fresh), forward} \\
    (x \leftarrow P_x \semi Q_x) & \mapsto & P_z \mid_z Q_z & \mbox{($z$ fresh), spawn} \\
    (Rx.k \semi P) \mid_x \mathbf{case}\, Lx\, (\ell \Rightarrow Q_\ell)_{\ell \in L} & \mapsto & P \mid_x Q_k & \mbox{send label $k\in L$ right} \\
    \mathbf{case}\, Rx\, (\ell \Rightarrow P_\ell)_{\ell \in L} \mid_x (Lx.k \semi Q) & \mapsto & P_k \mid_x Q & \mbox{send label $k\in L$ left} \\
    \mathbf{close}\, Rx \mid_x (\mathbf{wait}\, Lx \semi Q) & \mapsto & Q & \mbox{close channel right} \\
    (\mathbf{wait}\, Rx \semi P) \mid_x \mathbf{close}\, Lx & \mapsto & P & \mbox{close channel left} \\
    (Rx.\mu_t \semi P) \mid_x \mathbf{case}\, Lx\, (\mu_t \Rightarrow Q) & \mapsto & P \mid_x Q & \mbox{send $\mu_t$ unfolding message right} \\
    \mathbf{case}\, Rx\, (\nu_t \Rightarrow P) \mid_x (Lx.\nu_t \semi Q) & \mapsto & P \mid_x Q & \mbox{send $\nu_t$ unfolding message  left} \\
    \bar{y} \leftarrow X \leftarrow \bar{x}  & \mapsto &  P_{\bar{x},\bar{y}} & \mbox{where $\bar{u} : \omega \vdash X = P_{\bar{u},\bar{w}} :: (\bar{w} : \gamma) \in V$}
  \end{array}
\]
The forward rule removes process $y \leftarrow x$ from the configuration and replaces both channels $x$ and $y$ in the rest of the configuration with a fresh channel $z$. The rule for $x \leftarrow P_x\, ; Q_x$  spawns process $P_z$ and continues as $Q_z$. To ensure uniqueness of channels, we need $z$ to be a fresh channel. For internal choice, $Rx.k;P$ sends label $k$ along channel $x$ to the process on its right and continues as $P$. The process on the right, $\mathbf{case} \, Lx\, (\ell \Rightarrow Q_{\ell})$, receives the label $k$ sent from the left along channel $x$, and chooses the branch with label $k$  to continue with $Q_k$. The last transition rule unfolds the definition of a process variable $X$ while instantiating the left and right channels $\bar{u}$ and $\bar{w}$ in the process definition with proper channel names, $\bar{x}$ and $\bar{y}$ respectively.

\section{Ensuring communication and local validity}
\label{session}

In this section we motivate our algorithm as an effectively decidable compositional and local criterion which ensures that a program always terminates either in an empty configuration or one attempting to communicate along external channels. By defining type variables in the signature and process variables in the program, we can generate (mutually) recursive processes which correspond to circular pre-proofs in the sequent calculus. In Examples~\ref{begen} and~\ref{block}, we provide such recursive processes along with explanations of their computational steps and their corresponding derivations.
\begin{example}\label{begen}
Take the signature
\[\Sigma_1:=  \mathsf{nat}=^{1}_{\mu} \oplus\{ \mathit{z}:\mathsf{1}, \mathit{s}:\mathsf{nat}\}.\]
 We define a process
 \[\begin{aligned}
 \cdot \vdash \mathtt{Loop} :: (y:\mathsf{nat}) \\
 \end{aligned},\]
 where $\mathtt{Loop}$ is defined as 
 \begin{align}
 y \leftarrow \mathtt{Loop} \leftarrow \cdot = \  & Ry.\mu_{nat};  &&\phantom{lo space} \%\ \textit{send}\ \mu_{nat}\  \textit{to right} \tag{i}\\
& \phantom{low}  Ry.s;  && \phantom{lo space} \%\ \textit{send label}\ \mathit{s} \  \textit{to right} \tag{ii}\\
& \phantom{low s} y \leftarrow \mathtt{Loop} \leftarrow \cdot && \phantom{lo space} {\color{red} \%\ \textit{recursive call}}\ \tag{iii}
\end{align}
$\mathcal{P}_1:=\langle \{\mathtt{Loop}\}, \mathtt{Loop} \rangle$ forms a program over the signature $\Sigma_1$.  It (i) sends a {\it positive} fixed point unfolding message to the right, (ii) sends the label $\ \mathit{s}$, as another message corresponding to $\mathit{successor}$, to the right, (iii) calls itself and loops back to (i). The program runs forever, sending {\it successor} labels to the right, without receiving any fixed point unfolding messages from the left or the right. We can obtain the following infinite derivation in the system of Figure~\ref{fig:rule} via the Curry-Howard correspondence of the unique typing derivation of process $\mathtt{Loop}$:
\begin{center}
\begin{tikzpicture}
\node[](2){$\infer[\mu R]{\cdot \vdash \mathsf{nat}}{\infer[\oplus R_s]{\cdot \vdash \oplus\{ \mathit{z}:\mathsf{1}, \mathit{s}:\mathsf{nat}\}}{\cdot \vdash \mathsf{nat}}}$};
\draw[->] (-1,0.7).. controls (-3,.6) and (-3,-.7) .. (-1,-.7);
\end{tikzpicture} 
\end{center} 
\end{example}

\begin{example}\label{block}
Define process
 \[\begin{aligned}
 x:\mathsf{nat} \vdash \mathtt{Block} :: (y:\mathsf{1}) \\
 \end{aligned}\]
over the signature $\Sigma_1$ as 
 \begin{align}
  & y \leftarrow \mathtt{Block} \leftarrow x= && \notag \\ & \phantom{small} \mathbf{case}\, Lx\ (\mu_{nat} \Ra \tag{i} && \%
  \ \textit{receive}\ \mu_{nat}\ \textit{from left} \\ & \phantom{small spac}\mathbf{case}\,Lx && \%\ \textit{receive a label} \ \textit{from left} \tag{ii}\\ & \phantom{small space ti}  ( \ z \Rightarrow \mathbf{wait}\, Lx; && \% \ \textit{wait for}\ x\ \textit{to close}  \tag{ii-a}\\
   & \phantom{small space times two} \mathbf{close}\, Ry && \%\ \textit{close}\ y \notag\\
 & \phantom{small space ti} \mid s \Rightarrow  y \leftarrow \mathtt{Block} \leftarrow x)) && {\color{red} \% \  \textit{recursive call \tag{ii-b}}} 
\end{align}
$\mathcal{P}_2:=\langle \{\mathtt{Block}\}, \mathtt{Block} \rangle$ forms a program over the signature $\Sigma_1$:\\
(i) $\mathtt{Block} $ {\it \bf waits}, until it receives a {\it positive} fixed point unfolding message from the left, (ii) waits for another message from the left to determine the path it will continue with:\\ \indent (a) If the message is a $\ \mathit{z}\ $ label, (ii-a) the program waits until a closing message is received from the left. Upon receiving that message, it closes the left and then the right channel. \\ \indent (b) If the message is an $\ \mathit{s} \ $ label, (ii-b) the program calls itself and loops back to (i).\\
Process $\mathtt{Block}$ corresponds to the following infinite derivation:
\begin{center}
\begin{tikzpicture}
\node[](2){$\infer[\mu L]{\mathsf{nat} \vdash \mathsf{1}}{\infer[\oplus L]{ \oplus\{ \mathit{z}:\mathsf{1}, \mathit{s}:\mathsf{nat}\} \vdash \mathsf{1}}{\infer[\mathsf{1} L]{\mathsf{1} \vdash \mathsf{1}}{\infer[\mathsf{1}R]{\cdot \vdash \mathsf{1}}{}}& \mathsf{nat} \vdash \mathsf{1} }}$};
\draw[->] (1.5,0.15).. controls (3,0.15) and (3,-1.2) .. (1,-1.3);
\end{tikzpicture}
\end{center}
\end{example}
Derivations corresponding to both of these programs are cut-free. Also no internal loop takes place during their computation, in the sense that they both communicate with their left or right channels after finite number of steps. For process $\mathtt{Loop}$ this communication is restricted to sending infinitely many unfolding and successor messages to the right. Process $\mathtt{Block}$, on the other hand, receives the same type of messages after finite number of steps as long as they are provided by a process on its left. Composing these two processes as in $x \leftarrow \mathtt{Loop} \leftarrow \cdot \mid_x y \leftarrow \mathtt{Block} \leftarrow x$ results in an internal loop: process $\mathtt{Loop}$ keeps providing unfolding and successor messages for process $\mathtt{Block}$ so that they both can continue the computation and call themselves recursively. Because of this internal loop, the composition is not acceptable: it never communicates with its left (empty channel) or right (channel $y$). The infinite derivation corresponding to the composition $x \leftarrow \mathtt{Loop} \leftarrow \cdot \mid_x y \leftarrow \mathtt{Block} \leftarrow x$ therefore should be rejected as invalid:
\begin{center}
\begin{tikzpicture}
\draw[->] (-3.5,0.5).. controls (-4.7,0.5) and (-4.7,-1) .. (-3.5,-0.95);
\node[](2){$\infer[\msc{Cut}_{nat}]{\cdot \vdash \mathsf{1}}{\infer[\mu R]{\cdot \vdash \mathsf{nat}}{\infer[\oplus R_s]{\cdot \vdash  \oplus\{ \mathit{z}:\mathsf{1}, \mathit{s}:\mathsf{nat}\}}{\cdot \vdash \mathsf{nat}}} & \infer[\mu L]{\mathsf{nat} \vdash \mathsf{1}}{\infer[\oplus L]{ \oplus\{ \mathit{z}:\mathsf{1}, \mathit{s}:\mathsf{nat}\} \vdash \mathsf{1}}{\infer[\mathsf{1} L]{\mathsf{1} \vdash \mathsf{1}}{\infer[\mathsf{1}R]{\cdot \vdash \mathsf{1}}{}}& \mathsf{nat} \vdash \mathsf{1} }}}$};
\draw[->] (3.5,0.5).. controls (4.7,0.5) and (4.7,-1) .. (3.2,-0.9);
\end{tikzpicture}
\end{center}
The cut elimination algorithm introduced by Fortier and Santocanale uses a reduction function $\textsc{Treat}$ that may never halt. They proved that for derivations satisfying the guard condition $\textsc{Treat}$ is locally terminating since it always halts on guarded proofs~\cite{Fortier13csl}. The above derivation is an example of one that does not satisfy the FS guard condition and the cut elimination algorithm does not locally terminate on it.

The {\it progress property} for a configuration of processes ensures that during its computation it either: (i) takes a step; (ii) is empty; or (iii) communicates along its left or right channel. Without (mutually) recursive processes, this property is enough to make sure that computation never gets stuck. Having (mutually) recursive processes and fixed points, however, this property is not strong enough to restrict internal loops. The composition $x \leftarrow \mathtt{Loop} \leftarrow \cdot \mid_x y \leftarrow \mathtt{Block} \leftarrow x$, for example, satisfies the progress property but it never interacts with any other external process. We introduce a stronger form of the progress property, in the sense that it requires one of the conditions (ii) or (iii) to hold after a finite number of computation steps.

Like cut elimination, strong progress is not compositional. Processes $\mathtt{Loop}$ and $\mathtt{Block}$ both satisfy the strong progress property but their composition $x \leftarrow \mathtt{Loop} \leftarrow \cdot \mid_x y \leftarrow \mathtt{Block} \leftarrow x$ does not. We will show in Section~\ref{semantics} that FS validity implies strong progress. But, in contrast to strong progress, FS validity is compositional in the sense that composition of two disjoint valid programs is also valid. 
However, the FS guard condition is not local. Locality is particularly important from the programming point of view.
It is the combination of two properties that are pervasive and often implicit in the study of programming languages. First, the algorithm is syntax-directed, following the structure of the program and second, it checks each process definition separately, requiring only the signature and the types of other processes but not their definition. One advantage is indeed asymptotic complexity, and, furthermore, a practically very efficient implementation. In Remark~\ref{complexity} we show that the time complexity of our validity algorithm is linear in the total input, which consists of the signature and the process definitions.
Another is precision of error messages: locality implies that there is an exact program location where the condition is violated.  Validity is a complex property, so the value of precise error messages cannot be overestimated.  The final advantage is modularity: all we care about a process is its interface, not its definition, which means we can revise definitions individually without breaking validity of the rest of the program as long as we respect their interface. Our goal is to construct a locally checkable validity condition that accepts (a subset of) programs satisfying strong progress and is compositional.

In functional programming languages a program is called {\it terminating} if it reduces to a value in a finite number of steps, and is called {\it productive} if every piece of the output is generated in a finite number of steps (even if the program potentially runs forever). As in the current work, the theoretical underpinnings for terminating and productive programs are also least and greatest fixed points, respectively, but due to the functional nature of computation they take a different and less symmetric form than here (see, for example, \cite{Birkedal13lics,Grathwohl16phd}).

In our system of session types, least and greatest fixed points correspond to defined type variables with positive and negative polarity, respectively, and their behaviors are quite symmetric:  As in Definition~\ref{process}, an unfolding message $\mu_t$ for a type variable $t$ with positive polarity is received from the left and sent to the right, while for a variable $t$ with negative polarity, the unfolding message $\nu_t$ is received from the right and sent to the left. Going back to Examples~\ref{begen} and~\ref{block}, process $\mathtt{Loop}$ seems less acceptable than process $\mathtt{Block}$: process $\mathtt{Loop}$ does not receive any least or greatest fixed point unfolding messages. It is neither a terminating nor a productive process. We want our algorithm to accept process $\mathtt{Block}$ rather than $\mathtt{Loop}$, since it cannot accept both. This motivates a definition of reactivity on session-typed processes.

A  program defined over a signature $\Sigma$ is {\it reactive to the left} if it only continues forever if for a positive fixed point $t\in \Sigma$ with priority $i$ it receives a fixed point unfolding message $\mu_t$ from the left infinitely often.
A program is {\it reactive to the right} if it only continues forever if for a negative fixed point $t\in \Sigma$ with priority $i$ it receives a fixed point unfolding message $\nu_t$ from the right infinitely often.

A program is called {\it reactive} if it is either {\it reactive to the right} or {\it to the left}.
By this definition, process $\mathtt{Block}$ is reactive while process $\mathtt{Loop}$ is not. 
Although reactivity is not local we use it as a motivation behind our algorithm. We construct our local validity condition one step at a time. In each step, we expand the condition to accept one more family of interesting reactive programs, provided that we can check the condition locally. We first establish a local algorithm for programs with only direct recursion. We expand the algorithm further to support mutual recursions as well. Then we examine a subtlety regarding the cut rule to accept more programs locally. The reader may skip to Section~\ref{alg} which provides our complete finitary algorithm. Later, in Sections~\ref{Guard} and~\ref{semantics} we prove that our algorithm ensures the FS guard condition and strong progress.

Priorities of type variables in a signature are central to ensure that a process defined based on them satisfies strong progress. Throughout the paper we assume that the priorities are assigned (by a programmer) based on the intuition of why strong progress holds. We conclude this section with an example of a reactive process $\mathtt{Copy}$. This process, similar to $\mathtt{Block}$, receives a natural number from the left but instead of consuming it, sends it over to the right along a channel of type $\mathsf{nat}$.
\begin{example}\label{prex}
With signature $\Sigma_1:=  \mathsf{nat}=^{1}_{\mu} \oplus\{ \mathit{z}:\mathsf{1}, \mathit{s}:\mathsf{nat}\}$
we define process $\mathtt{Copy}$, $x:\mathsf{nat} \vdash \mathtt{Copy} :: (y:\mathsf{nat})$, as
 \begin{align}
  & y \leftarrow \mathtt{Copy} \leftarrow x= && \notag \\ & \phantom{small} \mathbf{case}\, Lx\ (\mu_{nat} \Ra \tag{i} && \%
  \ \textit{receive}\ \mu_{nat}\ \textit{from left} \\ & \phantom{small spac}\mathbf{case}\, Lx && \%\ \textit{receive a label} \ \textit{from left} \tag{ii}\\ & \phantom{small space tim}  (\ z\Ra Ry.\mu_{nat}; && \%\ \textit{send}\ \mu_{nat}\ \textit{ to right}  \tag{ii-a}\\  & \phantom{small space times two plu}  Ry.z; && \%\ \textit{send label}\ \mathit{z} \ \textit{ to right}  \notag\\ & \phantom{small space times two plus}  \mathbf{wait}\, Lx; && \% \ \textit{wait for}\ x\ \textit{to close} \notag \\
   & \phantom{small space times two plus on} \mathbf{close}\, Ry && \%\ \textit{close}\ y \notag\\
 & \phantom{small space tix} \mid s \Rightarrow Ry.\mu_{nat};  && \% \ \textit{send}\ \mu_{nat}\  to \ \textit{right} \tag{ii-b}\\ & \phantom{small space times two plu} Ry.s; && \% \ \textit{send label}\ \mathit{s}\ \textit{to right} \notag \\
 & \phantom{small space times two plus}  y \leftarrow \mathtt{Copy} \leftarrow x)) && {\color{red} \% \  \textit{recursive call}} \notag
\end{align}
This is an example of a recursive process, and $\mathcal{P}_3:=\langle \{\mathtt{Copy}\}, \mathtt{Copy} \rangle$ forms a {\it left reactive} program over the signature $\Sigma_1$:\\
(i) It waits until it receives a {\it positive} fixed point unfolding message from the left, (ii) waits for another message from the left to determine the path it will continue with:\\ \indent (a) If the message is a $\ \mathit{z}\ $ label, (ii-a) the program sends a {\it positive} fixed point unfolding message to the right, followed by the label $\ \mathit{z}$, and then waits until a closing message is received from the left. Upon receiving that message, it closes the right channel. \\ \indent (b) If the message is an $\ \mathit{s} \ $ label, (ii-b) the program sends a {\it positive} fixed point unfolding message to the right, followed by the label $\ \mathit{s}\ $, and then calls itself and loops back to (i).\\
The computational content of $\mathtt{Copy}$ is to simply copy a natural number given from the left to the right. Process $\mathtt{Copy}$ does not involve spawning (its underlying derivation is cut-free) and satisfies the strong progress property. This property is preserved when composed with $\mathtt{Block}$ as $y\leftarrow \mathtt{Copy} \leftarrow x\mid_y z \leftarrow \mathtt{Block} \leftarrow y$.
\end{example}

\section{Local Validity Algorithm: Naive Version}
\label{alg1}
In this section we develop a first naive version of our local validity algorithm using Examples~\ref{bits}-\ref{cobits}. 
\begin{example}\label{bits}
Let the signature be \[\Sigma_2:=\\  \mathsf{bits}=^{1}_{\mu} \oplus\{ \mathit{b0}:\mathsf{bits},\ \mathit{b1}:\mathsf{bits}\}\]
and define the process $\mathtt{BitNegate}$
 \[\begin{aligned}
 x:\mathsf{bits} \vdash \mathtt{BitNegate} :: (y:\mathsf{bits}) \\
 \end{aligned}\]
with
 \begin{align}
  & y \leftarrow \mathtt{BitNegate} \leftarrow x= && \notag \\ & \phantom{small} \mathbf{case}\, Lx\ (\mu_{bits} \Ra \tag{i} && \%
  \ \textit{receive}\ \mu_{bits}\ \textit{from left} \\ & \phantom{small spac}\mathbf{case}\, Lx && \%\ \textit{receive a label} \ \textit{from left} \tag{ii}\\ & \phantom{small space ti}(\ \mathit{b0}\Ra Ry.\mu_{bits}; && \%\ \textit{send}\ \mu_{bits}\ \textit{ to right}  \tag{ii-a}\\  & \phantom{small space times two plu}  Ry.b1; && \%\ \textit{send label}\ \mathit{b1} \ \textit{ to right}  \notag\\ & \phantom{small space times two plus}  y\leftarrow \mathtt{BitNegate} \leftarrow x && {\color{red}\% \ \textit{recursive call} } \notag \\
 & \phantom{small space ti} \mid \mathit{b1} \Rightarrow Ry.\mu_{bits};  && \% \ \textit{send}\ \mu_{bits} to \ \textit{right} \tag{ii-b}\\ & \phantom{small space times two plu} Ry.b0; && \% \ \textit{send label}\ \mathit{b0} \textit{ to right} \notag \\
 & \phantom{small space times two plus}  y \leftarrow \mathtt{BitNegate} \leftarrow x)) && {\color{red} \% \  \textit{recursive call}} \notag
\end{align}
 $\mathcal{P}_4:=\langle \{\mathtt{BitNegate}\}, \mathtt{BitNegate} \rangle$ forms a {\it left reactive} program over the signature $\Sigma_2$ quite similar to $\mathtt{Copy}$.
 Computationally, $\mathtt{BitNegate}$ is a buffer with one bit capacity that receives a bit from the left and stores it until a process on its right asks for it. After that, the bit is negated and sent to the right and the buffer becomes free to receive another bit. 
\end{example}
\begin{example}\label{cobits}
Dual to Example~\ref{bits}, we can define $\mathtt{coBitNegate}$.
Let the signature be \[\Sigma_3:=\\  \mathsf{cobits}=^{1}_{\nu} \&\{ \mathit{b0}:\mathsf{cobits},\ \mathit{b1}:\mathsf{cobits}\}\]
with process
 \[\begin{aligned}
 x:\mathsf{cobits} \vdash \mathtt{coBitNegate} :: (y:\mathsf{cobits}) \\
 \end{aligned}\]
where $\mathtt{coBitNegate}$ is defined as
 \begin{align}
  & y \leftarrow \mathtt{coBitNegate} \leftarrow x= && \notag \\ & \phantom{small} \mathbf{case}\, Ry\ (\nu_{cobits} \Ra \tag{i} && \%
  \ \textit{receive}\ \nu_{cobits}\ \textit{from right} \\ & \phantom{small spac}\mathbf{case}\, Ry&& \%\ \textit{receive a label} \ \textit{from right} \tag{ii}\\ & \phantom{small space ti} (\ \mathit{b0}\Ra Lx.\nu_{cobits}; && \%\ \textit{send}\ \nu_{cobits}\ \textit{ to left}  \tag{ii-a}\\  & \phantom{small space times two plu}  Lx.b1; && \%\ \textit{send label}\ \mathit{b1} \ \textit{ to left}  \notag\\ & \phantom{small space times two plus}  y\leftarrow \mathtt{coBitNegate} \leftarrow x && {\color{red}\% \ \textit{recursive call} } \notag \\
 & \phantom{small space ti} \mid \mathit{b1} \Rightarrow Lx.\nu_{cobits};  && \% \ \textit{send}\ \nu_{cobits} \ \textit{to left} \tag{ii-b}\\ & \phantom{small space times two plu} Lx.b0; && \% \ \textit{send label}\ \mathit{b0}\ \textit{to left} \notag \\
 & \phantom{small space times two plus}  y \leftarrow \mathtt{coBitNegate} \leftarrow x)) && {\color{red} \% \  \textit{recursive call}} \notag
\end{align}
 $\mathcal{P}_5:=\langle \{\mathtt{coBitNegate}\}, \mathtt{coBitNegate} \rangle$ forms a {\it right reactive} program over the signature $\Sigma_3$.
 Computationally, $\mathtt{coBitNegate}$ is a buffer with one bit capacity. In contrast to $\mathtt{BitNegate}$ in Example~\ref{bits}, its types have negative polarity: it receives a bit from the {\it right}, and stores it until a process on its {\it left} asks for it. After that the bit is negated and sent to the {\it left} and the buffer becomes free to receive another bit. 
\end{example}

\begin{remark} \label{wait} The property that assures the reactivity of the previous examples lies in their step (i) in which the program {\it blocks} until an unfolding message is received, i.e., the program can only continue the computation if it {\it receives} a message at step (i), and even after receiving the message it can only take finitely many steps further before the computation ends or another unfolding message is needed.
\end{remark}

We first develop a naive version of our algorithm which captures the property explained in Remark~\ref{wait}: associate an initial integer value (say $0$) to each channel and define the basic step of our algorithm to be \emph{decreasing} the value associated to a channel \emph{by one} whenever it \emph{receives} a fixed point unfolding message. Also, for a reason that is explained later in Remark~\ref{*}, whenever a channel \emph{sends} a fixed point unfolding message its value is \emph{increased by one}. Then at each recursive call, the value of the left and right channels are compared to their initial value.

For instance, in Example~\ref{prex}, in step (i) where the process receives a $\mu_{nat}$ message via the left channel ($x$), the value associated with $x$ is decreased by one, while in steps (ii-a) and (ii-b) in which the process sends a $\mu_{nat}$ message via the right channel ($y$) the value associated with $y$ is increased by one:
 \begin{align*}
 \phantom{y \leftarrow \mathtt{Copy} \leftarrow x=} & \phantom{caseLx (\mu_{nat}}  & x &\phantom{space} y \\
 y \leftarrow \mathtt{Copy} \leftarrow x= & \phantom{caseLx (\mu_{nat}}  & 0 &\phantom{space} 0 \\
\phantom{X=}& \mathbf{case}\, Lx\, (\mu_{nat} \Ra &  {\color{red}-1} &\phantom{space} 0 \\ \phantom{X=} &  \phantom{caseLx (\mu_{nat}}    \mathbf{case}\, Lx\, ( \mathit{z} \Rightarrow Ry.\mu_{nat};  &  -1 &\phantom{space} {\color{red}1} \\
\phantom{X=}&   \phantom{caseLx (\mu_{nat} \Ra caseLx ( z \Rightarrow}  R.z; \mathbf{wait}\, Lx; \mathbf{close}\, Ry & -1 &\phantom{space} 1\\
\phantom{X=}&   \phantom{caseLx (\mu_{nat} \Ra cas}  \mathit{s} \Ra  Ry.\mu_{nat}; & -1 &\phantom{space} {\color{red}1}\\
\phantom{X=}&   \phantom{caseLx (\mu_{nat} \Ra caseLx  \mathit{s} \Ra  } Ry.s;y\leftarrow \mathtt{Copy} \leftarrow x)) & {\color{blue}-1} &\phantom{space} {\color{blue}1}
 \end{align*}
When the recursive call occurs, channel $x$ has the value ${\color{blue}-1} < 0$, meaning that at some point in the computation it received a positive fixed point unfolding message. We can simply compare the value of the list $[x,y]$ lexicographically at the beginning and just before the recursive call: ${\color{blue}[-1,1]}$ being less than $[0,0]$ exactly captures the property observed in Remark~\ref{wait} for the particular signature $\Sigma_1$. Note that by the definition of $\Sigma_1$, $y$ never receives a fixed point unfolding message, so its value never decreases, and $x$ never sends a fixed point unfolding message, thus its value never increases.

 The same criterion works for the program $\mathcal{P}_3$ over the signature $\Sigma_2$ defined in Example~\ref{bits}, since $\Sigma_2$ also contains only one positive fixed point:
 \begin{align*}
   \phantom{y \leftarrow \mathtt{BitNegate} \leftarrow} & \phantom{caseLx (\mu_{nat}}  & x &\phantom{space} y  \\
   y \leftarrow \mathtt{BitNegate} \leftarrow x= & \phantom{caseLx (\mu_{nat}}  & 0 &\phantom{space} 0  \\ \phantom{X=} & \mathbf{case}\, Lx\ (\mu_{bits} \Ra  & {\color{red} -1} &\phantom{space} 0 \\ \phantom{X=}& \phantom{caseLx} \mathbf{case}\, Lx\ (b0 \Rightarrow Ry.\mu_{bits};\  & -1 & \phantom{space} {\color{red} 1} \\ \phantom{X=}&  \phantom{caseLx (\mu_{nat} \Rightarrow Ry.} Ry.b1; y\leftarrow \mathtt{BitNegate} \leftarrow x & {\color{blue}-1} & \phantom{space} {\color{blue} 1}\\ 
 \phantom{X=} & \phantom{caseLx (\mu_{nat} \Rightarrow (}  \mathit{b1} \Rightarrow Ry.\mu_{bits};  & -1 & \phantom{space} {\color{red} 1} \\ \phantom{X=} &   \phantom{caseLx (\mu_{nat} \Rightarrow Ry.} Ry.b0;  y \leftarrow \mathtt{BitNegate} \leftarrow x)) &   {\color{blue} -1} & \phantom{space}{\color{blue} 1}  
\end{align*}
 At both recursive calls the value of the list $[x,y]$ is less than $[0,0]$: ${\color{blue}[-1,1]} < [0,0]$.
 
However, for a program defined on a signature with a negative polarity such as the one defined in Example~\ref{cobits}, this condition does not work: 
  \begin{flalign*}
   \phantom{y \leftarrow \mathtt{coBitNegate} \leftarrow} & \phantom{caseRy (\nu_{nat}}  & x &\phantom{space} y  \\
   y \leftarrow \mathtt{coBitNegate} \leftarrow x= & \phantom{caseRy (\mu_{nat}}  & 0 &\phantom{space} 0  \\ \phantom{X} & \mathbf{case}\, Ry\ (\nu_{cobits} \Ra  & {0} &\phantom{spac} {\color{red} -1} \\ \phantom{X=}& \phantom{caseLx} \mathbf{case}\, Ry\ (b0 \Rightarrow Lx.\nu_{cobits};\  & {\color{red} 1} &\phantom{spac} {-1} \\ \phantom{X=}&  \phantom{caseLx (\mu_{nat} \Rightarrow Ry.} Lx.b1; y\leftarrow \mathtt{coBitNegate} \leftarrow x & {\color{blue} 1} &\phantom{spac} {\color{blue} -1}\\ 
 \phantom{X=} & \phantom{caseLx (\mu_{nat} \Rightarrow (}  \mathit{b1} \Rightarrow Lx.\nu_{cobits};  & {\color{red}1} &\phantom{spac} {-1} \\ \phantom{X=} &   \phantom{caseLx (\mu_{nat} \Rightarrow Ry.} Lx.b0;  y \leftarrow \mathtt{coBitNegate} \leftarrow x)) &   {\color{blue} 1} &\phantom{spac}{\color{blue} -1}  
\end{flalign*}
By the definition of $\Sigma_3$,
$y$ only receives unfolding fixed point messages, so its value only decreases. On the other hand, $x$ cannot receive an unfolding fixed point from the left and thus its value never decreases. In this case the property in Remark~\ref{wait} is captured by comparing the initial value of the list $[y,x]$, instead of $[x,y]$, with its value just before the recursive call: ${\color{blue}[-1,1]}<[0,0]$.

For a signature with only a single recursive type we can form a list by looking at the polarity of its type such that the value of the channel that receives the unfolding message comes first, and the value of the other one comes second. Our algorithm ensures that the value of the list right before a recursive call is lexicographically less than the initial value of the list. In this section, we implemented the algorithm by counting the number of fixed point unfolding messages sent or received along each channel. However, keeping track of the exact number of unfolding messages is too much information and unnecessary. At the end of the next section, we introduce an alternative implementation that establishes the relation between channels after sending or receiving a fixed point unfolding message without tracking the exact number of the messages.

\section{Priorities in the Local Validity Algorithm}
\label{priority}

The property explained in Remark~\ref{wait} of previous section is not strict enough, particularly when the signature has more than one recursive type. In that case not all programs that are waiting for a fixed point unfolding message before a recursive call are reactive.

\begin{example}\label{begen1} Consider the signature \[\begin{aligned}
 \Sigma_4:=\ & \mathsf{ack}=^{1}_{\mu} \oplus\{ \mathit{ack}:\mathsf{astream}\},\\
& \mathsf{astream}=^{2}_{\nu} \& \{\mathit {head}: \mathsf{ack}, \ \ \mathit{tail}: \mathsf{astream}\},\\
& \mathsf{nat}=^{3}_{\mu}\oplus \{\mathit{z}:1,\ \  \mathit{s}: nat\}
 \end{aligned}\]
 $\mathsf{ack}$ is a type with {\it positive} polarity that, upon unfolding, describes a protocol requiring an {\it acknowledgment} message to be sent to the right (or be received from the left). $\mathsf{astream}$ is a type with {\it negative} polarity of a potentially infinite stream where its $\mathit{head}$ is always followed by an acknowledgement while $\mathit{tail}$ is not.
 
 $\mathcal{P}_6:=\langle \{ \mathtt{Ping}, \mathtt{Pong}, \mathtt{PingPong}\},  \mathtt{PingPong} \rangle$ forms a program over the signature $\Sigma_4$ with the typing of its processes
 \[
 \begin{array}{l}
 x:\mathsf{nat} \vdash  \mathtt{Ping} :: (w:\mathsf{astream}) \\
 w:\mathsf{astream} \vdash  \mathtt{Pong}:: (y:\mathsf{nat}) \\
 x:\mathsf{nat} \vdash  \mathtt{PingPong}:: (y:\mathsf{nat})
 \end{array}
 \]
 We define processes $\mathtt{Ping}$, $\mathtt{Pong}$, and $\mathtt{PingPong}$ over $\Sigma_4$ as:
 \begin{align}
  & y \leftarrow \mathtt{PingPong} \leftarrow x=&& \notag\\&  \phantom{smal} w \leftarrow  \mathtt{Ping} \leftarrow x ; \tag{i} && \% \  \textit{spawn process}\ \mathtt{Ping} \\& \phantom{small space}  y \leftarrow \mathtt{Pong} \leftarrow w && {\color{red}\% \ \textit{continue with a tail call}}\notag\\ \notag\\
 & y\leftarrow \mathtt{Pong} \leftarrow w=&& \notag\\ & \phantom{smal} Lw.\nu_{astream};&&  \% \  \textit{send}\ \mathit{\nu_{astream}} \ \textit{to left} \tag{ii-$\mathtt{Pong}$}\\& \phantom{small s} Lw.\mathit{head};&&  \% \  \textit{send label}\ \mathit{head} \ \textit{to left} \tag{iii-$\mathtt{Pong}$}\\ & \phantom{small sp} \mathbf{case}\, Lw\ (\mu_{ack} \Ra && \% \  \textit{receive}\ \mathit{\mu_{ack}} \ \textit{from left} \tag{iv-$\mathtt{Pong}$}\\ &\phantom{small space } \mathbf{case}\, Lw\ ( && \% \  \textit{receive a label from left} \notag \\ & \phantom{small space times two}ack \Rightarrow Ry.\mu_{nat};&&\% \  \textit{send}\ \mathit{\mu_{nat}} \ \textit{to right} \notag\\ & \phantom{small space times two plu} Ry.s; && \% \  \textit{send label}\ \mathit{s} \ \textit{to right} \notag \\ & \phantom{small space times two plus o} y \leftarrow  \mathtt{Pong} \leftarrow w))&& {\color{red} \% \  \textit{recursive call}} \notag 
\end{align}
 \begin{align}
 & w \leftarrow \mathtt{Ping} \leftarrow x= && \notag \\ &  \phantom{smal} \mathbf{case}\, Rw\  (\nu_{astream} \Ra && \% \  \textit{receive } \mathit{\nu_{astream}} \textit{from right} \tag{ii-$\mathtt{Ping}$}\\ & \phantom{small sp} \mathbf{case}\, Rw\ ( &&\notag \% \  \textit{receive a label from right} \\ &\phantom{small space times} \mathit{head} \Rightarrow Rw.\mu_{ack};&&   \% \  \textit{send}\ \mu_{ack} \ \textit{to right} \tag{iii-$\mathtt{Ping}$}\\ & \phantom{small space times two plus on }Rw.\mathit{ack}; \notag && \% \  \textit{send label}\ \mathit{ack}\ \textit{to right}\\ & \phantom{small space times two plus one an } w\leftarrow \mathtt{Ping}\leftarrow x && {\color{red} \% \  \textit{recursive call}} \notag \\ 
 & \phantom{small space times } \mid \mathit{tail} \Rightarrow w \leftarrow  \mathtt{Ping}\leftarrow x)) \notag&& {\% \  \textit{recursive call}}
\end{align}

(i) Program $\mathcal{P}_6$ starting from $\mathtt{PingPong}$, spawns a new process $\mathtt{Ping}$ and continues as $\mathtt{Pong}$:\\ \indent (ii-$\mathtt{Pong}$) Process $\mathtt{Pong}$ sends an $\mathsf{astream}$ unfolding and then a $\mathit{head}$ message to the left, and then (iii-$\mathtt{Pong}$) {\it waits} for an acknowledgment, i.e., $\mathit{ack}$, from the left.\\ \indent
(ii-$\mathtt{Ping}$) At the same time process $\mathtt{Ping}$ {\it waits} for an $\mathsf{astream}$ fixed point unfolding message from the right, which becomes available after step (ii-$\mathtt{Pong}$). Upon receiving the message, it waits to receive either $\mathit{head}$ or $\mathit{tail}$ from the right, which is also available from (ii-$\mathtt{Pong}$) and is actually a $\mathit{head}$. So (iii-$\mathtt{Ping}$) it continues with the path corresponding to $\mathit{head}$, and acknowledges receipt of the previous messages by sending an unfolding messages and the label $\mathit{ack}$ to the right, and then it calls itself (ii-$\mathtt{Ping}$).\\
\indent (iv-$\mathtt{Pong}$) Process $\mathtt{Pong}$ now receives the two messages sent at (iii-$\mathtt{Ping}$) and thus can continue by sending a $\mathsf{nat}$ unfolding message and the label $\mathit{s}$ to the right, and finally calling itself (ii-$\mathtt{Pong}$).\\
Although both recursive processes $\mathtt{Ping}$ and $\mathtt{Pong}$ at some point {\it wait} for a fixed point unfolding message, this program runs infinitely without receiving any messages from the outside, and thus is not reactive.
\end{example}

The back-and-forth exchange of fixed point unfolding messages between two processes in the previous example can arise when at least two mutually recursive types with different polarities are in the signature. To avoid such non-reactive behavior, we need to incorporate priorities of the type variables into the validity checking algorithm and track both sending and receiving of the unfolding messages.

\begin{remark} \label{*}
 In Example~\ref{begen1}, for instance, waiting to {\it receive} an unfolding message $\nu_{astream}$ of priority 2 in line (ii-$\mathtt{Ping}$) is not enough to ensure validity of the recursive call because later in line (iii-$\mathtt{Ping}$) the process {\it sends} an unfolding message of a higher priority 1.
\end{remark}

To preclude such a call we form a list for each process. This list stores the information of the fixed point unfolding messages that the process received and sent before a recursive call for each type variable in their order of priority.

\begin{example}\label{c}
Consider the signature and program $\mathcal{P}_6$ as defined in Example~\ref{begen1}. For the process $x:\mathsf{nat} \vdash w \leftarrow \mathtt{Ping} \leftarrow x:: (w:\mathsf{astream})$ form the list
\[[\mathsf{ack}-\mathit{received}, \mathsf{ack}-\mathit{sent}, \mathsf{astream}-\mathit{received}, \mathsf{astream}-\mathit{sent}, \mathsf{nat}-\mathit{received}, \mathsf{nat}-\mathit{sent} ].\]
Types with positive polarity, i.e., $\mathsf{ack}$ and $\mathsf{nat}$, receive messages from the left channel ($x$) and send messages to the right channel ($w$), while those with negative polarity, i.e., $\mathsf{astream}$, receive from the right channel ($w$) and send to the left one ($x$). Thus, the above list can be rewritten as
\[[x_\mathsf{ack}, w_{\mathsf{ack}}, w_{\mathsf{astream}}, x_{\mathsf{astream}}, x_\mathsf{nat}, w_{\mathsf{nat}}].\]
To keep track of the sent/received messages, we start with $[0,0,0,0,0,0]$ as the value of the list, when the process $x:\mathsf{nat} \vdash \mathtt{Ping} :: (w:\mathsf{astream})$ is first spawned. Then, similar to the first version of our algorithm, on the steps in which the process {\it receives} a fixed point unfolding message, the value of the corresponding element of the list is {\it decreased by one}. And on the steps it {\it sends} a fixed point unfolding message, the corresponding value is {\it increased by one}:\\
 \begin{align*}
 & w \leftarrow \mathtt{Ping} \leftarrow x= & \phantom{caseRw}   [0,0 \ ,0\ ,0, 0, 0]\\
& \phantom{\mathtt{Ping}=} \mathbf{case}\, Rw\ (\nu_{astream} \Ra & [0,0,{\color{red}-1},0, 0 ,0] \\  &  \phantom{caseRw (\nu_{astream})} \mathbf{case}\, Rw\ ( \mathit{head} \Rightarrow Rw.\mu_{ack};  &  [0,{\color{red}1},-1,0, 0, 0]\\
&  \phantom{caseRw (\nu_{astream} \Ra caseRw ( head \ \ }  Rw.\mathit{ack}; w\leftarrow \mathtt{Ping} \leftarrow x & {\color{blue} [0,1,-1,0, 0, 0]}\\
&   \phantom{caseRw (\nu_{astream} \Ra casex} \mid\mathit{tail} \Ra w\leftarrow \mathtt{Ping} \leftarrow x))& [0,0,-1,0, 0, 0]
 \end{align*}
The two last lines are the values of the list on which process $\mathtt{Ping}$ calls itself recursively.
The validity condition as described in Remark~\ref{*} holds iff the value of the list at the time of the recursive call is less than the value the process started with, in lexicographical order.
    Here, for example, ${\color{blue} [0,1,-1,0]} \not < [0,0,0,0]$, and the validity condition does not hold for this recursive call. 
    
    We leave it to the reader to verify that no matter how we assign priorities of the type variables in $\Sigma_4$, our condition rejects $\mathtt{PingPong}$.
\end{example}

The following definition captures the idea of forming lists described above.  Rather than directly referring to type variables such as $\mathsf{ack}$ or $\mathsf{astream}$ we just refer to their priorities, since that is the relevant information.

\begin{definition}\label{list}
For a process \[ \bar{x}:\omega \vdash P :: (y:B),\] over the signature $\Sigma$, define $list(\bar{x},y)= [f_i]_{i \le n}$ such that
\begin{enumerate}
    \item $f_i= ( \bar{x}_i, y_i )$ if $\epsilon(i)= \mu$, and
    \item $f_i= (y_i, \bar{x}_i)$ if $\epsilon(i)=\nu,$
\end{enumerate}
where $n$ is the lowest priority in $\Sigma.$
\end{definition}

In the remainder of this section we use $n$ to denote the lowest priority in $\Sigma$ (which is numerically maximal).

\begin{example}\label{easy}
Consider the signature $\Sigma_1$ and program $\mathcal{P}_3:= \langle \{\mathtt{Copy}\}, \mathtt{Copy} \rangle$, from Example~\ref{prex}:\\$ \Sigma_1 :=  \mathsf{nat}=^{1}_{\mu} \oplus\{ \mathit{z}:\mathsf{1}, \mathit{s}:\mathsf{nat}\},$ and
\begin{align*}
 y \leftarrow \mathtt{Copy} \leftarrow x = \mathbf{case}\,Lx\ (\mu_{nat} \Ra \mathbf{case}\, Lx\ &(\ z\Ra Ry.\mu_{nat}; Ry.z; \mathbf{wait}\, Lx; \mathbf{close}\, Ry  \\
 & \mid s \Rightarrow Ry.\mu_{nat}; Ry.s; y \leftarrow \mathtt{Copy} \leftarrow x)) 
\end{align*}
By Definition~\ref{list}, for process $ x:\mathsf{nat} \vdash \mathtt{Copy} :: (y:\mathsf{nat})$, we have $n=1$, and $list(x,y)= [(x_1,y_1)]$ since $\epsilon(1)= \mu$.  Just as for the naive version of the algorithm, we can trace the value of $list(x,y)$:
 \begin{align*}
 y \leftarrow \mathtt{Copy} \leftarrow x= & \phantom{caseLx (\mu_{nat}}  & [0,0]\\
\phantom{X=}& \mathbf{case}\, Lx\ (\mu_{nat} \Ra & [{\color{red}-1},0] \\ \phantom{X=} &  \phantom{caseLx (\mu_{nat}} \mathbf{case}\, Lx\ ( \mathit{z} \Rightarrow Ry.\mu_{nat};  &  [-1,{\color{red}1}]\\
\phantom{X=}&   \phantom{caseLx (\mu_{nat} \Ra caseLx ( z \Rightarrow}  R.z; \mathbf{wait}\, Lx; \mathbf{close}\, Ry & { [-1,1]}\\
\phantom{X=}&   \phantom{caseLx (\mu_{nat} \Ra case} \mid \mathit{s} \Ra  Ry.\mu_{nat}; & [-1,{\color{red}1}]\\
\phantom{X=}&   \phantom{caseLx (\mu_{nat} \Ra caseLx  \mathit{s} \Ra  } Ry.s;y\leftarrow \mathtt{Copy} \leftarrow x & {\color{blue}[-1,1]}
 \end{align*}
 Here, ${\color{blue}[-1,1]} <[0,0]$ and the recursive call is classified as valid.
\end{example}

To capture the idea of {\it decreasing}/{\it increasing} the value of the elements of $list (\_,\_)$ by {\it one}, as depicted in Example~\ref{c} and Example~\ref{easy}, we assume that a channel transforms into a new generation of itself after sending or receiving a fixed point unfolding message.

\begin{example}\label{su}
Process $ x:\mathsf{nat} \vdash y\leftarrow \mathtt{Copy} \leftarrow x :: (y:\mathsf{nat})$ in Example~\ref{easy} starts its computation with the initial generation of its left and right channels: \[x^0:\mathsf{nat} \vdash y^0 \leftarrow \mathtt{Copy} \leftarrow x^0 :: (y^0:\mathsf{nat}).\] The channels evolve as the process sends or receives a fixed point unfolding message along them:
\begin{align*}
 y^0 \leftarrow \mathtt{Copy}\leftarrow x^0 & =  \phantom{caseL{x^0} (\mu_{nat}} &  \\
\phantom{X=}& \mathbf{case}\, L{x^0}\ (\mu_{nat} \Ra & {\color{red} x^0 \rightsquigarrow x^1} \\ \phantom{X=} &  \phantom{caseL {x^1} (\mu_{nat}} \mathbf{case}\, L{x^1}\ (\mathit{z} \Rightarrow R{y^0}.\mu_{nat};& {\color{red} y^0 \rightsquigarrow y^1} \\
\phantom{X=}&   \phantom{caseL{x^1} (\mu_{nat} \Ra caseL{x^1} ( z \Rightarrow}  R{y^1}.z; \mathbf{wait}\, L{y^1}; \mathbf{close}\, R{x^1}& \\
\phantom{X=}&   \phantom{caseL{x^0} (\mu_{nat} \Ra case}  \mid \mathit{s} \Ra  R{y^0}.\mu_{nat}; & {\color{red} y^0 \rightsquigarrow y^1}\\
\phantom{X=}&   \phantom{caseL{x^0} (\mu_{nat} \Ra caseL{x^1}  \mathit{s} \Ra  } R{y^1}.s;y^1 \leftarrow \mathtt{Copy} \leftarrow x^1))
 \end{align*}
  On the last line the process \[x^1:\mathsf{nat} \vdash y^1 \leftarrow \mathtt{Copy} \leftarrow x^1 :: (y^1:\mathsf{nat})\] is called recursively with a new generation of variables.
 \end{example}
 
In the inference rules introduced in Section~\ref{rules}, instead of recording the value of each element of $list(\ \_,\_)$ as we did in Example~\ref{c} and Example~\ref{easy}, we introduce $\Omega$ to track the relation between different generations of a channel indexed by their priority of types.

 \begin{remark}
 Generally speaking, $x^{\alpha+1}_{i}<x^{\alpha}_{i}$ is added to $\Omega$ when $x^\alpha$ {\it receives} a fixed point unfolding message for a type with priority $i$ and transforms to $x^{\alpha+1}$.  This corresponds to the {\it decrease by one} in the previous examples.
 
 If $x^\alpha$ {\it sends} a fixed point unfolding message for a type with priority $i$ is \emph{sent} on $x^\alpha$, which then evolves to $x^{\alpha+1}$, $x^{\alpha}_i$ and $x^{\alpha+1}_i$ are considered to be incomparable in $\Omega$. This corresponds to {\it increase by one} in the previous examples, since for the sake of lexicographically comparing the value of $list(\ \_,\_)$ at the {\it first call} of a process to its value just before a {\it recursive call}, there is no difference whether $x^{\alpha+1}$ is greater than $x^\alpha$ or incomparable to it.
 
When $x^\alpha$ receives/sends a fixed point unfolding message of a type with priority $i$ and transforms to $x^{\alpha+1}$, for any type with priority $j\neq i$, the value of $x^{\alpha}_j$ and $x^{\alpha+1}_j$ must remain equal. In these steps, we add $x^{\alpha}_j = x^{\alpha+1}_j$ for $j\neq i$ to $\Omega$.
 \end{remark}
 
 A process in the formalization of the intuition above is therefore typed as \[ x^\alpha:A \vdash_{\Omega} P :: (y^{\beta}:B),\] 
where $x^\alpha$ is the $\alpha$-th generation of channel $x$. The syntax and operational semantics of the processes with generational channels are the same as the corresponding definitions introduced in Section~\ref{operat}; we simply ignore generations over the channels to match processes with the previous definitions. We enforce the assumption that channel $x^\alpha$ transforms to its next generation $x^{\alpha+1}$ upon sending/receiving a fixed point unfolding message in the typing rules of Section~\ref{rules}.

The relation between the channels indexed by their priority of types is built step by step in $\Omega$ and represented by $\le$. The reflexive transitive closure of $\Omega$ forms a partial order $\le_{\Omega}$. We extend $\le_{\Omega}$ to the {\it list} of channels indexed by the priority of their types considered lexicographically. We may omit subscript $\Omega$ from $\le_{\Omega}$ whenever it is clear from the context. In the next examples, we present the set of relations $\Omega$ in the rightmost column.

The reader may refer to Figure~\ref{fig:stp-order} for the typing rules of processes enriched with generational channels and the set Omega. The cut rule presented in Figure~\ref{fig:stp-order} will be explained in detail later in Section~\ref{modcut}. For now, the reader may consider the following simplified version of the cut rule instead:

\[\infer[\msc{Cut}^{w}]{ \bar{x}^{\alpha}: \omega \vdash_{\Omega}  (w \leftarrow P_{w} ; Q_{w}) :: (y^{\beta}: C)}{ \bar{x}^{\alpha}:  \omega \vdash_{\Omega} P_{w^0} ::(w^0:A) & w^0: A \vdash_{\Omega} Q_{w^0} :: (y^{\beta}: C)}\]
\section{Mutual Recursion in the Local Validity Condition}
\label{mutual}

In examples of previous sections, the recursive calls were not \emph{mutual}. In the general case, a process may call any other process variable in the program, and this call can be mutually recursive. In this section, we incorporate mutual recursive calls into our algorithm.

\begin{example}\label{simplem}
Recall signature $\Sigma_4$ from Example~\ref{begen1}
\[
\begin{aligned}
 \Sigma_4:=\  & \mathsf{ack}=^{1}_{\mu} \oplus\{ \mathit{ack}:\mathsf{astream}\},\\
& \mathsf{astream}=^{2}_{\nu} \& \{\mathit {head}: \mathsf{ack}, \ \ \mathit{tail}: \mathsf{astream}\},\\
& \mathsf{nat}=^{3}_{\mu}\oplus \{\mathit{z}:1,\ \  \mathit{s}: nat\}
 \end{aligned}
\]
Define program $\mathcal{P}_7= \langle \{\mathtt{Idle}, \mathtt{Producer}\}, \mathtt{Producer} \rangle$, where
\[
\begin{array}{l}
z:\mathsf{ack} \vdash w \leftarrow \mathtt{Idle}\leftarrow z :: (w:\mathsf{nat})\\
x:\mathsf{astream} \vdash y \leftarrow \mathtt{Producer} \leftarrow x :: (y:\mathsf{nat}), 
\end{array}
\]
and processes $\mathtt{Idle}$ (or simply $\mathtt{I}$)  and $\mathtt{Producer}$ (or simply $\mathtt{P}$ ) are defined as:
\[
\begin{array}{l}
w \leftarrow \mathtt{I}\leftarrow z  = \mathbf{case}\, Lz\ (\mu_{ack} \Ra \mathbf{case}\, L{z}\ ( \mathit{ack}\Ra Rw.\mu_{nat}; Rw.\mathit{s}; w \leftarrow \mathtt{P}\leftarrow z)) 
\\[1ex]
  y \leftarrow \mathtt{P} \leftarrow x  = Lx.\nu_{astream}; Lx. \mathit{head}; y \leftarrow \mathtt{I}\leftarrow x.
\end{array}
\]
We have $list(x,y)= [(x_1,y_1), (y_2,x_2), (x_3,y_3)]$ and $list(z,w)= [(z_1,w_1), (w_2,z_2), (z_3,w_3)]$ since $\epsilon(1)=\epsilon(3)= \mu$ and $\epsilon(2)= \nu$.

By analyzing the behavior of this program step by step, we see that it is a {\it reactive} program that counts the number of acknowledgements received from the left.
The program starts with the process $ x^0:\mathsf{astream} \vdash_{\emptyset} y^0 \leftarrow \mathtt{Producer} \leftarrow x^0 :: (y^0:\mathsf{nat}) $. It first sends one message to left to unfold the {\it negative} fixed point type, and its left channel evolves to a next generation. Then another message is sent to the left to request the $\mathit{head}$ of the stream and after that it calls  process $y^0 \leftarrow \mathtt{Idle} \leftarrow x^1$.
 \begin{align*}
 & y^0 \leftarrow \mathtt{Producer}\leftarrow x^0= & \phantom{Lx^0.\nu_{astream}; }    \phantom{L.\nu } [0,0,0,0,0,0]& \phantom{L.\nu }\\
&\phantom{sma} Lx^0.\nu_{astream};  & \phantom{L.\nu } [0,0,0,1,0,0] & \phantom{L.\nu } x^1_1=x^0_1, x^1_3=x^0_3 \\ &  \phantom{small sp} Lx^1.\mathit{head};y^0 \leftarrow \mathtt{Idle} \leftarrow x^1  &\phantom{L.\nu } {\color{blue} [0,0,0,1,0,0]}& 
 \end{align*}
 Process $x^1:\mathsf{ack} \vdash y^0 \leftarrow \mathtt{Idle} \leftarrow x^1 :: (y^0:\mathsf{nat}) $, then {\it waits} to receive an acknowledgment from the left via a {\it positive} fixed point unfolding message for $\mathsf{ack}$ and its left channel transforms into a new generation upon receiving it. Then it waits for the label $\mathit{ack}$ and, upon receiving it, sends one message to the right to unfold the {\it positive} fixed point $\mathsf{nat}$ (and this time the right channel evolves). Then it sends the label $\mathit{s}$ to the right and calls $y^1 \leftarrow \mathtt{Producer}\leftarrow x^2$ recursively:
 
\begin{small}
 \begin{align*} 
 & y^0 \leftarrow \mathtt{Idle}\leftarrow x^1= & \phantom{caseLx^1} {\color{blue}[0, 0,\ 0,\ 1,\ 0, 0]} &  \phantom{L}\\
 &\phantom{sma} \mathbf{case}\, Lx^1\ (\mu_{ack} \Ra & [-1,0,0,1,0,0] &  \phantom{L} x^2_1 < x^1_1, x^2_2=x^1_2, x^2_3=x^1_3\\  &  \phantom{small sp} \mathbf{case}\, Lx^2\ ( \mathit{ack} \Rightarrow Ry^0.\mu_{nat};  & [-1,0,0,1,0,1] &  \phantom{Lx^0.\nu} y^1_1=y^0_1, y^1_2=y^0_2\\
&   \phantom{small space  }  Ry^1.\mathit{s}; y^1 \leftarrow \mathtt{Producer} \leftarrow x^2)) & {\color{red} [-1,0,0,1,0,1]}  &\phantom{L.}\\
\end{align*}
\end{small}%
Observe that the actual recursive call for $\mathtt{Producer}$ occurs at the last line (in red) above, where $\mathtt{Producer}$ eventually calls itself. At that point the value of $list(x^2,y^1)$ is recorded as ${\color{red} [-1,0,0,1,0,1]}$, which is less than the value of $list(x^0,y^0)$ when $\mathtt{Producer}$ was called for the first time:  \[{\color{red} [-1,0,0,1,0,1]} < {[0,0,0,0,0,0]}.\]
The same observation can be made by considering the relations introduced in the last column \[{\color{red}list(x^2,y^1)}={\color{red} [(x^2_1,y^1_1), (y^1_2,x^2_2), (x^2_3,y^1_3)]}< [(x^0_1,y^0_1), (y^0_2,x^0_2), (x^0_3,y^0_3)]= list(x^0,y^0)\] since  $x^2_1 < x^1_1= x^0_1$.
This recursive call is valid regardless of the fact that ${\color{blue} [0,0,0,1,0,0]} \not < {[0,0,0,0,0,0]},$ i.e.
\[{\color{blue}list(x^1,y^0)}={\color{blue} [(x^1_1,y^0_1), (y^0_2,x^1_2), (x^1_3,y^0_3)]}\not < [(x^0_1,y^0_1), (y^0_2,x^0_2), (x^0_3,y^0_3)]= list(x^0,y^0)\] since  $x^1_1= x^0_1$ but $x^1_2$ is incomparable to $x^0_2$. Similarly, we can observe that the actual recursive call on $\mathtt{Idle}$, where $\mathtt{Idle}$ eventually calls itself, is valid. 

To account for this situation, we introduce an order on {\it process variables} and trace the last seen variable on the path leading to the recursive call. In this example we define $\mathtt{Idle}$ to be less than $\mathtt{Producer}$ at position $2$ ($\mathtt{I} \subset_{2} \mathtt{P}$), i.e.:
\begin{quote} We incorporate process variables $\mathtt{Producer}$ and $\mathtt{Idle}$ into the lexicographical order on $list(\_,\_)$ such that their values are placed exactly before the element in the list corresponding to the {\it sent} unfolding messages of the type with priority $2$.
\end{quote}
We now trace the ordering as follows:

\begin{small}
 \begin{align*} 
 & y^0 \leftarrow \mathtt{Producer}\leftarrow x^0= & \phantom{L }    [0,0,0,\mathtt{P},0,0,0] & \phantom{L.\nu } \\
& \phantom{small } Lx^0.\nu_{astream}; & \phantom{L } [0,0,0,\mathtt{P},1,0,0] & \phantom{Lx^0.\nu } x^1_1=x^0_1, x^1_3=x^0_3 \\ &  \phantom{Lx^0.\nu_{strea}} Lx^1.\mathit{head};y^0 \leftarrow \mathtt{Idle} \leftarrow x^1  & \phantom{L }  {\color{blue} [0,0,0,\mathtt{I},1,0,0]} & \phantom{L.\nu }\\ \\
 & y^0 \leftarrow \mathtt{Idle} \leftarrow x^1= & \phantom{L}  {\color{blue}[0,0,0,\mathtt{I},1,0,0]} & \phantom{L.\nu} \\
& \phantom{small} \mathbf{case} \, Lx^1 (\mu_{ack} \Ra & [-1,0,0,\mathtt{I},1,0,0] &  \phantom{L.\nu} x^2_1<x^1_1, x^2_2=x^1_2, x^2_3=x^1_3\\  &  \phantom{small spa }   \mathbf{case} \,Lx^2 ( \mathit{ack} \Rightarrow Ry^0.\mu_{nat};  & [-1,0,0,\mathtt{I},1,0,1] &  \phantom{L.\nu} y^1_1=y^0_1, y^1_2=y^0_2 \\
&   \phantom{caseeLx^1 ( label }  Ry^1.\mathit{s};  y^1 \leftarrow \mathtt{Producer} \leftarrow x^2  &  {\color{red} [-1,0,0,\mathtt{P},1,0,1]} & \\
\end{align*}
\end{small}%

\noindent ${\color{red} [-1,0,0,\mathtt{P},1,0,1]} < {[0,0,0,\mathtt{I},1,0,0]}$ and ${\color{blue} [0,0,0, \mathtt{I}, 1,0,0]} < {[0,0,0,\mathtt{P},0,0,0]}$ hold, and both mutually recursive calls are recognized to be valid, as they are, without a need to substitute process definitions.
\end{example}
However, not every relation over the process variables forms a partial order. For instance, having both $\mathtt{P} \subset_{2} \mathtt{I}$ and $\mathtt{I} \subset_{2} \mathtt{P}$ violates the antisymmetry condition. Introducing the position of process variables into $list(\_,\_)$ is also a delicate issue. For example, if we have both $\mathtt{I} \subset_{1} \mathtt{P}$ and $\mathtt{I} \subset_{2} \mathtt{P}$, it is not determined where to insert the value of $\mathtt{Producer}$ and $\mathtt{Idle}$ on the $list(\_,\_)$. Definition~\ref{subs} captures the idea of Example~\ref{simplem}. It defines the relation $\subseteq$, given that the programmer introduces a family of partial orders such that their domains partition the set of process variables $V$. We again assume that the programmer defines this family based on the intuition of why a program satisfies strong progress. Definition~\ref{subs} ensures that $\subseteq$ is a well-defined partial order and it is uniquely determined in which position of $list(\_,\_)$ the process variables shall be inserted. Definition~\ref{defmain} gives the lexicographic order on $list(\_,\_)$  augmented with the $\subseteq$ relation.

\begin{definition}\label{subs}
Consider a program $\mathcal{P}=\langle V,S \rangle$ defined over a signature $\Sigma$. Let $\{\subseteq_i\}_{0 \le i \le n}$ be a disjoint family of partial orders whose domains partition the set of process variables $V$, where (a) $X \cong_{i}Y$ iff $X \subseteq_{i} Y$ and $Y \subseteq_{i} X$, and (b) $X \subset_{i}Y$ iff $X\subseteq_{i}Y$ but $X \not \cong_{i} Y$. 

We define $\subseteq$ as  $\bigcup_{i \le n} \subseteq_i$, i.e. $F \subseteq G$ iff $F \subseteq_{i} G$ for some (unique) $i \le n$. It is straightforward to see that $\subseteq$ is a partial order over the set of process variables $V$. Moreover, we define
(c) $X \cong Y$ iff $X \cong_{i} Y$ for some (unique) $i$, and (d) $X \subset Y$ iff $X\subset_{i}Y$  for some (unique) $i$.
\end{definition}

To integrate the order on process variables ($\subset$) with the order $<$, we need a prefix of the list from Definition~\ref{list}. We give the following definition of  $list(x,y,j)$ to crop $list(x,y)$ exactly before the element corresponding to a {\it sent} fixed point unfolding message for types with priority $j$. 
\begin{definition}
For a process \[\bar{x}:A \vdash P :: y:B,\] over signature $\Sigma$,
 and $0 \le j\le n$, define $list(\bar{x},y,j)$, as a prefix of the list $list(\bar{x},y)=[v_i]_{i \le n}$ by
\begin{enumerate}
\item $[]$ if $i=0$,
    \item $[ [v_i]_{ i < j},\ (\bar{x}_j)]$ if $\epsilon(j)=\mu$,
    \item $[[v_i]_{ i < j},\ (y_j)]$ if $\epsilon(j)=\nu$.
\end{enumerate}
\end{definition}
We use these prefixes in the following definition.
\begin{definition} \label{defmain}
Using the orders $\subset$ and $\le$, we define a new combined order $(\subset, <)$ (used in the local validity condition in Section~\ref{alg}). \[F,list(\bar{x},y)\ (\subset, <)\  G,list(\bar{z},w)\ \]
iff
\begin{enumerate}
    \item If $F \subset G$, i.e., $F \subset_i G$ for a unique $i$, then $list(\bar{x},y,i) \le list(\bar{z},w, i)$, otherwise,
    \item if $F \cong G$ and $list(\bar{x},y) < list(\bar{z},w)$, otherwise
    \item$list(\bar{x},y,\min(i,j)) < list(\bar{z},w,\min(i,j))$, where $F$ is in the domain of $\subseteq_{i}$ and $G$ is in the domain of $\subseteq_{j}$.
 \end{enumerate}
By conditions of Definition~\ref{subs}, $(\subset, <)$ is an irreflexive and transitive relation and thus a strict partial order.
\end{definition}\label{algn}
\begin{example}
Consider the signature of Example~\ref{simplem}
\[\begin{aligned}
 \Sigma_4 :=\ & \mathsf{ack}=^{1}_{\mu} \oplus\{ \mathit{ack}:\mathsf{astream}\},\\
& \mathsf{astream}=^{2}_{\nu} \& \{\mathit {head}: \mathsf{ack}, \ \ \mathit{tail}: \mathsf{astream}\},\\
& \mathsf{nat}=^{3}_{\mu}\oplus \{\mathit{z}:1,\ \  \mathit{s}: nat\}
 \end{aligned}\]
\renewcommand{\arraystretch}{1}%
and program $\mathcal{P}_7:=\langle \{\mathtt{Idle}, \mathtt{Producer}\}, \mathtt{Producer} \rangle$ with the relation $\subseteq$ defined over process variables as $\mathtt{Idle} \subset_2 \mathtt{Producer}$, $\mathtt{Producer} \subseteq_2 \mathtt{Producer}$, and $\mathtt{Idle} \subseteq_2 \mathtt{Idle}$.
For process $x:\mathsf{astream} \vdash \mathtt{Producer} :: (y:\mathsf{nat})$:
\[
\begin{array}{l}
list(x,y)=[(x_1,y_1),(y_2,x_2), (x_3,y_3)],\\
list(x,y,3)=[(x_1,y_1),(y_2, x_2), (x_3)],\\
list(x,y,2)=[(x_1,y_1),(y_2)],\\
list(x,y,1)=[(x_1)],\ \mathit{and}\\
list(x,y,0)=[].
\end{array}
\]
To check the validity of the recursive calls in Example~\ref{simplem} we observe that
\begin{itemize}
    \item $\mathtt{Producer}, list(x^2,y^1) \,(\subset, <)\, \mathtt{Idle}, list(x^1,y^0)$ since $list(x^2,y^1,2)<list(x^1,y^0,2)$, and
       \item $\mathtt{Idle}, list(x^1,y^0) \,(\subset, <)\, \mathtt{Producer}, list(x^0,y^0)$ since $list(x^1,y^0,2)=list(x^0,y^0,2)$ and\\ ${\mathtt{Idle} \subset_2 \mathtt{Producer}}$.\qedhere
\end{itemize}
\end{example}
\section{A Modified Rule for Cut}
\label{modcut}

There is a subtle aspect of local validity that we have not discussed yet. We need to relate a fresh channel, created by spawning a new process, with the previously existing channels. Process $y^\alpha:A \vdash (x\leftarrow P_x; Q_x):: (z^\beta:B)$, for example, creates a fresh channel $w^0$, spawns process
\mbox{$P_{w^0}$} providing along channel $w^0$, and then continues as $Q_{w^0}$. For the sake of our algorithm, we need to identify the relation between $w^0$, $y^\alpha$, and $z^\beta$. Since $w^0$ is a fresh channel, a naive idea is to make $w^0$ incomparable to any other channel for any type variable $t\in \Sigma$. To represent this incomparability in our examples we write ``$\infty$'' for the value of the fresh channel. While sound, we will see in Example~\ref{cutchannel} that we can improve on this naive approach to cover more valid processes.

\begin{example}\label{cutchannel}
Define the signature
\[\begin{aligned}
 \Sigma_5:=\ & \mathsf{ctr}=^{1}_{\nu} \& \{\mathit {inc}: \mathsf{ctr}, \ \ \mathit{val}: \mathsf{bin}\},\\
 & \mathsf{bin}=^{2}_{\mu} \oplus\{ \mathit{b0}:\mathsf{bin}, \mathit{b1}:\mathsf{bin}, \mathit{\$}:\mathsf{1}\}.\\
 \end{aligned}
\]
which provides numbers in binary representation as well as an interface to a counter.
We explore the following program $\mathcal{P}_8= \langle \{\mathtt{BinSucc},\mathtt{Counter},\mathtt{NumBits},\mathtt{BitCount}\}, \mathtt{BitCount} \rangle,$ where
\[
\begin{array}{l}
x:\mathsf{bin} \vdash y \leftarrow \mathtt{BinSucc} \leftarrow x :: (y:\mathsf{bin})\\ 
x:\mathsf{bin} \vdash y \leftarrow \mathtt{Counter}\leftarrow x :: (y:\mathsf{ctr})\\
x:\mathsf{bin} \vdash y \leftarrow \mathtt{NumBits} \leftarrow x :: (y:\mathsf{bin}) \\
x:\mathsf{bin} \vdash y \leftarrow \mathtt{BitCount}\leftarrow x :: (y:\mathsf{ctr})\\
\end{array}
\]
We define the relation $\subset$ on process variables as $\mathtt{BinSucc} \subset_{0} \mathtt{Counter} \subset_{0} \mathtt{BitCount}$ and $\mathtt{BinSucc} \subset_{0} \mathtt{NumBits} \subset_{0} \mathtt{BitCount}$.
The process definitions are as follows, shown here already with their termination analysis.
\begin{small}
  \begin{align*}
 & w^{\beta} \leftarrow \mathtt{BinSucc}\leftarrow z^{\alpha}= &  {\color{blue}[0, 0,\ 0,\ 0]} &  \phantom{L}\\
 &\phantom{s} \mathbf{case}\, Lz^{\alpha}\ (\mu_{bin} \Ra & [0,0,-1,0] &  \phantom{L} z^{\alpha+1}_1=z^\alpha_1, z^{\alpha+1}_2 < z^{\alpha}_2\\  &  \phantom{small}     \mathbf{case}\, Lz^{\alpha+1}\ ( \mathit{b0} \Rightarrow Rw^{\beta}.\mu_{bin};   & {[0,0,-1,1]} & { \phantom{Lz^0.\nu} w^{\beta+1}_1=w^{\beta}_1}\\
&   \phantom{small space  more than}  Rw^{\beta+1}.\mathit{b1}; w^{\beta+1} \leftarrow z^{\alpha+1} & {[0,0,-1,1]} \\
&  \phantom{small space morex}    \mid \mathit{b1} \Rightarrow Rw^{\beta}.\mu_{bin};   & {[0,0,-1,1]} & { \phantom{Lz^0.\nu} w^{\beta+1}_1=w^{\beta}_1}\\
&   \phantom{small space  more than}  Rw^{\beta+1}.\mathit{b0};  w^{\beta+1} \leftarrow \mathtt{BinSucc} \leftarrow z^{\alpha+1} & { \color{red}[0,0,-1,1]} \\
&  \phantom{small space morex}    \mid \mathit{\$}\Rightarrow Rw^{\beta}.\mu_{\mathsf{bin}}; Rw^{\beta+1}.b1;\ & { [0,0,-1,1]} & { \phantom{Lz^0.\nu} w^{\beta+1}_1=w^{\beta}_1}\\
&  \phantom{small space more than} Rw^{\beta+1}.\mu_{\mathsf{bin}}; Rw^{\beta+2}.\$;\ w^{\beta+2} \leftarrow z^{\alpha+1})) & { [0,0,-2,2]} & { \phantom{Lz^0.\nu} w^{\beta+2}_1=w^{\beta+1}_1}\\[10pt]
  & y^{\beta} \leftarrow \mathtt{Counter}\leftarrow w^{\alpha}= &  {\color{blue}[0, 0,\ 0,\ 0]} \\
 &\phantom{s} \mathbf{case}\, Ry^{\beta}\ (\nu_{ctr} \Ra & [-1,0,0,0] &  \phantom{L} y^{\beta+1}_1<y^\beta_1, y^{\beta+1}_2 = y^{\beta}_2\\  &  \phantom{small}     \mathbf{case}\, Ry^{\beta+1}\ ( \mathit{inc} \Rightarrow z^{0} \leftarrow \mathtt{BinSucc} \leftarrow w^{\alpha};  & & \phantom{L}{\mathtt{BinSucc} \subset_{0} \mathtt{Counter}} \\
&   \phantom{small space  more than}  y^{\beta+1} \leftarrow \mathtt{Counter} \leftarrow z^{0} & {\color{red} [-1, \infty, \infty ,0]}  &\phantom{L.}\\
&  \phantom{small space morex}    \mid \mathit{val} \Rightarrow y^{\beta+1} \leftarrow w^\alpha))  & { [-1,0,0,0]} 
\end{align*}\\[-10pt]
\begin{align*}
 & w^{\beta} \leftarrow \mathtt{NumBits}\leftarrow x^{\alpha}= & {\color{blue}[0, 0,\ 0,\ 0]} \\
 &\phantom{s} \mathbf{case}\, Lx^{\alpha}\ (\mu_{bin} \Ra & [0,0,-1,0] &  \phantom{L} x^{\alpha+1}_1=x^\alpha_1, x^{\alpha+1}_2 < x^{\alpha}_2\\  & \phantom{small}     \mathbf{case}\, Lx^{\alpha+1}\ (\mathit{b0} \Rightarrow z^{0} \leftarrow \mathtt{NumBits} \leftarrow x^{\alpha+1};   & { \color{red}[?,0,-1,?]} & {\color{red} \phantom{Lx^0.\nu } z^0_1\stackrel{?}{=}w^{\beta}_1, z_2^0\stackrel{?}{=}w^{\beta}_2}\\
&   \phantom{small space  more than}  w^{\beta} \leftarrow \mathtt{BinSucc} \leftarrow z^{0} &  &\phantom{L} \mathtt{BinSucc} \subset_{0} \mathtt{NumBits}\\
&  \phantom{small space more}    \mid \mathit{b1} \Rightarrow z^{0} \leftarrow \mathtt{NumBits} \leftarrow x^{\alpha+1};   & { \color{red}[?,0,-1,?]} & {\color{red} \phantom{Lx^0.\nu } z^0_1\stackrel{?}{=}w^{\beta}_1, z_2^0\stackrel{?}{=}w^{\beta}_2}\\
&   \phantom{small space  more than}  w^{\beta} \leftarrow \mathtt{BinSucc} \leftarrow z^{0} &   &\phantom{L1} \mathtt{BinSucc} \subset_{0} \mathtt{NumBits}\\
&  \phantom{small space more}    \mid \mathit{\$}\Rightarrow Rw^{\beta}.\mu_{\mathsf{bin}}; Rw^{\beta+1}.\$;\ w^{\beta+1} \leftarrow x^{\alpha+1})) & { [0,0,-1,1]} & 
\end{align*}
$\;\;\, y^{\beta} \leftarrow \mathtt{BitCount} \leftarrow x^{\alpha} = w^{0}\leftarrow \mathtt{NumBits}\leftarrow x^{\alpha}; y^{\beta}\leftarrow \mathtt{Counter}\leftarrow w^{0}$\\
\end{small}%

\noindent The program starts with process $\mathtt{BitCount}$ which creates a fresh channel $w^0$, spawns a new process $w^0 \leftarrow \mathtt{NumBits} \leftarrow x^\alpha$, and continues as $y^\beta \leftarrow \mathtt{Counter} \leftarrow w^0$.

Process $y^\beta \leftarrow \mathtt{Counter} \leftarrow w^\alpha$ as its name suggests works as a counter where $w : \mathit{bin}$ is the current value of the counter.  When it receives an increment message $\mathit{inc}$ it computes the successor of $w$, accessible through channel $z$.  If it receives a $\mathit{val}$ message it simply forwards the current value ($w$) to the client ($y$). Note that in this process, both calls are valid according to the condition developed so far.  This is also true for the binary successor process $\mathtt{BinSucc}$, which presents no challenges.
The only recursive call represents the ``carry'' of binary addition when a number with lowest bit $\mathit{b1}$ has to be incremented.

The process $w^\beta \leftarrow \mathtt{NumBits} \leftarrow x^\alpha$ counts the number of bits in the binary number $x$ and sends the result along $w$, also in the form of a binary number.  It calls itself recursively for every bit received along $x$ and increments the result $z$ to be returned along $w$.  Note that if there are no leading zeros, this computes essentially the integer logarithm of $x$.
The process $\mathtt{NumBits}$ is reactive.   However with our approach toward spawning a new process, the recursive calls have the list value $[\infty, 0, -1, \infty]  \not < [0,0,0,0]$, meaning that the local validity condition developed so far fails.

Note that we cannot just define $z_1^0= w_1^\beta$ and $z_2^0=w_2^\beta$, or $z_1^0=z_2^0=0$. Channel $z^0$ is a fresh one and its relation with the future generations depends on how it evolves in the process $w^\beta \leftarrow \mathtt{BinSucc} \leftarrow z^0$. But by definition of type $\mathsf{bin}$, no matter how $z^0:\mathsf{bin}$ evolves to some $z^\eta$ in process $\mathtt{BinSucc}$, it won't be the case that $z^\eta: \mathsf{ctr}$. In other words, the type $\mathsf{ctr}$ is not visible from $\mathsf{bin}$ and for any generation $\eta$, channel $z^\eta$ does not send or receive a $\mathsf{ctr}$ unfolding message. So in this recursive call, the value of $z_1^\eta$ is not important anymore and we safely put $z_1^0=w^{\beta}_1$. In the improved version of the condition we have:

\begin{small}
 \begin{align*} 
 & w^{\beta} \leftarrow \mathtt{NumBits}\leftarrow x^{\alpha}= & {\color{blue}[0, 0,\ 0,\ 0]} &  \phantom{L}\\
 &\phantom{s} \mathbf{case}\, Lx^{\alpha}\ (\mu_{bin} \Ra & [0,0,-1,0] &  \phantom{L} x^{\alpha+1}_1=x^\alpha_1, x^{\alpha+1}_2 < x^{\alpha}_2\\  &  \phantom{small}     \mathbf{case}\, Lx^{\alpha+1}\ ( \mathit{b0} \Rightarrow z^{0} \leftarrow \mathtt{NumBits} \leftarrow x^{\alpha+1};   & { \color{red}[0,0,-1,\infty]} & {\color{red} \phantom{Lx^0.\nu} z^0_1=w^{\beta}_1 }\\
&   \phantom{small space  more than}  w^{\beta} \leftarrow \mathtt{BinSucc} \leftarrow z^{0} &  & \mathtt{BinSucc} \subset_{0} \mathtt{NumBits}\\
&  \phantom{small space more}    \mid \mathit{b1} \Rightarrow z^{0} \leftarrow \mathtt{NumBits} \leftarrow x^{\alpha+1};   & { \color{red}[0,0,-1,\infty]} & {\color{red} \phantom{Lx^0.\nu} z^0_1=w^{\beta}_1}\\
&   \phantom{small space  more than}  w^{\beta} \leftarrow \mathtt{BinSucc} \leftarrow z^{0} &   &\mathtt{BinSucc} \subset_{0} \mathtt{NumBits}\\
&  \phantom{small space more}    \mid \mathit{\$}\Rightarrow Rw^{\beta}.\mu_{\mathsf{bin}}; Rw^{\beta+1}.\$;\ w^{\beta+1} \leftarrow x^{\alpha+1})) & { [0,0,-1,1]} & 
\end{align*}
\end{small}%

\noindent This version of the algorithm recognizes both recursive calls as valid. In the following definition we capture the idea of visibility from a type more formally.
\end{example}

\begin{definition}
For type $A$ in a given signature $\Sigma$ and a set of type variables $\Delta$, we define $\mathtt{c}(A; \Delta)$ inductively as:
\[
\begin{array}{l}
\mathtt{c}(1; \Delta)=\emptyset,\\
\mathtt{c}(\oplus\{\ell:A_{\ell}\}_{\ell \in L}; \Delta)= \mathtt{c}(\&\{\ell:A_{\ell}\}_{\ell \in L}; \Delta)= \bigcup_{\ell \in L}\mathtt{c}(A_{\ell}; \Delta),\\
\mathtt{c}(t; \Delta)= \{t\} \cup \mathtt{c}(A; \Delta \cup \{t\})\ \text{if}\ t=_{a}\ A\in \Sigma\ \text{and}\ t \not \in \Delta,\\
\mathtt{c}(t; \Delta)= \{t\}\ \text{if}\ t=_{a} A\in \Sigma\ \text{and} \ t \in \Delta.
\end{array}
\]
We put priority $i$ in the set $\mathtt{c}(A)$ iff for some type variable $t$ with $i=p(t)$, $t \in \mathtt{c}(A; \emptyset)$.
We say that \emph{priority $i$ is visible from type $A$} if and only if  $i \in \mathtt{c}(A)$. 
\end{definition}
In Example~\ref{cutchannel}, we have $\mathtt{c}(\mathsf{bin})=\{p(\mathsf{bin})\}=\{2\}$ and $\mathtt{c}(\mathsf{ctr})=\{p(\mathsf{bin}),p(\mathsf{ctr})\} = \{1,2\}$ which means that $\mathsf{bin}$ is visible from
$\mathsf{ctr}$ but not the other way around.  This expresses that the definition of $\mathsf{ctr}$ references $\mathsf{bin}$, but the definition of $\mathsf{bin}$ does not reference $\mathsf{ctr}$.

\section{Typing Rules for Session-Typed Processes with Channel Ordering}
\label{rules}

In this section we introduce inference rules for session-typed processes corresponding to derivations in subsingleton logic with fixed points. This is a refinement of the inference rules in Figure~\ref{fig:annotrule} to account for channel generations and orderings introduced in previous sections. The judgments are of the form 
\[\bar{x}^{\alpha}:\omega \vdash_{\Omega} P :: (y^{\beta}: A),\] where $P$ is a process, and $x^\alpha$ (the $\alpha$-th generation of channel $x$) and $y^ \beta$  (the $\beta$-th generation of channel $y$) are its left and right channels of types $\omega$ and $A$, respectively.  The order relation between the generations of left and right channels indexed by their priority of types is built step by step in $\Omega$ when reading the rules from the conclusion to the premises. We only consider judgments in which all variables $x^{\alpha'}$ occurring in $\Omega$ are such that $\alpha' \leq \alpha$ and, similarly, for $y^{\beta'}$ in $\Omega$ we have $\beta' \leq \beta$. This presupposition guarantees that if we construct a derivation bottom-up, any future generations for $x$ and $y$ are fresh and not yet constrained by $\Omega$. All our rules, again read bottom-up, will preserve this property.

We fix a signature $\Sigma$ as in Definition~\ref{signature}, a finite set of process definitions $V$ over $\Sigma$ as in Definition~\ref{process}, and define $\bar{x}^{\alpha}:\omega \vdash_{\Omega} P :: (y^{\beta}: A)$ with the rules in Figure~\ref{fig:stp-order}. To preserve freshness of channels and their future generations in $\Omega$, the channel introduced by $\msc{Cut}$ rule must be distinct from any variable mentioned in $\Omega$.
Similar to its underlying sequent calculus in Section~\ref{operat}, this system is infinitary, i.e., an infinite derivation may be produced for a given program. However, we can remove the first premise from the $\msc{Def}$ rule and check typing for each process definition in $V$ separately.

Programs derived in this system are all {\it well-typed}, but not necessarily \emph{valid}. It is, however, the basis for our finitary condition in Section~\ref{alg} and in Section~\ref{Guard} where we prove that local validity is stricter than Fortier and Santocanale's guard condition.

\begin{figure}
{\small\[ \infer[\msc{Id}]{x^{\alpha}: A \vdash_{\Omega} y^{\beta} \leftarrow x^{\alpha} :: (y^{\beta}: A)}{}\]
\[ \infer[\msc{Cut}^{w}]{ \bar{x}^{\alpha}: \omega \vdash_{\Omega}  (w \leftarrow P_{w} ; Q_{w}) :: (y^{\beta}: C)}{ \deduce {\bar{x}^{\alpha}:  \omega \vdash_{\Omega \cup \mathtt{r}(y^\beta)} P_{w^0} ::(w^0:A)}{\mathtt{r}(v)= \{w^{0}_{i}=v_{i}\mid i \not \in \mathtt{c}(A)\, \mbox{and}\, i \le n\}} & w^0: A \vdash_{\Omega \cup \mathtt{r}(\bar{x}^\alpha)} Q_{w^0} :: (y^{\beta}: C)}\]\\[-7.95pt]
\[ \infer[\oplus R]{\bar{x}^{\alpha}:\omega \vdash_{\Omega} Ry^{\beta}.k; P :: (y^{\beta}: \oplus\{\ell:A_{\ell}\}_{\ell \in L})}{\bar{x}^{\alpha}: \omega \vdash_{\Omega} P :: (y^{\beta}: A_{k}) \quad (k \in L)} \]
    \[\infer[\oplus L]{x^{\alpha}:\oplus\{ \ell:A_\ell \}_{ \ell \in L} \vdash_{\Omega} \mathbf{case}\, Lx^{\alpha} \ (\ell\Ra P_{\ell}):: (y^{\beta}: C)}{\forall \ell\in L \quad x^{\alpha}:A_{\ell} \vdash_{\Omega} P_\ell :: (y^{\beta}:C)}\]\\[-7.95pt]
    \[
    \infer[\& R]{\bar{x}^{\alpha}: \omega \vdash_{\Omega} \mathbf{case}\, Ry^{\beta}\ (\ell \Ra P_\ell) :: (y^{\beta}: \& \{\ell:A_\ell\}_{\ell \in L})}{\forall \ell\in L \quad \bar{x}^{\alpha}: \omega \vdash_{\Omega} P_\ell :: (y^{\beta}:A_{\ell})}\]
    \[\infer[\& L]{x^{\alpha}: \&\{ \ell:A_l \}_{ \ell \in L} \vdash_{\Omega} Lx^{\alpha}.k; P :: (y^{\beta}:C)}{k\in L \quad x^{\alpha}:  A_{k} \vdash_{\Omega}  P :: (y^{\beta}:C)}
    \]\\[-7.95pt]
    \[\begin{tabular}{c c}
    \infer[1R]{. \vdash_{\Omega} \mathbf{close}\, Ry^\beta :: (y^{\beta}: 1)}{} & \infer[1L]{x^{\alpha}: 1 \vdash_{\Omega} \mathbf{wait}\, Lx^\alpha;Q :: (y^{\beta}: A)}{  . \vdash_{\Omega} Q :: (y^{\beta}: A) }
\end{tabular}
\]\\
\[\infer[\mu R]{ \bar{x}^{\alpha}: \omega \vdash_{\Omega} Ry^{\beta}.\mu_t; P_{y^{\beta}} :: (y^{\beta}:t)}{\deduce{  \bar{x}^{\alpha}: \omega \vdash_{\Omega'} P_{y^{\beta+1}} :: (y^{\beta+1}:A)}{\Omega'= \Omega \cup \{(y^{\beta})_{p(s)} =(y^{\beta+1})_{p(s)} \mid p(s)\neq p(t)\}} & t=_{\mu}A & }\]
\[ \infer[\mu L]{x^{\alpha}: t \vdash_{\Omega} \mathbf{case}\, Lx^{\alpha}\ (\mu_{t} \Ra Q_{x^{\alpha}}):: (y^{\beta}: C)}{\deduce{x^{\alpha+1}: A \vdash_{\Omega'} Q_{x^{\alpha+1}} :: (y^{\beta}:C)}{ \Omega'=\Omega \cup \{x^{\alpha+1}_{p(t)} < x^{\alpha}_{p(t)}\} \cup \{x^{\alpha+1}_{p(s)} =x^{\alpha}_{p(s)} \mid p(s)\neq p(t)\} }  &  t=_{\mu} A }\]\\
\[
\infer[\nu R]{\bar{x}^{\alpha}: \omega \vdash_{\Omega} \mathbf{case}\, Ry^{\beta} \ (\nu_t \Ra P_{y^{\beta}}) :: (y^{\beta}: t)}{\deduce{\bar{x}^{\alpha}: \omega \vdash_{\Omega'} P_{y^{\beta+1}} :: (y^{\beta+1}: A)}{\Omega'= \Omega \cup \{y^{\beta+1}_{p(t)} < y^{\beta}_{p(t)}\} \cup \{y^{\beta+1}_{p(s)} = y^{\beta}_{p(s)} \mid p(s)\neq p(t)\}}  & t=_{\nu}A  }\]
\[ \infer[\nu L]{x^{\alpha}: t \vdash_{\Omega} Lx^{\alpha}.\nu_{t}; Q_{x^{\alpha}}:: (y^{\beta}: C)}{\deduce{x^{\alpha+1}: A \vdash_{\Omega'} Q_{x^{\alpha+1}} :: (y^{\beta}: C)} {\Omega'=\Omega \cup \{(x^{\alpha+1})_{p(s)} =(x^{\alpha})_{p(s)} \mid p(s)\neq p(t) \}} & t=_{\nu} A }\]\\
\[\infer[\msc{Def}(X)]{\bar{x}^{\alpha}: \omega \vdash_{\Omega} y^\beta \leftarrow X \leftarrow \bar{x}^\alpha:: (y^{\beta}: C)}{\bar{x}^{\alpha}: \omega \vdash_{\Omega} P_{\bar{x}^\alpha, y^\beta} :: (y^{\beta}: C) & \bar{u}:\omega \vdash X=P_{\bar{u},w} :: (w:C) \in V }\]}
\caption{Infinitary Typing Rules for Processes with Channel Ordering}
\label{fig:stp-order}
\end{figure}
\begin{figure}
\small{
\[\infer[\msc{Id}]{\langle\bar{u}^\gamma,X,v^{\delta}\rangle; z^\alpha:A \vdash_{\Omega, \subset} w^\beta \leftarrow z^\alpha :: (w^{\beta}:A)}{}\]
\[ \infer[\msc{Cut}^{x}]{ \langle\bar{u}^\gamma,X,v^{\delta}\rangle; \bar{z}^\alpha: \omega \vdash_{\Omega \subset}  (x \leftarrow P_x; Q_x) :: (w^{\beta}:C)}{ \deduce{ \langle\bar{u}^\gamma,X,v^{\delta}\rangle;  \bar{z}^\alpha:\omega \vdash_{\Omega \cup \mathtt{r}(w^\beta), \subset} P_{x^0} :: (x^{0}:A)} {\mathtt{r}(y)= \{x^{0}_{i}=y_{i}\mid i \not \in \mathtt{c}(A) \,\mbox{and}\, i \le n\}} \ \  \langle\bar{u}^\gamma,X,v^{\delta}\rangle; x^0:A \vdash_{\Omega \cup \mathtt{r}(\bar{z}^\alpha), \subset} Q_{x^0} :: (w^{\beta}:C) } \]\\[-8.5pt]
\[\infer[\oplus R]{\langle\bar{u}^\gamma,X,v^{\delta}\rangle; \bar{z}^\alpha: \omega \vdash_{\Omega, \subset} Rw^{\beta}.k; P :: (w^{\beta}: \oplus\{\ell:A_l\}_{\ell \in L})}{\langle\bar{u}^\gamma,X,v^{\delta}\rangle; \bar{z}^\alpha: \omega \vdash_{\Omega, \subset} P::(w^{\beta}: A_{k}) \quad (k \in L)}\]
\[\infer[\oplus L]{\langle\bar{u}^\gamma,X,v^{\delta}\rangle; z^\alpha: \oplus\{ \ell:A \}_{ \ell \in L} \vdash_{\Omega, \subset} \mathbf{case}\, Lz^{\alpha} \ (\ell\Ra P_{\ell}):: ( w^{\beta}: C)}{\forall \ell\in L \quad \langle\bar{u}^\gamma,X,v^{\delta}\rangle; z^\alpha:A_{\ell} \vdash_{\Omega, \subset}  P_\ell :: ( w^{\beta}: C)}\]\\[-8.5pt]
\[\infer[\& R]{\langle\bar{u}^\gamma,X,v^{\delta}\rangle; \bar{z}^\alpha: \omega \vdash_{\Omega, \subset}\mathbf{case}\, Rw^\beta\ (\ell \Ra P_\ell) :: ( w^{\beta}: \& \{\ell:A_\ell\}_{\ell \in L})}{\forall \ell\in L \quad \langle\bar{u}^\gamma,X,v^{\delta}\rangle; \bar{z}^\alpha: \omega \vdash_{\Omega, \subset} P_\ell :: ( w^{\beta}: A_{\ell})}\]\\[-8.5pt]
\[
\infer[\& L]{\langle\bar{u}^\gamma,X,v^{\delta}\rangle; z^\alpha: \&\{ \ell:A_\ell \}_{ \ell \in L} \vdash_{\Omega, \subset} Lz^{\alpha}.k; P :: ( w^{\beta}: C)}{(k\in L) \quad \langle\bar{u}^\gamma,X,v^{\delta}\rangle;  z^\alpha: A_{k} \vdash_{\Omega, \subset}  P :: ( w^{\beta}:C)}\]\\[-8.5pt]
\[\begin{tabular}{c c}
\infer[1R]{\langle\bar{u}^\gamma,X,v^{\delta}\rangle; \cdot \vdash_{\Omega, \subset} \mathbf{close}\, R :: (w^{\beta}:1)}{} &
\infer[1L]{\langle\bar{u}^\gamma,X,v^{\delta}\rangle; z^\alpha: 1 \vdash_{\Omega, \subset} \mathbf{wait}\, Lz^\alpha;Q :: (w^{\beta}:A)}{ \langle\bar{u}^\gamma,X,v^{\delta}\rangle; \cdot \vdash_{\Omega, \subset} Q :: (w^{\beta}:A) }
\end{tabular}\]\\[-6.5pt]
\[ \infer[\mu R]{ \langle\bar{u}^\gamma,X,v^{\delta}\rangle; \bar{z}^{\alpha}: \omega \vdash_{\Omega, \subset} Rw^\beta.\mu_t; P_{w^\beta} :: (w^{\beta}:t)}{\deduce{\langle\bar{u}^\gamma,X,v^{\delta}\rangle;  \bar{z}^\alpha:\omega \vdash_{\Omega', \subset} P_{w^{\beta+1}} :: (w^{\beta+1}:A)}{ \Omega'= \Omega \cup \{w^{\beta}_{p(s)} = w^{\beta+1}_{p(s)} \mid p(s)\neq p(t)\}} & t=_{\mu}A }\]\\[-8.5pt]
\[\infer[\mu L]{\langle\bar{u}^\gamma, X,v^{\delta}\rangle; z^\alpha: t \vdash_{\Omega, \subset} \mathbf{case}\, Lz^{\alpha}\ (\mu_{t} \Ra Q_{z^{\alpha}}):: (w^{\beta}:C)}{\deduce{\langle\bar{u}^\gamma,X,v^{\delta}\rangle; z^{\alpha+1}:A \vdash_{\Omega', \subset} Q_{z^{\alpha+1}} : (w^{\beta}::C) } {\Omega'=\Omega \cup \{z^{\alpha+1}_{p(t)} < z^{\alpha}_{p(t)}\} \cup \{z^{\alpha+1}_{p(s)} =z^{\alpha}_{p(s)} \mid p(s)\neq p(t)\}  } & t=_{\mu} A }\]\\[-8.5pt]
\[ \infer[\nu R]{\langle\bar{u}^\gamma,X,v^{\delta}\rangle; \bar{z}^\alpha: \omega \vdash_{\Omega, \subset} \mathbf{case}\, Rw^{\beta} \ (\nu_t \Ra P_{w^{\beta}}) :: (w^{\beta}:t)}{\deduce{\langle\bar{u}^\gamma,X,v^{\delta}\rangle; \bar{z}^\alpha: \omega \vdash_{\Omega', \subset} P_{w^{\beta+1}} :: (w^{\beta+1}:A)}{\Omega'= \Omega \cup \{w^{\beta+1}_{p(t)} < w^{\beta}_{p(t)}\} \cup \{w^{\beta+1}_{p(s)} =w^{\beta}_{p(s)} \mid p(s)\neq p(t)\}} & t=_{\nu} A }\]\\[-9pt]
\[\infer[\nu L]{\langle\bar{u}^{\gamma},X,v^{\delta}\rangle; z^{\alpha}: t \vdash_{\Omega, \subset} Lz^{\alpha}.\nu_{t}; Q_{z^{\alpha}}:: (w^{\beta}:C)}{\deduce{  \langle\bar{u}^{\gamma},X,v^{\delta}\rangle; z^{\alpha+1}: A \vdash_{\Omega', \subset} Q_{z^{\alpha+1}} :: (w^{\beta}:C)} {\Omega'=\Omega \cup \{z^{\alpha+1}_{p(s)} =z^{\alpha}_{p(s)} \mid p(s)\neq p(t)\} } & t=_{\nu} A }\]\\[-9pt]
\[ \infer[\msc{Call}]{\langle \bar{u}^\gamma,X,v^{\delta}\rangle; \bar{z}^\alpha: \omega \vdash_{\Omega, \subset} w^\beta \leftarrow Y \leftarrow \bar{z}^\alpha :: (w^{\beta}:C)}{ Y, list(\bar{z}^\alpha,w^{\beta}) \ (\subset, <_{\Omega})\ X, list(\bar{u}^\gamma, v^{\delta}) &  \bar{x}:\omega \vdash Y=P_{\bar{x},y} :: (y:C) \in V}\]}
\caption{Finitary Rules for Local Validity}
\label{fig:validity}
\end{figure}

\section{A Local Validity Condition}
\label{alg}

In Sections~\ref{session} to~\ref{algn}, using several examples, we developed an algorithm for identifying {\it valid} programs. Illustrating the full algorithm based on the inference rules in Section~\ref{rules} was postponed to this section. We reserve for the next section our main result that the programs accepted by this algorithm satisfy the guard condition introduced by Fortier and Santocanale~\cite{Fortier13csl}.

The condition checked by our algorithm is a \emph{local} one in the sense that we check validity of each process definition in a program separately. The algorithm works on the sequents of the form 
\[\langle\bar{u}^\gamma,X,v^{\delta}\rangle; \bar{z}^\alpha:\omega \vdash_{\Omega, \subset} P:: (w^{\beta}:C),\] where $\bar{u}^\gamma$ is the left channel of the process the algorithm started with and can be either empty or $u^\gamma$. Similarly, $v^{\delta}$ is the right channel of the process the algorithm started with (which cannot be empty). And $X$ is the last process variable a definition rule has been applied to (reading the rules bottom-up). Again, in this judgment the (in)equalities in $\Omega$ can only relate variables $z$ and $w$ from earlier generations to guarantee freshness of later generations.

Generally speaking, when analysis of the program starts with $\bar{u}^\gamma:\omega \vdash v^\delta \leftarrow X \leftarrow \bar{u}^\gamma :: (v^\delta: B)$, a snapshot of the channels $\bar{u}^\gamma$ and $v^{\delta}$ and the process variable $X$ are saved. Whenever the process reaches a call $\bar{z}^\alpha:\_\vdash  w^\beta \leftarrow Y \leftarrow \bar{z}^\alpha:: (w^\beta: \_)$, the algorithm compares $X,list(\bar{u}^\gamma, v^\delta)$ and $Y,list(\bar{z}^\alpha, w^\beta)$ using the $(\subset, <)$ order to determine if the call is (locally) valid.
This comparison is made by the $\msc{Call}$ rule in the rules in Figure~\ref{fig:validity}, and is local in the sense that only the interface of a process is consulted at each call site, not its definition.  Since it otherwise follows the structure of the program it is also local in the sense of Pierce and Turner~\cite{Pierce2000toplas}.

\begin{definition}
A program $\mathcal{P}= \langle V,S \rangle$ over signature $\Sigma$ and a fixed order $\subset$ satisfying the properties in Definition~\ref{subs} is {\it locally valid} iff for every $\bar{z}:A\vdash X=P_{\bar{z},w}:: (w:C) \in V$, there is a derivation for \[\langle\bar{z}^0,X,w^0\rangle;  \bar{z}^0:\omega \vdash_{\emptyset, \subset} P_{\bar{z}^0, w^0}:: (w^0:C)\] in the rule system in Figure~\ref{fig:validity}. This set of rules is \emph{finitary} so it can be directly interpreted as an algorithm.  This results from substituting the $\msc{Def}$ rule (of Figure~\ref{fig:stp-order}) with the $\msc{Call}$ rule (of Figure~\ref{fig:validity}). Again, to guarantee freshness of future generations of channels, the channel introduced by $\msc{Cut}$ rule is distinct from other variables mentioned in $\Omega$.
\end{definition}

The starting point of the algorithm can be of an arbitrary form \[\langle\bar{z}^\alpha,X,w^{\beta}\rangle; \bar{z}^\alpha:\omega \vdash_{\Omega, \subset} P_{z^{\alpha},w^{\beta}}::(w^{\beta}: C),\]
as long as $\bar{z}^{\alpha+i}$ and $w^{\beta+i}$ do not occur in $\Omega$
for every $i>0$. In both the inference rules and the algorithm, it is implicitly assumed that the next generation of channels introduced in the $\mu/\nu-R/L$ rules do not occur in $\Omega$. Having this condition we can convert a proof for \[\langle\bar{z}^0,X,w^0\rangle; \bar{z}^0:\omega \vdash_{\emptyset, \subset} P_{z^{0},w^{0}}:: (w^0:C),\] to a proof for \[\langle\bar{z}^\alpha,X,w^{\beta}\rangle; \bar{z}^\alpha:\omega \vdash_{\Omega, \subset} P_{z^{\alpha},w^{\beta}}::(w^{\beta}: C),\] by rewriting each $\bar{z}^{\gamma}$ and $w^{\delta}$ in the proof as $\bar{z}^{\gamma+\alpha}$ and $w^{\delta+\beta}$, respectively. This simple proposition is used in the next section where we prove that every locally valid process accepted by our algorithm is a valid proof according to the FS guard condition.

\begin{proposition}\label{subuv}
If there is a deduction of \[\langle\bar{z}^0,X,w^0\rangle; \bar{z}^0:\omega \vdash_{\emptyset, \subset} P_{z^{0},w^{0}}:: (w^0:C),\] then there is also a deduction of \[\langle\bar{z}^\alpha,X,w^{\beta}\rangle; \bar{z}^\alpha:\omega \vdash_{\Omega, \subset} P_{z^{\alpha},w^{\beta}}::(w^{\beta}: C),\]  if for all $0<i$, $\bar{z}^{\alpha+i}$ and $w^{\beta+i}$ do not occur in $\Omega$.
\end{proposition}
\begin{proof}
By substitution, as explained above.
\end{proof}
To show the algorithm in action we run it over program $\mathcal{P}_3:= \langle \{\mathtt{Copy}\}, \mathtt{Copy} \rangle$ previously defined in Example~\ref{prex}.
\begin{example}\label{ex:copyalg}
Consider program $\mathcal{P}_3:= \langle \{\mathtt{Copy}\}, \mathtt{Copy} \rangle$ over signature $\Sigma_1$ where $\mathtt{Copy}$ has types $x:\mathsf{nat} \vdash \mathtt{Copy} ::(y:\mathsf{nat})$.
\begin{align*}
    \Sigma_1:=  \mathsf{nat}=^{1}_{\mu} \oplus\{ \mathit{z}:\mathsf{1}, \mathit{s}:\mathsf{nat}\},
\end{align*} 
\begin{align*}
y \leftarrow \mathtt{Copy} \leftarrow x  =\mathbf{case}\, Lx\ (\mu_{nat} \Ra \mathbf{case}\, Lx\ & (\ z\Ra Ry.\mu_{nat}; Ry.z; \mathbf{wait}\, Lx; \mathbf{close}\, Ry  \\
 & \mid s \Rightarrow Ry.\mu_{nat}; Ry.s; y \leftarrow \mathtt{Copy} \leftarrow x)). 
\end{align*}
In this example, following Definition~\ref{subs} the programmer has to define $\mathtt{Copy}\subseteq_1 \mathtt{Copy}$ since the only priority in $\Sigma$ is $1$.
To verify local validity of this program we run our algorithm over the definition of $\mathtt{Copy}$. Here we show the interesting branch of the constructed derivation:

\setlength{\inferLineSkip}{5pt} 
\setlength{\inferLabelSkip}{1pt}

\[\small \infer[\mu L]{ x^0{:}\mathsf{nat} \vdash_{\emptyset} \mathbf{case}\, Lx^0\ (\mu_{nat}\Rightarrow \cdots):: (y^0{:}\mathsf{nat})}{  \infer[\oplus L]{ x^1{:}1\oplus \mathsf{nat} \vdash_{\{x^1_1<x^0_1\}} \mathbf{case}\, Lx^1\ (\cdots):: (y^0{:} \mathsf{nat})}{\boldsymbol{\cdots} \hspace{-20pt} & \infer[\mu R]{x^1{:}\mathsf{nat} \vdash_{\{x^1_1<x^0_1\}} Ry^0.\mu_{nat}; \cdots :: (y^0{:} \mathsf{nat})}{\infer[\oplus R]{x^1{:}\mathsf{nat} \vdash_{\{x^1_1<x^0_1\}} Ry^1.s; \cdots :: (y^1{:} 1 \oplus \mathsf{nat})}{\infer[\msc{Call}]{ x^1{:}\mathsf{nat} \vdash_{\{x^1_1<x^0_1\}} y^1 \leftarrow \mathtt{Copy} \leftarrow x^1 :: (y^1{:} \mathsf{nat})}{ [x^1_1, \mathtt{Copy}, y^1_1]\,(\subset, <_{\{x^1_1<x^0_1\}})\,[x^0_1, \mathtt{Copy}, y^0_1]\ \  &  x{:}\mathsf{nat} \vdash \mathtt{Copy}=(\mathbf{case}\, Lx\,(\cdots) )_{x,y} :: (y{:}\mathsf{nat}) \in V}}}}}\]

As being checked by the $\msc{Call}$ rule, $[x^1_1, \mathtt{Copy}, y^1_1]\,(\subset, <_{\{x^1_1<x^0_1\}})\,[x^0_1, \mathtt{Copy}, y^0_1]$ and the recursive call is accepted. In this particular setting in which $\mathsf{Copy}$ calls \emph{itself} recursively, the condition of the $\msc{Call}$ rule can be reduced to $[x^1_1,  y^1_1]\, <_{\{x^1_1<x^0_1\}}\,[x^0_1, y^0_1]$.

Note that at a meta-level the generations on channel names and the set $\Omega$ are both used for bookkeeping purposes. We showed in this example that using the rules of Figure~\ref{fig:validity} as an algorithm we can annotate the given definition of a process variable with the generations and the set $\Omega$.  
\end{example}
\section{Local Validity and Guard Conditions}
\label{Guard}

Fortier and Santocanale \cite{Fortier13csl} introduced a {\it guard condition} for identifying valid circular proofs among all infinite pre-proofs in the singleton logic with fixed points. They showed that the pre-proofs satisfying this condition, which is based on the definition of left $\mu$- and right $\nu$-traces, enjoy the cut elimination property. In this section, we translate their guard condition into the context of session-typed concurrency and generalize it for subsingleton logic. It is straightforward to show that the cut elimination property holds for a proof in subsingleton logic if it satisfies the generalized version of the guard condition. The key idea is that cut reductions for individual rules stay untouched in subsingleton logic and rules for the new constant $1$ only provide more options for the cut reduction algorithm to terminate. We prove that all locally valid programs in the session typed system, determined by the algorithm in Section~\ref{alg}, also satisfy the guard condition. We conclude that our algorithm imposes a stricter but local version of validity on the session-typed programs corresponding to circular pre-proofs.

Here we adapt definitions of the {\it left} and {\it right traceable} paths, {\it left $\mu$-} and {\it right $\nu$-traces,} and then {\it validity} to our session type system.

\begin{definition}\label{rt}
Consider path $\mathbb{P}$ in the (infinite) typing derivation of a program $\mathcal{Q}=\langle V, S\rangle$ defined on a signature $\Sigma$:
\[{\infer{\bar{z}^\alpha: \omega \vdash_{\Omega} Q :: (w^{\beta}:C)}{ \infer{\vdots}{{\bar{x}^\gamma: \omega' \vdash_{\Omega'} Q' :: (y^{\delta}:C')}}}  }\]
$\mathbb{P}$ is called {\it left traceable} if $\bar{z}$ and $\bar{x}$ are non-empty and $\bar{z}=\bar{x}$. It is called {\it right traceable} if $w=y$.

Moreover, $\mathbb{P}$ is called a \emph{cycle} over program $\mathcal{Q}$, if for some $X\in V$, we have $Q= w^\beta \leftarrow X \leftarrow \bar{z}^\alpha $ and $Q'=y^\delta \leftarrow X \leftarrow \bar{x}^\gamma$.
\end{definition}

\begin{definition}\label{tr}
A path $\mathbb{P}$ in the (infinite) typing derivation of a program  $\mathcal{Q}=\langle V, S\rangle$ defined over signature $\Sigma$ is a left $\mu$-trace if (i) it is left-traceable, (ii) there is a left fixed point rule applied in it, and (iii) the highest priority of its left fixed point rules is $i\le n$ such that $\epsilon(i)=\mu$.
Dually, $\mathbb{P}$ is a right $\nu$-trace if (i) it is right-traceable, (ii) there is a right fixed point rule applied in it, and (iii) the highest priority of its right fixed point rules is $i \le n$ such that $\epsilon(i)=\nu$.
\end{definition} 

\begin{definition}[FS guard condition on cycles]
\label{guard}
A program $\mathcal{Q}=\langle V, S\rangle$ defined on signature $\Sigma$ satisfies the FS guard condition if every cycle $\mathbb{C}$
\[{\infer{\bar{z}^\alpha: \omega \vdash_{\Omega} w^\beta \leftarrow X \leftarrow \bar{z}^\alpha :: (w^{\beta}:C)}{ \infer{\vdots}{{\bar{x}^\gamma: \omega' \vdash_{\Omega'} y^\delta \leftarrow X \leftarrow \bar{x}^\gamma :: (y^{\delta}:C')}}}  }\]
over $\mathcal{Q}$ is either  a left $\mu$-trace or a right $\nu$-trace.
Similarly, we say a single cycle $\mathbb{C}$ satisfies the guard condition if it is either  a left $\mu$-trace or a right $\nu$-trace.
\end{definition}

Definitions~\ref{rt}-\ref{guard} are equivalent to the definitions of the same concepts by Fortier and Santocanale using our own notation. In particular, Definition~\ref{guard} is equivalent to \emph{FS guard condition on cycles}. For the intended use of infinite derivations in this paper in which $V$ is a finite set of process definitions, the FS guard condition on infinite paths is equivalent to their condition on cycles.

As an example, consider program $\mathcal{P}_3:= \langle \{\mathtt{Copy}\}, \mathtt{Copy} \rangle$ over signature $\Sigma_1$, defined in Example~\ref{prex}, where $\mathtt{Copy}$ has types $x:\mathsf{nat} \vdash \mathtt{Copy} ::(y:\mathsf{nat})$.
\begin{align*}
    \Sigma_1:=  \mathsf{nat}=^{1}_{\mu} \oplus\{ \mathit{z}:\mathsf{1}, \mathit{s}:\mathsf{nat}\}
\end{align*} 
\begin{align*}
y \leftarrow \mathtt{Copy} \leftarrow x  =\mathbf{case}\, Lx\ (\mu_{nat} \Ra \mathbf{case}\, Lx\ & (\ z\Ra Ry.\mu_{nat}; Ry.z; \mathbf{wait}\, Lx; \mathbf{close}\, Ry  \\
 & \mid s \Rightarrow Ry.\mu_{nat}; Ry.s; y \leftarrow \mathtt{Copy} \leftarrow x))
\end{align*}
Consider the first several steps of the derivation of the program starting with $x^0:\mathsf{nat} \vdash_{\emptyset} y^0 \leftarrow \mathtt{Copy} \leftarrow x^0 ::(y^0:\mathsf{nat})$:

\setlength{\inferLineSkip}{5pt} 
\setlength{\inferLabelSkip}{1pt}
\begin{small}
\[\infer[\msc{Def}(\mathtt{Copy})]{\color{red}x^0:\mathsf{nat}\vdash_{\emptyset} y^0 \leftarrow \mathtt{Copy} \leftarrow x^0 :: (y^0:\mathsf{nat})}{ \infer[\mu L]{ x^0:\mathsf{nat} \vdash_{\emptyset} \mathbf{case}\, Lx^0\ (\mu_{nat}\Rightarrow \cdots):: (y^0:\mathsf{nat})}{  \infer[\oplus L]{  x^1:1\oplus \mathsf{nat} \vdash_{\{x^1_1<x^0_1\}} \mathbf{case}\, Lx^1\ (\cdots):: (y^0: \mathsf{nat})}{x^1:\mathsf{1} \vdash_{\{x^1_1<x^0_1\}} Ry^0.\mu_{nat}; \cdots :: (y^0: \mathsf{nat}) & \infer[\mu R]{x^1:\mathsf{nat} \vdash_{\{x^1_1<x^0_1\}} Ry^0.\mu_{nat}; \cdots :: (y^0: \mathsf{nat})}{\infer[\oplus R]{x^1:\mathsf{nat} \vdash_{\{x^1_1<x^0_1\}} Ry^1.s; \cdots :: (y^1: 1 \oplus \mathsf{nat})}{\color{blue} x^1:\mathsf{nat} \vdash_{\{x^1_1<x^0_1\}} y^1 \leftarrow \mathtt{Copy} \leftarrow x^1 :: (y^1: \mathsf{nat})}}}}}\]
\end{small}

\noindent
The path between \[{\color{red} x^0:\mathsf{nat} \vdash_{\emptyset} y^0\leftarrow \mathtt{Copy} \leftarrow x^0 :: (y^0: \mathsf{nat})}\] and \[{\color{blue} x^1:\mathsf{nat} \vdash_{\{x^1_1<x^0_1\}} y^1 \leftarrow \mathtt{Copy} \leftarrow x^1 :: (y^1: \mathsf{nat})}\] is by definition both left traceable and right traceable, but it is only a left $\mu$-trace and not a right $\nu$-trace: the highest priority of a fixed point applied on the left-hand side on this path belongs to a positive type; this application of the $\mu L$ rule added $x^1_1<x^0_1$ to the set defining the $<$ order. However, there is no negative fixed point rule applied on the right, and $y_1^1$ and $y_1^0$ are incomparable to each other.

 This cycle satisfies the {\it guard condition} by being a left $\mu$-trace. We showed in Example~\ref{ex:copyalg} that it is also accepted by our algorithm since $list(x^1,y^1)=[(x^1_1,y^1_1)]<[(x^0_1,y^0_1)]=list(x^0,y^0)$.

Here, we can observe that being a left $\mu$-trace coincides with having the relation $x^1_1<x^0_1$ between the left channels, and not being a right $\nu$-trace coincides with not having the relation $y^1_1<y^0_1$ for the right channels. We can generalize this observation to every path and every signature with $n$ priorities.

\begin{definition}
Consider a signature $\Sigma$ and a channel $x^\gamma$. We define the list $[x^\gamma]$ as $[x_1^\gamma, \cdots, x_n^\gamma]$.
\end{definition}

\begin{theorem}\label{listor}
A cycle $\mathbb{C}$ 
\[{\infer{\bar{z}^\alpha: \omega \vdash_{\Omega} w^\beta \leftarrow X \leftarrow \bar{z}^\alpha :: (w^{\beta}:C)}{ \infer{\vdots}{{\bar{x}^\gamma: \omega' \vdash_{\Omega'} y^\delta \leftarrow X \leftarrow \bar{x}^\gamma :: (y^{\delta}:C')}}}  }\]
on a program $\mathcal{Q}=\langle V, S\rangle$ defined over Signature $\Sigma$ is a {\it left $\mu$-trace} if $\bar{x}$ and $\bar{z}$ are non-empty and the list $[x^\gamma]=[x_1^\gamma, \cdots, x_n^\gamma]$ is lexicographically less than the list $[z^\alpha]=[z_1^\alpha, \cdots, z_n^\alpha]$ by the order $<_{\Omega'}$ built in $\Omega'$.
Dually, it is a {\it right $\nu$-trace}, if the list $[y^\delta]=[y_1^\delta, \cdots, y_n^\delta]$ is lexicographically less than the list $[w^\beta]=[w_1^\beta, \cdots, w_n^\beta]$ by the strict order $<_{\Omega'}$ built in $\Omega'$
\end{theorem}
\begin{proof}
This theorem is a corollary of Lemmas~\ref{ordtrace} and~\ref{dordtrace}.
\end{proof}

\begin{lemma}
\label{ordtrace}
Consider a path $\mathbb{P}$ in the (infinite) typing derivation of a program $\mathcal{Q}=\langle V, S\rangle$ defined on a signature $\Sigma$,
\[{\infer{z^\alpha: \omega \vdash_{\Omega} P :: (w^{\beta}:C)}{ \infer{\vdots}{{x^\gamma: \omega' \vdash_{\Omega'} P' :: (y^{\delta}:C') }}}  }\]
with $n$ the maximum priority in $\Sigma$.
\begin{enumerate}
    \item[(a)] For every $i  \in \mathtt{c}(\omega')$ with $\epsilon(i)=\mu$, if  $x^\gamma_i \le_{\Omega'} z^\alpha_i$ 
    then $x=z$ and $i \in \mathtt{c}(\omega) $.
    \item[(b)] For every $i<n$, if $ x^\gamma_i <_{\Omega'} z^\alpha_i$, then $i \in \mathtt{c}(\omega)$ and a $\mu L$ rule with priority $i$ is applied on $\mathbb{P}$.
    \item[(c)] For every $c \le n$ with $\epsilon(c)=\nu$, if $ x^\gamma_c \le_{\Omega'} z^\alpha_c$ , then no $\nu L$ rule with priority $c$ is applied on $\mathbb{P}$.
\end{enumerate}
\end{lemma}
\begin{proof}
 The complete proof by induction on the structure of $\mathbb{P}$ is given in Appendix~\ref{proofapp}.
\end{proof}
\begin{lemma}\label{dordtrace}
Consider a path $\mathbb{P}$ in the (infinite) typing derivation of a program $\mathcal{Q}=\langle V, S\rangle$ defined on a Signature $\Sigma$,
\[{\infer{\bar{z}^\alpha: \omega \vdash_{\Omega} P ::(w^{\beta}: C)}{ \infer{\vdots}{{\bar{x}^\gamma: \omega' \vdash_{\Omega'} P' :: ( y^{\delta}:C') }}}  }\]
with $n$ the maximum priority in $\Sigma$.
\begin{enumerate}
    \item[(a)] For every $i \in \mathtt{c}(\omega')$ with $\epsilon(i)=\nu$, if $ y^{\delta}_i \le_{\Omega'} w^{\beta}_i$, then $y=w$ and $i \in \mathtt{c}(\omega)$.
    \item[(b)] If $ y^{\delta}_i<_{\Omega'} w_i^{\beta}$, then $i \in \mathtt{c}(\omega) $ and a $\nu L$ rule with priority $i$ is applied on $\mathbb{P}$ .
    \item[(c)] For every $c\le n$ with $\epsilon(c)=\mu$, if $y^{\delta}_c \le_{\Omega'} w^{\beta}_c$, then no $\mu R$ rule with priority $c$ is applied on $\mathbb{P}$ .
\end{enumerate}
\end{lemma}
\begin{proof}
Dual to the proof of Lemma~\ref{ordtrace} given in Appendix~\ref{proofapp}.
\end{proof}

To illustrate Theorem~\ref{listor}, we present a few additional examples. The reader may skip the examples to get more directly to the main theorems (Lemma~\ref{algtoinf} and Theorem~\ref{thm:main}).

Define a program $\mathcal{P}_9:=\langle \{\mathtt{Succ},\mathtt{Copy},\mathtt{SuccCopy}\}, \mathtt{SuccCopy} \rangle$, over the signature $\Sigma_1$, using the process $\ w:\mathsf{nat} \vdash \mathtt{Copy} :: (y:\mathsf{nat})$ and two other processs: 
$x:\mathsf{nat} \vdash \mathtt{Succ} :: (w:\mathsf{nat})$ and $\ x:\mathsf{nat} \vdash \mathtt{SuccCopy}:: (y:\mathsf{nat})$. The processes are defined as
\begin{align*}
w \leftarrow \mathtt{Succ} \leftarrow x= Rw.\mu_{nat};Rw.s; w\leftarrow x 
\end{align*}
\begin{align*}
 y \leftarrow \mathtt{Copy} \leftarrow w= \mathbf{case}\, Lw\ (\mu_{nat} \Ra \mathbf{case}\, Lw\ & (\ \mathit{s} \Ra Ry.\mu_{nat}; Ry.\mathit{s}; y \leftarrow \mathtt{Copy} \leftarrow w\\
 & \mid \mathit{z} \Rightarrow Ry.\mu_{nat}; Ry.\mathit{z}; \mathbf{wait}\, Lw;\mathbf{close}\, Ry)) 
\end{align*}
\begin{align*}
y \leftarrow \mathtt{SuccCopy}\leftarrow x= w \leftarrow \mathtt{Succ} \leftarrow x ; y \leftarrow \mathtt{Copy} \leftarrow w,
\end{align*}
Process $\mathtt{SuccCopy}$ spawns a new process $\mathtt{Succ}$ and continues as $\mathtt{Copy}$. The $\mathtt{Succ}$ process prepends an $\mathit{s}$ label to the beginning of the finite string representing a natural number on its left hand side and then forwards the string as a whole to the right. $\mathtt{Copy}$ receives this finite string representing a natural number,  and \emph{forwards} it to the right label by label.

The only recursive process in this program is $\mathtt{Copy}$. So program $\mathcal{P}_9$, itself, does not have a further interesting point to discuss. We consider a bogus version of this program in Example~\ref{bogus} that provides further intuition for Theorem~\ref{listor}.

\begin{example}\label{bogus}
Define program $\mathcal{P}_{10}:=\langle \{\mathtt{Succ},\mathtt{BogusCopy},\mathtt{SuccCopy}\}, \mathtt{SuccCopy} \rangle$ over the signature 
\begin{align*}
 \Sigma_1 := \mathsf{nat} =^{1}_{\mu}\oplus \{\mathit{z}:1,\ \  \mathit{s}: \mathsf{nat}\},
\end{align*}
The processes
 $x:\mathsf{nat} \vdash \mathtt{Succ} :: (w:\mathsf{nat})$, $\ \ w:\mathsf{nat} \vdash \mathtt{BogusCopy} :: (y:\mathsf{nat})$, and $\ x:\mathsf{nat} \vdash \mathtt{SuccCopy}:: (y:\mathsf{nat})$,
 are defined as
 \begin{align*}
w \leftarrow \mathtt{Succ} \leftarrow x= Rw.\mu_{nat};Rw.s; w \leftarrow x 
\end{align*}
\begin{align*}
 y\leftarrow \mathtt{BogusCopy}\leftarrow w = \mathbf{case} Lw\, (\mu_{nat} \Ra \mathbf{case} Lw\, & (\mathit{s}\Rightarrow Ry.\mu_{nat}; Ry.\mathit{s}; y \leftarrow \mathtt{SuccCopy} \leftarrow w\\
 & \mid\mathit{z} \Rightarrow Ry.\mu_{nat}; Ry.\mathit{z}; \mathbf{wait}\, Lw;\mathbf{close}\, Ry)) 
\end{align*}
\begin{align*}
y \leftarrow \mathtt{SuccCopy} \leftarrow x = w \leftarrow \mathtt{Succ} \leftarrow x ; y \leftarrow \mathtt{BogusCopy}\leftarrow w
\end{align*}
Program $\mathcal{P}_{10}$ is a non-reactive \emph{bogus} program, since $\mathtt{BogusCopy}$ instead of calling itself recursively, calls $\mathtt{SuccCopy}$.  At the very beginning $\mathtt{SuccCopy}$ spawns $\mathtt{Succ}$ and continues with $\mathtt{BogusCopy}$ for a fresh channel $w$. $\mathtt{Succ}$ then sends a fixed point unfolding message and a {\it successor} label via $w$ to the right, while $\mathtt{BogusCopy}$ receives the two messages just sent by $\mathtt{Succ}$ through $w$ and calls $\mathtt{SuccCopy}$ recursively again. This loop continues forever, without any messages being received from the outside.

The first several steps of the derivation of $x^0:\mathsf{nat} \vdash_{\emptyset} \mathtt{SuccCopy}:: (y^0:\mathsf{nat})$ in our inference system (Section~\ref{rules}) are given below.

\setlength{\inferLineSkip}{5pt} 
\setlength{\inferLabelSkip}{1pt} 
{\small \[\infer[\msc{Def}]{{\color{red}x^0{:}\mathsf{nat}\vdash_{\emptyset} y^0 \leftarrow \mathtt{SuccCopy} \leftarrow x^0 :: (y^0{:}\mathsf{nat})}}{\infer[\msc{Cut}^w]{x^0{:}\mathsf{nat} \vdash_{\emptyset} w \leftarrow \mathtt{Succ}; y^0 \leftarrow \mathtt{BogusCopy} \leftarrow w:: (y^0{:}\mathsf{nat})}{\infer[\msc{Def}]{x^0{:}\mathsf{nat} \vdash_{\emptyset} {w^0} \leftarrow \mathtt{Succ}\leftarrow x^0:: (w^0{:} \mathsf{nat})}{\infer[\mu R]{x^0{:}\mathsf{nat} \vdash_{\emptyset} Rw^{0}.\mu_{nat}; \cdots :: (w^0{:} \mathsf{nat})}{\infer[\oplus R]{x^0{:}\mathsf{nat} \vdash_{\emptyset}  Rw^1.\mathit{s};\cdots :: (w^1{:}1 \oplus \mathsf{nat})}{\infer[\msc{Id}]{x^0{:}\mathsf{nat} \vdash_{\emptyset} w^1 \leftarrow x^0 :: (w^1{:} \mathsf{nat})}{}}}} \hspace{-16pt} & {\infer[\msc{Def}]{ w^0{:} \mathsf{nat}\vdash_{\emptyset} y^0 \leftarrow \mathtt{BogusCopy} \leftarrow w^0:: (y^0{:}\mathsf{nat})}{\infer[\mu L]{w^0{:} \mathsf{nat}\vdash_{\emptyset} \mathbf{case}\, Lw^0\ (\mu_{nat} \Ra \cdots) ::(y^0{:}\mathsf{nat})}{\infer[\oplus L]{w^1{:} 1\oplus \mathsf{nat}\vdash_{\{w^1_1<w^0_1\}} \mathbf{case}\, Lw^1\ (\cdots) ::(y^0{:}\mathsf{nat})}{\boldsymbol{\cdots} & {\color{blue} w^1{:}\mathsf{nat} \vdash_{\{w^1_1<w^0_1\}} y^0 \leftarrow \mathtt{SuccCopy}\leftarrow w^1 ::(y^0{:}\mathsf{nat})}}}}}}}\]}

\noindent
Consider the cycle between
\[\color{red}x^0:\mathsf{nat}\vdash_{\emptyset} y^0 \leftarrow \mathtt{SuccCopy} \leftarrow x^0 :: (y^0:\mathsf{nat})\]

\noindent and 
\[\color{blue}w^1:\mathsf{nat}\vdash_{\{w^1_1 < w^0_1\}} y^0 \leftarrow \mathtt{SuccCopy} \leftarrow w^1 :: (y^0:\mathsf{nat}).\]

\noindent
By Definition~\ref{tr}, this path is right traceable, but not left traceable. And by Definition~\ref{rt}, the path is neither a right $\nu$-trace nor a left $\mu$-trace:
\begin{enumerate}
\item No negative fixed point unfolding message is received from the right and $y^0$ does not evolve to a new generation that has a smaller value in its highest priority than $y^0_1$. In other words, $y^0_1 \not < y^0_1$ since no negative fixed point rule has been applied on the right channel.
\item The positive fixed point unfolding message that is received from the left is received through the channel $w^0$, which is a fresh channel created after $\mathtt{SuccCopy}$ spawns the process $\mathtt{Succ}$. Although $w^1_1<w^0_1$, since $x^0_1$ is incomparable to $w^0_1$, the relation $w^1_1<x^0_1$ does not hold. This path is not even a left-traceable path.
\end{enumerate}
Neither $[w^1]= [w^1_1] < [x^0_1]= [x^0]$, nor  $[y^0]= [y^0_1] < [y^0_1] =[y^0]$ hold, and this cycle does not satisfy the guard condition. This program is not locally valid either since $[w^1_1,y^0_1]\not < [x^0_1, y^0_1]$.
\end{example}

As another example consider the program
$\mathcal{P}_6=\langle\{\mathtt{Ping},\mathtt{Pong},\mathtt{PingPong}\},  \mathtt{PingPong} \rangle$ over the signature $\Sigma_4$ as defined in Example~\ref{begen1}. We discussed in Section~\ref{priority} that this program is not accepted by our algorithm as locally valid.

\begin{example}
Recall the definition of signature $\Sigma_4$:
\[\begin{aligned}
 \Sigma_4 :=\  & \mathsf{ack}=^{1}_{\mu} \oplus\{ \mathit{ack}:\mathsf{astream}\},\\
& \mathsf{astream}=^{2}_{\nu} \& \{\mathit {head}: \mathsf{ack}, \ \ \mathit{tail}: \mathsf{astream}\},\\
& \mathsf{nat}=^{3}_{\mu}\oplus \{\mathit{z}:1,\ \  \mathit{s}: \mathsf{nat}\}
 \end{aligned}\]
Processes
\[
  \begin{array}{l}
  x:\mathsf{nat} \vdash \mathtt{Ping} :: (w:\mathsf{astream}),\\
  w:\mathsf{astream} \vdash \mathtt{Pong} :: (y:\mathsf{nat}),\\
  x:\mathsf{nat} \vdash \mathtt{PingPong}:: (y:\mathsf{nat})
  \end{array}
\]
are defined as
\begin{align*}
 w \leftarrow \mathtt{Ping} \leftarrow x= \mathbf{case}\, Rw\, (\nu_{astream} \Ra \mathbf{case}\, Rw\, & ( \mathit{head} \Rightarrow Rw.\mu_{ack}; Rw.\mathit{ack}; w\leftarrow \mathtt{Ping} \leftarrow x\\ 
 & \mid \mathit{tail} \Rightarrow w \leftarrow \mathtt{Ping} \leftarrow x)) 
\end{align*}
\begin{align*}
  y \leftarrow \mathtt{Pong} \leftarrow w = & Lw.\nu_{astream};Lw.\mathit{head}; \\
  &\quad \mathbf{case}\, Lw\ (\mu_{ack} \Ra \mathbf{case}\, Lw\ (ack \Rightarrow Ry.\mu_{nat}; Ry.s;  y \leftarrow \mathtt{Pong} \leftarrow w))
\end{align*}
\begin{align*}
 y \leftarrow \mathtt{PingPong} \leftarrow x= w \leftarrow \mathtt{Ping}\leftarrow x ; y \leftarrow \mathtt{Pong} \leftarrow w
\end{align*}
 The first several steps of the proof of $x^0:nat \vdash_{\emptyset}  \mathtt{PingPong}:: (y^0:nat)$ in our inference system (Section~\ref{rules}) are given below (with some abbreviations).

\setlength{\inferLineSkip}{5pt} 
\setlength{\inferLabelSkip}{1pt}
{\small\[\infer[\msc{Def}]{{x^0{:}\mathsf{nat}\vdash_{\emptyset} y^0 \leftarrow \mathtt{PingPong} \leftarrow x^0 :: (y^0{:}\mathsf{nat})}}{ \infer[\msc{Cut}]{ x^0{:}\mathsf{nat} \vdash_{\emptyset} w \leftarrow \mathtt{Ping}\leftarrow x^0 ; y^0 \leftarrow \mathtt{Pong} \leftarrow w:: (y^0{:}\mathsf{nat})}{  \infer[\msc{Def}]{ {\color{red} x^0{:}\mathsf{nat} \vdash_{\emptyset} w^0 \leftarrow \mathtt{Ping}\leftarrow x^0 :: (w^0{:} \mathsf{astream})}}{\infer[\nu R]{x^0:\mathsf{nat} \vdash_{\emptyset} \mathbf{case}\, Rw^0\ (\nu_{astream} \Rightarrow \cdots) :: (w^0{:} \mathsf{astream})}{ \infer[\& R]{x^0{:}\mathsf{nat} \vdash_{A}  \mathbf{case}\, Rw^1\ (\cdots) :: (w^1{:}\mathsf{ack}\  \&\ \mathsf{astream})}{ \infer[\mu R]{x^0{:}\mathsf{nat} \vdash_{A} Rw^1. \mu_{ack}; \cdots ::(w^1{:}\mathsf{ack})}{ \infer[\oplus R]{x^0{:}\mathsf{nat} \vdash_{B} Rw^2. \mathit{ack}; \cdots ::(w^2{:}\oplus\{\mathsf{astream}\})}{ \color{blue} x^0{:}\mathsf{nat} \vdash_{B} w^2 \leftarrow \mathtt{Ping} \leftarrow x^0::(w^2{:}\mathsf{astream})}}\hspace{-25pt} & { x^0{:}\mathsf{nat} \vdash_{A} \boldsymbol{\cdots} ::(w^1{:}\mathsf{astream})}}}} \hspace{-53pt} & { w^0{:} \mathsf{astream}\vdash_{\emptyset} \boldsymbol{\cdots}:: (y^0{:}\mathsf{nat})}}}\]}

\noindent
where $A= {\{w^1_1=w^0_1, w^1_2<w^0_2, w^1_3=w^0_3\}}$, and
$B=\{w^1_1=w^0_1, w^2_2=w^1_2<w^0_2, w^2_3=w^1_3=w^0_3\}$.
The cycle between the processes \[{\color{red}x^0:\mathsf{nat} \vdash_{\emptyset}w^0 \leftarrow \mathtt{Ping} \leftarrow {x^0} :: (w^0: \mathsf{astream})}\] and 
  \[ {\color{blue} x^0:\mathsf{nat} \vdash_{B} w^2 \leftarrow \mathtt{Ping} \leftarrow x^0 ::(w^2:\mathsf{astream})}\] is neither a left $\mu$-trace, nor a right $\nu$-trace:
  \begin{enumerate}
    \item No fixed point unfolding message is received or sent through the left channels in this path and thus $[x^0]= [x_1^0,x_2^0,x_3^0]\not < [x_1^0,x_2^0,x_3^0]=[x^0]$.
    \item On the right, fixed point unfolding messages are both sent and received: (i) $w^0$ receives an unfolding message for a negative fixed point with priority $2$ and evolves to $w^1$, and then later (ii) $w^1$ sends an unfolding message for a positive fixed point with priority $1$ and evolves to $w^2$. But the positive fixed point has a higher priority than the negative fixed point, and thus this path is not a right $\nu$-trace either.
\end{enumerate}
This reasoning can also be reflected in our observation about the list of channels in Theorem~\ref{listor}: When, first, $w^0$ evolves to $w^1$ by receiving a message in (i) the relations $w^1_1=w^0_1$, $w^1_2<w^0_2$, and $w^1_3=w^0_3$ are recorded. And, later, when $w^1$ evolves to $w^2$ by sending a message in (ii) the relations  $w^2_2=w^1_2$, and $w^2_3=w^1_3$ are added to the set. This means that $w^2_1$ as the first element of the list $[w^2]$ remains incomparable to $w^0_1$ and thus  $[w^2]= [w_1^2,w_2^2,w_3^2] \not < [w_1^0,w_2^0,w_3^0]=[w^0]$.
\end{example}

We are now ready to state our main theorem that proves the local validity algorithm introduced in Section~\ref{alg} is {\it stricter} than the FS guard condition. Since the guard condition is defined over an infinitary system, we need to first map our local condition into the infinitary calculus given in Section~\ref{fig:stp-order}. 

\begin{lemma}\label{algtoinf}
Consider a finitary derivation (Figure~\ref{fig:validity}) for \[\langle\bar{u}, X, v\rangle; \bar{x}^\alpha: \omega \vdash_{\Omega,\subset} P ::(y^\beta:C),\] on a locally valid program $\mathcal{Q}= \langle V, S \rangle$ defined on signature $\Sigma$ and order $\subset$. There is a (potentially infinite) derivation $\mathbb{D}$ for 
\[\bar{x}^\alpha: \omega \vdash_{\Omega} P ::(y^\beta:C),\]
in the infinitary system of Figure~\ref{fig:stp-order}.

Moreover, for every $\bar{w}^{\gamma}:\omega' \vdash_{\Omega'} z^{\delta} \leftarrow Y \leftarrow \bar{w}^{\gamma} :: (z^{\delta}:C')$ in $\mathbb{D}$, we have  \[Y,list(\bar{w}^{\gamma},z^{\delta}) \mathrel{(\subset, <_{\Omega'})} X,list(\bar{u},v).\]
\end{lemma}

\begin{proof}
This lemma is a special case of Lemma~\ref{lem:coinduction} proved in Appendix~\ref{proofapp}.
\end{proof}

\begin{theorem}\label{thm:main}
A locally valid program satisfies the FS guard condition.
\end{theorem}
\begin{proof}
Consider a cycle $\mathbb{C}$ on a (potentially infinite) derivation produced from $\langle \bar{u}, Y, v \rangle; \bar{z}^\alpha: \omega \vdash_{\Omega} w^\beta \leftarrow  X \leftarrow \bar{z}^\alpha :: (w^{\beta}:C)$ as in Lemma~\ref{algtoinf},
\[\infer[\msc{Def}]{\bar{z}^\alpha: \omega \vdash_{\Omega} w^\beta \leftarrow  X \leftarrow \bar{z}^\alpha :: (w^{\beta}:C)}{\infer{\bar{z}^\alpha: \omega \vdash_{\Omega} P_{\bar{z}^\alpha, w^\beta} :: (w^{\beta}:C)}{ \infer{\vdots}{\infer[\msc{Def}]{\bar{x}^\gamma: \omega \vdash_{\Omega'} y^\delta \leftarrow X \leftarrow \bar{x}^\gamma :: (y^{\delta}:C) }{\bar{x}^\gamma: \omega \vdash_{\Omega'} P_{\bar{x}^\gamma, y^\delta} :: (y^{\delta}:C) & \bar{z}:\omega \vdash X=P_{\bar{z},w} :: (w:C) \in V }}} &\hspace{-90pt} \bar{z}:\omega \vdash X=P_{\bar{z},w} :: (w:C) \in V}\]
By Lemma~\ref{algtoinf} we get 
\[X, list(\bar{x}^\gamma, y^{\delta})\ (\subset, <_{\Omega'})\ X, list(\bar{z}^\alpha,w^{\beta}),\] and thus by definition of $(\subset, <_{\Omega'})$, 
\[list (\bar{x}^\gamma, y^{\delta}) <_{\Omega'}  list (\bar{z}^\alpha,w^{\beta}).\]
Therefore, there is an $i\le n$, such that either 
\begin{enumerate}
    \item $\epsilon(i)=\mu,$ $x^\gamma_i < z^\alpha_i,$ and $x^\gamma_l = z^\alpha_l$ for every $l<i$, having that $\bar{x}=x$ and $\bar{z}=z$ are non-empty, or 
    \item $\epsilon(i)=\nu,$ $y^{\delta}_i < w^{\beta}_i,$ and $y^{\delta}_l = w^{\beta}_l$ for every $l<i$.
\end{enumerate}
In the first case, by part (b) of Lemma~\ref{ordtrace}, a $\mu L$ rule with priority $i \in \mathtt{c}(\omega)$ is applied on $\mathbb{C}$. By part (a) of the same Lemma $x=z$, and by its part (c), no $\nu L$ rule with priority $c<i\ $ is applied on $\mathbb{C}$. Therefore, $\mathbb{C}$ is a left $\mu$- trace. \\
In the second case, by part (b) of Lemma~\ref{dordtrace}, a $\nu R$ rule with priority $i \in \mathtt{c}(\omega)$ is applied on $\mathbb{C}$. By part (a) of the same Lemma $y=w$ and by its part (c), no $\mu R$ rule with priority $c<i\ $ is applied on $\mathbb{C}$. Thus, $\mathbb{C}$ is a right $\nu$- trace. 
\end{proof}

By Theorem~\ref{listor}, a cycle $\mathbb{C}$ 
\[{\infer{\bar{z}^\alpha: \omega \vdash_{\Omega}w^\beta \leftarrow X \leftarrow \bar{z}^\alpha :: (w^{\beta}:C)}{ \infer{\vdots}{{\bar{x}^\gamma: \omega' \vdash_{\Omega'} y^\delta \leftarrow X \leftarrow \bar{x}^\gamma :: (y^{\delta}:C')}}}  }\]
is either a {\it left $\mu$-trace} {or} a {\it right $\nu$-trace} if either $[x^\gamma]<_{\Omega'} [z^\alpha]$ {\bf or} $[y^\delta]<_{\Omega'} [w^\beta]$ holds. Checking a disjunctive condition for each cycle implies that the FS guard condition cannot simply analyze each path from the beginning of a definition to a call site in isolation and then compose the results---instead, it must unfold the definitions and examine every cycle, possibly composed of smaller individual cycles, in the infinitary derivation separately. 
In other words, the FS guard condition may accept two individual cycles but reject their combination.

In our algorithm, however, we form a transitive condition by merging the lists of left and right channels, e.g. $[x^\gamma]$ and $[y^\delta]$ respectively, into a single list $list(x^{\gamma},y^{\delta})$. The values in  $list(x^{\gamma},y^{\delta})$ from Definition~\ref{list} are still recorded in their order of priorities, but for the same priority the value corresponding to {\it receiving} a message precedes the one corresponding to {\it sending} a message. As described in Definition~\ref{defmain} we merge this list with process variables to check all immediate calls even those that do not form a cycle in the sense of the FS guard condition (that is, when process $X$ calls process $Y\neq X$).

{Transitivity of our validity check condition is the key to establishing its locality. Our algorithm only checks the condition for the immediate calls that a process makes. As this condition enjoys transitivity, it also holds for all possible non-immediate recursive calls, including any combination of the immediate calls. As a result, we do not need to search for every possible cycle in the infinitary derivation.}

\begin{remark}\label{complexity}
We briefly analyze the asymptotic complexity of our algorithm.
Let $n$ be the number of priorities and $s$ the size of the signature, where we add in the sizes of all types $A$ appearing in applications of the $\msc{Cut}$ rule.  In time $O(n\,s)$ we can compute a table to look up $i \in c(A)$ for all priorities $i$ and types $A$ appearing in cuts.

Now let $m$ be the size of the program (not counting the signature).
We traverse each process definition just once, maintaining a list of relations between the current and original channel pairs for each priority.  We need to update at most $2n$ entries in the list at each step and compare at most $2n$ entries at each $\msc{Call}$ rule.  Furthermore, for each $\msc{Cut}$ rule we have a constant-time table lookup to determine if $i \in c(A)$ for each priority $i$.
Therefore, analysis of the process definitions takes time $O(m\, n)$.

Putting it all together, the time complexity is bounded by $O(m\, n + n\, s) = O(n\, (m+s))$.
In practice the number of priorities, $n$, is a small constant so validity checking is linear in the total input, which consists of the signature and the process definitions.  As far as we are aware of, the best upper bound for the complexity of the FS guard condition is $\mathit{PSPACE}$~\cite{Doumane17phd}.

It is also interesting to note that the complexity of type-checking itself is bounded below by $O(m + s^2)$ since, in the worst case, we need to compute equality between each pair of types.  That is, validity checking is faster than type-checking.
\end{remark}

Another advantage of locality derives from the fact that our algorithm checks each process definition independent of the rest of the program: we can safely reuse a previously checked locally valid process in other programs defined over the same signature and order $\subset$ without the need to verify its local validity again. 

\section{Computational Meta-theory}
\label{semantics}

Fortier and Santocanale~\cite{Fortier13csl} defined a function $\textsc{Treat}$  as a part of their cut elimination algorithm. They proved that this function terminates on a list 
of pre-proofs fused by consecutive cuts if all of them satisfy their guard condition. 
 In our system, function $\textsc{Treat}$ corresponds to computation on a configuration of processes. In this section we first show that the usual preservation and progress theorems hold even if a program does not satisfy the validity condition. Then we use Fortier and Santocanale's result to prove a stronger compositional progress property for (locally) valid programs.
 
In Section~\ref{operat}, we introduced process configurations $\mathcal{C}$ as a list of processes connected by the associative, noncommutative parallel composition operator $\mid_x$. 
\[\mathcal{C} ::= \cdot \mid \mathtt{P} \mid (\mathcal{C}_1 \mid_{x}\ \mathcal{C}_2)\]
with unit $(\cdot)$.
The type checking judgments for configurations
$\bar{x}: \omega \Vdash \mathcal{C}:: (y:B)$ are:\\
\[
\begin{tabular}{c c c c c}
\infer[]{x:A \Vdash \cdot :: (x:A) }{} & &
\infer[]{\bar{x}:\omega \Vdash P :: (y:B)}{\bar{x}:\omega \vdash P :: (y:B)} &
& \infer[]{\bar{x}:\omega \Vdash C_1 |_{\ z} \ C_2 :: (y:B)}{\bar{x}:\omega \Vdash C_1:: (z:A) & z:A \Vdash C_2 :: (y:B)}\\\\
\end{tabular}
\]
A configuration can be read as a list of processes connected by consecutive cuts. Alternatively, considering $\mathcal{C}_1$  and $\mathcal{C}_2$ as two processes, configuration $\mathcal{C}_1 \mid_{z} \mathcal{C}_2$ can be read as their composition by a cut rule $(z \leftarrow \mathcal{C}_1; \mathcal{C}_2)$. In section~\ref{operat}, we defined an operational semantics on configurations using transition rules. Similarly, these computational transitions can be interpreted as cut reductions called ``internal operations'' by Fortier and Santocanale. The usual preservation theorem ensures types of a configuration are preserved during computation \cite{DeYoung16aplas}.

\begin{theorem}\label{preservation} (Preservation)
For a configuration  $\bar{x}:\omega \Vdash \mathcal{C} :: (y:A) $, if $\mathcal{C} \mapsto \mathcal{C}'$ by one step of computation, then $\bar{u}:\omega \Vdash \mathcal{C}' : (w:A)$ for some channels $\bar{u}$ and $w$.
\end{theorem}
\begin{proof}
This property follows directly from the correctness of cut reduction steps.
\end{proof}
The usual progress property as proved below ensures that computation makes progress or it attempts to communicate with an external process.
\begin{theorem}(Progress)
If $\bar{x}:\omega \Vdash \mathcal{C}:: (y:A)$, then either
\begin{enumerate}
    \item $\mathcal{C}$ can make a transition,
    \item or $\mathcal{C}=  (\cdot)$ is empty,
    \item or $\mathcal{C}$ attempts to communicate either to the left or to the right.
\end{enumerate}
\end{theorem}
\begin{proof}
The proof is by structural induction on the configuration typing.
\end{proof}
In the presence of (mutual) recursion, this progress property is not strong enough to ensure that a program does not get stuck in an infinite inner loop. Since our local validity condition implies the FS guard condition, we can use their results for a stronger version of the progress theorem on valid programs. 

\begin{theorem}\label{progress}(Strong Progress)
Configuration $\bar{x}:\omega \Vdash \mathit{C}:: (y:A)$ of (locally) valid processes satisfies the progress property. Furthermore, after a finite number of steps, either
\begin{enumerate}
    \item $\mathcal{C}= (\cdot)$ is empty,
    \item or $\mathcal{C}$ attempts to communicate to the left or right.
\end{enumerate}
\end{theorem}

\begin{proof}
There is a correspondence between the $\textsc{Treat}$ function's internal operations and the computational transitions introduced in Section~\ref{operat}. The only point of difference is the extra computation rule we introduced for the constant $1$. Fortier and Santocanale's proof of termination of the function $\textsc{Treat}$ remains intact after extending \textsc{Treat}'s primitive operation with a reduction rule for the constant $1$, since this reduction step only introduces a new way of closing a process in the configuration.  Under this correspondence, termination of the function $\textsc{Treat}$ on valid proofs implies the strong progress property for valid programs.
\end{proof}

As a corollary to Theorem~\ref{progress}, computation of a closed valid program $\mathcal{P}= \langle V,S\rangle$ with $\cdot \vdash S=P :: (y:1)$ always terminates by closing the channel $y$ (which follows by inversion on the typing derivation).

We conclude this section by briefly revisiting sources of invalidity in computation. In Example~\ref{begen} we saw that process $\mathtt{Loop}$ is not valid, even though its proof is cut-free. Its computation satisfies the strong progress property as it attempts to communicate with its right side in finite number of steps. However, its communication with left and right sides of the configuration is solely by sending messages. Composing $\mathtt{Loop}$ with any process $y:\mathsf{nat} \vdash \mathtt{P} :: (z:1)$ results in exchanging an infinite number of messages between them. For instance, for ${\mathtt{Block}}$, introduced in Example~\ref{block}, the infinite computation of $\cdot \Vdash  y \leftarrow \mathtt{Loop} \mid_{\, y} z \leftarrow \mathtt{Block} \leftarrow y :: (z:1)$ without communication along $z$ can be depicted as follows:
\begin{center}
\begin{tabular}{c l l }
& $y \leftarrow \mathtt{Loop} \mid_{\, y} z \leftarrow \mathtt{Block} \leftarrow y$ & $\mapsto$\\ &
$Ry.\mu_{\mathsf{nat}}; Ry.\mathit{s}; y \leftarrow \mathtt{Loop} \mid_{\, y} z \leftarrow \mathtt{Block} \leftarrow y$ & $  \mapsto$\\
& $Ry.\mu_{\mathsf{nat}}; Ry.\mathit{s}; y \leftarrow \mathtt{Loop} \mid_{\, y} \mathbf{case}\, Ly\ (\mu_{\mathsf{nat}}\Ra \mathbf{case}\, Ly\ \cdots)$ & $\mapsto$\\ & $Ry.s; y \leftarrow \mathtt{Loop} \mid_{\, y} \mathbf{case}\, Ly\ (\mathit{s} \Rightarrow  z \leftarrow \mathtt{Block} \leftarrow y \ \mid \mathit{z} \Rightarrow \mathbf{wait}\, Ly;\mathbf{close}\, Rz)$&$ \mapsto$\\
& $y \leftarrow \mathtt{Loop} \mid_{\, y} z \leftarrow \mathtt{Block} \leftarrow y$ & $\mapsto$\\ & $\mathbf{\cdots}$\\

\end{tabular}
\end{center}

In this example, the strong progress property of computation is violated. The configuration does not communicate to the left or right and a never ending series of internal communications takes place. This internal loop is a result of the infinite number of unfolding messages sent by $\mathtt{Loop}$ without any unfolding message with higher priority being received by it. In other words, it is the result of $\mathtt{Loop}$  not being valid.

\section{Incompleteness of Validity Conditions}
\label{negative}

In this section we provide a straightforward example of a program with the strong progress property that our algorithm cannot identify as valid. Intuitively, this program seems to preserve the strong progress property after being composed with other valid programs. We show that this example does not satisfy the FS guard condition, either. 

\begin{example}\label{ex:counter}
Define the signature
\[\begin{aligned}
 \Sigma_5 :=\ & \mathsf{ctr}=^{1}_{\nu} \& \{\mathit {inc}: \mathsf{ctr}, \ \ \mathit{val}: \mathsf{bin}\},\\
 & \mathsf{bin}=^{2}_{\mu} \oplus\{ \mathit{b0}:\mathsf{bin}, \mathit{b1}:\mathsf{bin}, \mathit{\$}:\mathsf{1}\}\\
 \end{aligned}\]
and program $\mathcal{P}_{11}= \langle \{\mathtt{Bit0Ctr},\mathtt{Bit1Ctr}, \mathtt{Empty}\}, \mathtt{Empty} \rangle,$ where
\[
\begin{array}{l}
x:\mathsf{ctr} \vdash y \leftarrow \mathtt{Bit0Ctr} \leftarrow x :: (y:\mathsf{ctr})\\
x:\mathsf{ctr} \vdash y \leftarrow \mathtt{Bit1Ctr}\leftarrow x :: (y:\mathsf{ctr})\\
\cdot \vdash y \leftarrow \mathtt{Empty}:: (y:\mathsf{ctr})
\end{array}
\] 
with
\begin{small}
\begin{align*} 
 & y^{\beta} \leftarrow \mathtt{Bit0Ctr}\leftarrow x^{\alpha}= & {\color{blue}[0, 0,\ 0,\ 0]} \\
 &\phantom{s} \mathbf{case}\, Ry^{\beta}\ (\nu_{ctr} \Ra & [-1,0,0,0] &  \phantom{L} y^{\beta+1}_1<y^\beta_1,\\[-4pt]
 &&& \phantom{L} y^{\beta+1}_2 = y^{\beta}_2\\[-10pt]  &  \phantom{sm}     \mathbf{case}\, Ry^{\beta+1} ( \mathit{inc} \Rightarrow 
 y^{\beta+1} \leftarrow \mathtt{Bit1Ctr} \leftarrow x^{\alpha} & {\color{red} [-1,\ 0,\ 0 ,0]}  &\phantom{L.}\\
&  \phantom{smal}    \mid \mathit{val} \Rightarrow Ry^{\beta+1}.\mu_{bin};Ry^{\beta+2}.\mathit{b0};Lx^\alpha.\nu_{ctr}; Lx^{\alpha+1}.val; y^{\beta+2} \leftarrow x^{\alpha+1})) & { [-1,1,0,1]} & \\
\end{align*}
\begin{align*} 
 & y^{\beta} \leftarrow \mathtt{Bit1Ctr}\leftarrow x^{\alpha}= &{\color{blue}[0, 0,\ 0,\ 0]} &  \phantom{L}\\
 &\phantom{s} \mathbf{case}\, Ry^{\beta}\ (\nu_{ctr} \Ra & [-1,0,0,0] &  \phantom{L} y^{\beta+1}_1<y^\beta_1,\\[-4pt] &&&\phantom{L} y^{\beta+1}_2 = y^{\beta}_2\\[-2pt]  &  \phantom{sm}     \mathbf{case}\, Ry^{\beta+1}\, ( \mathit{inc} \Rightarrow
 Lx^{\alpha}.\nu_{ctr}; Lx^{\alpha+1}.inc; y^{\beta+1} \leftarrow \mathtt{Bit0Ctr} \leftarrow x^{\alpha+1} & {\color{red} [-1,\ 1,\ 0 ,0]}  &\phantom{L}  x^{\alpha+1}_2 = x^{\alpha}_2\\ 
&  \phantom{smal}    \mid \mathit{val} \Rightarrow Ry^{\beta+1}.\mu_{bin};Ry^{\beta+2}.\mathit{b1};Lx^\alpha.\nu_{ctr}; Lx^{\alpha+1}.val; y^{\beta+2} \leftarrow x^{\alpha+1})) & { [-1,1,0,1]}
\end{align*}
\begin{align*} 
 & y^{\beta} \leftarrow \mathtt{Empty}\leftarrow \cdot = &  {\color{blue}[0, \_,\ \_,\ 0]} &  \phantom{L}\\
 &\phantom{s} \mathbf{case}\, Ry^{\beta}\ (\nu_{ctr} \Ra & [-1,\_,\_,0] &  \phantom{L} y^{\beta+1}_1<y^\beta_1, y^{\beta+1}_2 = y^{\beta}_2\\  &  \phantom{sm} \mathbf{case}\, Ry^{\beta+1}\ ( \mathit{inc} \Rightarrow 
  w^{0} \leftarrow \mathtt{Empty} \leftarrow\cdot ; & {\color{red} [\infty,\ \_,\ \_ ,\infty]}  &\phantom{L.} \mathsf{ctr}, \mathtt{bin} \in \mathtt{c}(\mathsf{ctr})\\ 
   &  \phantom{small space more than} y^{\beta+1} \leftarrow \mathtt{Bit1Ctr} \leftarrow w^{0} & {\color{red} [-1,\ \infty, \infty ,0]}  &\phantom{L.}  \mathsf{ctr}, \mathtt{bin} \in \mathtt{c}(ctr)\\ 
&  \phantom{small space mor}\mid \mathit{val} \Rightarrow Ry^{\beta+1}.\mu_{bin};Ry^{\beta+2}.\$ ; \mathbf{close}\, Ry^{\beta+2})) & { [-1,\_,\_,1]} 
\end{align*}
\end{small}

In this example we implement a counter slightly differently from Example~\ref{cutchannel}. We have two processes $\mathtt{Bit0Ctr}$ and $\mathtt{Bit1Ctr}$ that are holding one bit ($b0$ and $b1$ respectively) and a counter $\mathtt{Empty}$ that signals the end of the chain of counter processes. This program begins with an empty counter (representing value 0). If  a value is requested, then it sends $\$$ to the right and if an increment is requested it adds the counter $\mathtt{Bit1Ctr}$ with $b1$ value to the chain of counters. Then if another increment is asked, $\mathtt{Bit1Ctr}$ sends an increment ($\mathit{inc}$) message to its left counter (implementing the carry bit) and calls $\mathtt{Bit0Ctr}$. If $\mathtt{Bit0Ctr}$ receives an increment from the right, it calls $\mathtt{Bit1Ctr}$ recursively.

All (mutually) recursive calls in this program are recognized as valid by our algorithm, except the one in which $\mathtt{Empty}$ calls itself. In this recursive call, $y^\beta\leftarrow \mathtt{Empty}\leftarrow \cdot$ calls $w^0\leftarrow \mathtt{Empty}\leftarrow \cdot$, where $w$ is the fresh channel it shares with $y^{\beta+1} \leftarrow \mathtt{Bit1Ctr} \leftarrow w^0$.  The number of increment unfolding messages $\mathtt{Bit1Ctr}$ can send along channel $w^0$ are always less than or equal to the number of increment unfolding messages it receives along channel $y^{\beta+1}$. This implies that the number of messages $w^0\leftarrow \mathtt{Empty}\leftarrow \cdot$ may receive along channel $w^0$ is strictly less than the number of messages received by any process along channel $y^{\beta}$. There will be no infinite loop in the program without receiving an unfolding message from the right.  Indeed Fortier and Santocanale's cut elimination for the cut corresponding to the composition $\mathtt{Empty} \mid \mathtt{Bit1Ctr}$ locally terminates. 
Furthermore, since no valid program defined on the same signature can send infinitely many increment messages to the left, $\mathcal{P}_{11}$ composed with any other valid program satisfies strong progress.

This result is also a negative example for the FS guard condition. The path between $y^{\beta} \leftarrow \mathtt{Empty}\leftarrow \cdot$ and $w^{0} \leftarrow \mathtt{Empty}\leftarrow \cdot$ in the $\mathtt{Empty}$ process is neither left traceable not right traceable since $w \neq y$. By Definition~\ref{guard} it is therefore not a valid cycle.
\end{example}

Example~\ref{ex:counter} shows that neither our algorithm nor the FS guard condition are complete.
In fact, using Theorem~\ref{progress} we can prove that no effective procedure, including our algorithm, can recognize a maximal set $\Xi$ of programs with the strong progress property that is closed under composition. 

\begin{theorem}
 It is undecidable to recognize a maximal set $\Xi$ of session-typed programs in subsingleton logic with the strong progress property that is closed under composition.
\end{theorem}
\begin{proof}
Pfenning and DeYoung showed that any Turing machine can be represented as a process in subsingleton logic with equirecursive fixed points \cite{DeYoung16aplas,Pfenning16lectures}, easily embedded into our setting with isorecursive fixed points. It implies that a Turing machine on a given input halts if and only if the closed process representing it terminates. By definition of strong progress, a closed process terminates if and only if it satisfies strong progress property. 
Using this result, we reduce the halting problem to identifying closed programs $\mathcal{P}:=\langle V,S \rangle$ with $\cdot \vdash S :: x:1$ that satisfy strong progress. Note that a closed program satisfying strong progress is in every maximal set $\Xi$.
\end{proof}
\section{Concluding Remarks}
\label{conclusion}

  {\bf Related work.} The main inspiration for this paper is work by Fortier and Santocanale~\cite{Fortier13csl}, who provided a validity condition for pre-proofs in singleton logic with least and greatest fixed points. They showed that valid circular proofs in this system enjoy cut elimination. Circular proofs in singleton logic are interpreted as the winning strategy in parity games~\cite{Santocanale02fossacs}. A winning strategy corresponds to an asynchronous protocol for a deadlock free communication of the players \cite{Joyal96ICLMPS}. The cut elimination result for circular proofs is a ground for reasoning about these protocols of communication and the related categorical concept of $\mu$-bicomplete categories \cite{Santocanale02notes, Santocanale02ita}. Although session types and game semantics both model the concept of message passing based on a protocol \cite{Castellan19POPL}, another line of research is needed to fill the gap between semantics of parity games and recursive session-typed processes. Also related is work by Baelde et al.~\cite{Baelde16csl, Baelde12tocl}, in which they similarly introduced a validity condition on the pre-proofs in multiplicative-additive linear logic with fixed points and proved the cut-elimination property for valid derivations. Doumane \cite{Doumane17phd} proved that this condition can be decided in $\mathit{PSPACE}$ by reducing it to the emptiness problem of B\"{u}chi automata. Nollet et al.~\cite{Nollet18csl} introduced a local polynomial time algorithm for identifying a stricter version of Baelde's condition. At present, it is not clear to us how their algorithm would compare with ours on the subsingleton fragment due to the differences between classical and intuitionistic sequents~\cite{Laurent18lics}, different criteria on locality, and the prevailing computational interpretation of cut reduction as communicating processes in our work. Cyclic proofs have also been used for reasoning about imperative programs with recursive procedures~\cite{Rowe17cpp}. While there are similarities (such as the use of cycles to capture recursion), their system extends separation logic is therefore not based on an interpretation of cut reduction as computation.  Reasoning in their logic therefore has a very different character from ours. Hyvernat ~\cite{Hyvernat19arxiv} introduced a condition to identify terminating and productive programs in ML{/}Haskell like recursive programs with mixed inductive and coinductive data types. Although their language cannot be reduced to circular proofs, their condition is inspired by the FS guard condition and is shown to be \textsc{PSPACE}. Also related is a result by  Das and Pous \cite{Das2018csl} on cyclic proofs in LKA, where Kleene star is interpreted as a least fixed point and proofs are interpreted as transformers. They introduced a condition that ensures cut elimination of cyclic proofs and termination of transformers. DeYoung and Pfenning~\cite{DeYoung16aplas} provide a computational interpretation of subsingleton logic with equirecursive fixed points and showed that cut reduction on circular pre-proofs in this system has the computational power of Turing machines. Their result implies undecidability of determining all programs with a strong progress property.

\noindent
{\bf Our contribution.} In this paper we have established an extension of the Curry-Howard interpretation for intuitionistic linear logic by Caires et al.~\cite{Caires10concur,Caires16mscs} to include least and greatest fixed points that can mutually depend on each other in arbitrary ways, although restricted to the subsingleton fragment.  The key is to interpret circular pre-proofs in subsingleton logic as mutually recursive processes, and to develop a locally checkable, compositional validity condition on such processes.  We proved that our local condition implies Fortier and Santocanale's guard condition and therefore also 
implies cut elimination.  Analyzing this result in more detail leads to a computational strong progress property which means that a valid program will always terminate either in an empty configuration or one attempting to communicate along external channels.

\noindent
{\bf Implementation.} We have implemented the algorithm introduced in Section~\ref{alg} in SML, which is publicly available~\cite{Das22Zenodo}\footnote{The reader can find an implementation of the examples presented in this paper in the subdirectory \href{https://bitbucket.org/fpfenning/rast/src/master/ss/examples/lmcs22.}{https://bitbucket.org/fpfenning/rast/src/master/ss/examples/lmcs22}.}. Currently, the implementation collects constraints and uses them to construct a suitable priority ordering over type variables and a $\subset$ ordering over process variables if they exist, and rejects the program otherwise. However, this precomputation step is not local and it makes it difficult to produce informative error messages. Our plan is to delete this step and rely on the programmer for priority and $\subset$ ordering. The implementation also supports an implicit syntax where the fixed point unfolding messages are synthesized from the given communication patterns.  Our experience with a range of programming examples shows that our local validity condition is surprisingly effective.  Its main shortcoming arises when, intuitively, we need to know that a program's output is ``smaller'' than its input.

\noindent
{\bf Future work.} The main path for future work is to extend our results to full ordered or linear logic with fixed points and address the known shortcomings we learned about through the implementation. The main shortcoming of our local validity condition arises when, intuitively, we need to know that a program's output is ``smaller" than its input. An interesting item for future research is to generalize our local validity condition to handle more cuts by introducing a way to capture the relation between input and output size. Studying this generalization also allows us to compare our results with the sized-type approach introduced by Abel and Pientka~\cite{abel2016well}. The first step in this general direction was taken by Sprenger and Dam~\cite{sprenger2003structure} who justify cyclic inductive proofs using inflationary iteration and the work by Somayyajula et al.~\cite{somayyajula2021circular} for shared memory concurrency. We would also like to investigate the relationship to work by Nollet et al.~\cite{Nollet18csl}, carried out in a different context.
\section*{Acknowledgments}

\noindent The authors would like to acknowledge helpful comments by Stephanie Balzer, Ankush Das, Jan Hoffmann, and Siva Somayyajula on an earlier draft, discussions with Jonas Frey and Henry DeYoung regarding the subject of this paper, and support from the National Science Foundation under grant CCF-1718267 \emph{Enriching Session Types for Practical Concurrent Programming}.


\bibliographystyle{alphaurl}

\appendix

\section{}
\label{proofapp}

Here we provide the proof for the observations we made in Section~\ref{Guard}. We prove that every program accepted by our algorithm in Section~\ref{alg} corresponds to a valid circular proof in the sense of the FS guard condition.
As explained in Section~\ref{priority}, the reflexive transitive closure of $\Omega$ in judgment ${x^\gamma: \omega \vdash_{\Omega} P :: (y^{\delta}:C) }$ forms a partial order $\le_{\Omega}$. To enhance readability of proofs, throughout this section we use entailment $\Omega \Vdash x \le y$ instead of $x\le_{\Omega}y$.

We first prove Lemmas~\ref{ordtrace} and~\ref{dordtrace} from Section~\ref{Guard}, in Lemmas~\ref{lem:ordtrace-app} and ~\ref{lem:dordtrace-app}, respectively. Theorem~\ref{listor} is a direct corollary of these two lemmas.

\begin{lemma}[corresponding to Lemma~\ref{ordtrace}] \label{lem:ordtrace-app}
Consider a path $\mathbb{P}$ in the (infinite) typing derivation of a program $\mathcal{Q}=\langle V, S\rangle$ defined on a Signature $\Sigma$:
\[{\infer{z^\alpha: \omega \vdash_{\Omega} P :: (w^{\beta}:C)}{ \infer{\vdots}{{x^\gamma: \omega' \vdash_{\Omega'} P' :: (y^{\delta}:C') }}}  }\]
with $n$ the maximum priority in $\Sigma$.
\begin{enumerate}
    \item[(a)] For every $i  \in \mathtt{c}(\omega') $ 
    with $\epsilon(i)=\mu$, if  $\Omega' \Vdash x^\gamma_i \le z^\alpha_i$
    then $x=z$ and $i \in \mathtt{c}(\omega) $.
    \item[(b)] For every $i<n$, if $\Omega' \Vdash x^\gamma_i < z^\alpha_i$, then $i \in \mathtt{c}(\omega) $ and a $\mu L$ rule with priority $i$ is applied on $\mathbb{P}$.
    \item[(c)] For every $c \le n$ with $\epsilon(c)=\nu$, if $\Omega' \Vdash x^\gamma_c \le z^\alpha_c$ , then no $\nu L$ rule with priority $c$ is applied on $\mathbb{P}$.
\end{enumerate}
\end{lemma}
\begin{proof}
Proof is by induction on the structure of $\mathbb{P}$.
We consider each case for last (topmost) step in $\mathbb{P}$.
\begin{description}
\item[Case]  \[\infer[\msc{Def}(X)]{x^\gamma:\omega' \vdash_{\Omega'} y^\delta \leftarrow X \leftarrow x^\gamma ::(y^{\delta}:C')}{x^\gamma: \omega' \vdash_{\Omega'} P' :: (y^{\delta}:C') & x:\omega' \vdash X=P' :: (y:C') \in V }\]
    None of the conditions in the conclusion are different from the premise. Therefore, by the induction hypothesis, statements (a)-(c) hold. 

\item[Case] \[\infer[\msc{Cut}^{y}]{ x^\gamma: \omega' \vdash_{\Omega''}  (y \leftarrow P'_y ; Q_y) :: (v^{\theta}:C'')}{ x^\gamma: \omega' \vdash_{\Omega''\cup \mathtt{r}(v^\theta)} P'_{y^0} :: (y^0: C') & y^0: C' \vdash_{\Omega''\cup \mathtt{r}(x^\gamma)} Q_{y^0} ::(v^{\theta}: C'')}\]
         where $\mathtt{r}(u)= \{y^{0}_{j}=u_{j}\mid j \not \in \mathtt{c}(C')\, \mbox{and}\, j \le n\}$ and $\Omega'=\Omega'' \cup r(v^{\theta})$.
         All conditions in the conclusion are the same as the premise: the equations in $r(v^{\theta})$ only include channels $y^0$ and $v^\theta$. As a result $\Omega''\cup \mathtt{r}(v^\theta)\Vdash x^\gamma_i \le z^\alpha_i$ implies $\Omega''\Vdash x^\gamma_i \le z^\alpha_i$, and $\Omega''\cup \mathtt{r}(v^\theta)\Vdash x^\gamma_i < z^\alpha_i$ implies $\Omega''\Vdash x^\gamma_i < z^\alpha_i$. Therefore, by the induction hypothesis, statements (a)-(c) hold.
    
\item[Case] \[\infer[\msc{Cut}^{x}]{ u^\eta: \omega'' \vdash_{\Omega''}  (x \leftarrow Q_x ; P'_x) :: (y^{\delta}:C')}{ u^\eta: \omega'' \vdash_{\Omega''\cup \mathtt{r}(y^\delta)} Q_{a^0} :: (x^0: A) & x^0: A \vdash_{\Omega''\cup \mathtt{r}(u^\eta)} P'_{x^0} ::(y^{\delta}: C'')}\]
         where $\mathtt{r}(v)= \{x^{0}_{j}=v_{j}\mid j \not \in \mathtt{c}(A)\,\mbox{and}\, j \le n\}$ and $\Omega'=\Omega'' \cup r(u^\eta)$.
         \begin{enumerate}
             \item[(a)] $\Omega'' \cup r(u^\eta) \Vdash x^0_i \le z^\alpha_i$ does not hold for any $i\in c(A)$: $x$ is a fresh channel and does not occur in the equation of $\Omega''$. Moreover, since $i \in c(A)$, there is no equation in the set $r(u^\eta)$ including $x^0_i$. Therefore, this part is vacuously true.
             \item [(b)] By freshness of $x$, if $\Omega'' \cup r(u^\eta)\Vdash x_i^{0} < z^{\alpha}_{i}$, then $x_i^{0} ={u^{\eta}_{i}} \in r(u^\eta)$ and $\Omega'' \Vdash u^{\eta}_{i}< z^{\alpha}_{i}$.
             By the induction hypothesis,  $i \in \mathtt{c}(\omega)$ and a  $\mu L$ rule with priority $i$ is applied on  $\mathbb{P}$.
             \item[(c)] By freshness of $x$, if $\Omega'' \cup r(u^\eta)\Vdash x_c^{0} \le z^{\alpha}_{c}$, then $x_c^{0} ={u^{\eta}_{c}} \in r(u^\eta)$ and $\Omega'' \Vdash u^{\eta}_{c}\le z^{\alpha}_{c}$. By the induction hypothesis, no  $\nu L$ rule with priority $c$ is applied on $\mathbb{P}$.
         \end{enumerate}
           
\item[Case] \[\infer[1L]{u^\eta: 1 \vdash_{\Omega'} \mathbf{wait}\, Lu^\eta;P' :: (y^{\delta}: C')}{ . \vdash_{\Omega'} P' :: (y^{\delta}: C') }\]
 This case is not applicable since by the typing rules $\Omega' \not \Vdash .\le z^{\alpha}_i$ for any $i\le n$.
 
\item[Case] \[ \infer[\mu R]{ x^\alpha:  \omega' \vdash_{\Omega''} Ry^{\delta'}.\mu_t; P :: (y^{\delta'}:t)}{x^\alpha:  \omega' \vdash_{\Omega'} P :: (y^{\delta'+1}:C') & t=_{\mu}C' & \Omega'= \Omega'' \cup \{(y^{\delta'})_{p(s)} =(y^{\delta'+1})_{p(s)} \mid p(s)\neq p(t)\}} \]
   For every $i\le n$, if $\Omega'' \cup \{(y^{\delta'})_{p(s)} =(y^{\delta'+1})_{p(s)} \mid p(s)\neq p(t)\} \Vdash x^\gamma_i \le z^\alpha_i$, then $\Omega'' \Vdash x^\gamma_i \le z^\alpha_i $. Therefore, by the induction hypothesis, statements (a)-(c) hold.
   
\item[Case] {\small\[ \infer[\mu L]{x^{\gamma'}: t \vdash_{\Omega''} \mathbf{case}\, Lx^{\gamma'}\ (\mu_{t} \Ra P'):: (y^{\delta}:C')}{  x^{\gamma'+1}: \omega' \vdash_{\Omega'} P' :: (y^{\delta}:C') & t=_{\mu} \omega' & \Omega'=\Omega'' \cup \{x^{\gamma'+1}_{p(t)} < x^{\gamma'}_{p(t)}\} \cup \{x^{\gamma'+1}_{p(s)} =x^{\gamma'}_{p(s)} \mid p(s)\neq p(t)\}}\]}
By definition of $\mathtt{c}(x)$, we have $\mathtt{c}(\omega') \subseteq \mathtt{c}(t)$. By $\mu L$ rule, for all $i\le n$, $x^{\gamma'+1}_{i} \le x^{\gamma'}_i \in \Omega'$. But by freshness of channels and their generations, $x^{\gamma'+1}$ is not involved in any relation in $\Omega''$.
\begin{enumerate}
    \item[(a)]For every $i \in \mathtt{c}(\omega')$ with $\epsilon(i)=\mu$, if $\Omega' \Vdash x^{\gamma'+1}_{i} \le z^{\alpha}_{i}$ then $\Omega'' \Vdash x^{\gamma'}_{i} \le z^{\alpha}_{i}$. By the induction hypothesis, we have $x=z$ and $i \in \mathtt{c}(\omega)$.
    \item[(b)]We consider two subcases: (1) If $\Omega' \Vdash x^{\gamma'+1}_i< z_i^{\alpha}$ for $i \neq p(t)$, then $\Omega'\Vdash x^{\gamma'+1}_i= x^{\gamma'}_i$ and $\Omega'' \Vdash x^{\gamma'}_i< z_i^{\alpha}$. Now we can apply the induction hypothesis. (2) If $\Omega' \Vdash x^{\gamma'+1}_{p(t)}< z_{p(t)}^{\alpha}$, then $\Omega' \Vdash x^{\gamma'+1}_{p(t)} < x^{\gamma'}_{p(t)}$ and $\Omega'' \Vdash x^{\gamma'}_{p(t)} \le z^{\alpha}_{p(t)}$. Since a $\mu L$ rule is applied in this step on the priority $p(t)$, we only need to prove that $p(t) \in \mathtt{c}(\omega)$. By definition of $\mathtt{c}$, we have $p(t) \in \mathtt{c}(t) $ and we can use the induction hypothesis on part (a) to get $p(t) \in \mathtt{c}(\omega)$. 
    \item[(c)] For every $c \le n$ with $\epsilon(c)= \nu$, $\Omega' \Vdash x^{\gamma'+1}_c= x^{\gamma'}_c$ as $c \neq p(t)$. Therefore, if $\Omega' \Vdash x^{\gamma'+1}_c= z^{\alpha}_c$, then $\Omega'' \Vdash x^{\gamma'}_c= z^{\alpha}_c$. By the induction hypothesis no $\nu L$ rule with priority $c$ is applied on $\mathbb{P}$.
  \end{enumerate} 

\item[Case] {\small\[\infer[\nu R]{x^\gamma: \omega' \vdash_{\Omega''} \mathbf{case}\, Ry^{\delta'} \ (\nu_t \Ra P') :: (y^{\delta'}:t)}{  x^\gamma: \omega' \vdash_{\Omega'} P' :: y^{\delta'+1}: C'& t=_{\nu}C'&
\Omega'= \Omega'' \cup \{y^{\delta'+1}_{p(t)} < y^{\delta'}_{p(t)}\} \cup \{y^{\delta'+1}_{p(s)} =y^{\delta'}_{p(s)} \mid p(s)\neq p(t)\}  } \]}
   For every $i\le n$, if $\Omega'' \cup \{y^{\delta'+1}_{p(t)} < y^{\delta'}_{p(t)}\} \cup \{y^{\delta'+1}_{p(s)} =y^{\delta'}_{p(s)} \mid p(s)\neq p(t)\} \Vdash x^\gamma_i \le z^\alpha_i$, then $\Omega'' \Vdash x^\gamma_i \le z^\alpha_i $. Therefore, by the induction hypothesis, statements (a)-(c) hold.
   
\item[Case] \[ \infer[\nu L]{x^{\gamma'}: t \vdash_{\Omega''} Lx^{\gamma'}.\nu_{t}; P':: (y^{\delta}:C')}{  x^{\gamma'+1}: \omega' \vdash_{\Omega'} Q :: (y^{\delta}:C') & t=_{\nu} \omega' & \Omega'=\Omega'' \cup \{x^{\gamma'+1}_{p(s)} =x^{\gamma'}_{p(s)} \mid p(s)\neq p(t)\} }\]
By definition of $\mathtt{c}(x)$, we have $\mathtt{c}(\omega') \subseteq \mathtt{c}(t)$. By $\nu L$ rule, for all $i\neq p(t)\le n$, $x^{\gamma'+1}_{i} = x^{\gamma'}_i \in \Omega'$. In particular, for every $i \le n $ with $\epsilon(i)=\mu$, $x^{\gamma'+1}_{i} = x^{\gamma'}_i \in \Omega'$. But by freshness of channels and their generations, $x^{\gamma'+1}$ is not involved in any relation in $\Omega''$.
\begin{enumerate}
    \item[(a)] For every $i  \in \mathtt{c}(\omega') $ with $\epsilon(i)=\mu$, if  $\Omega' \Vdash x^{\gamma'+1}_i \le z^\alpha_i$, then $\Omega' \Vdash x^{\gamma'+1}_i=x^{\gamma'}_i$ and $\Omega'' \Vdash  x^{\gamma'}_i \le z^\alpha_i$. By the induction hypothesis $x=z$ and $i \in \mathtt{c}(\omega)$.
    \item[(b)] If $\Omega' \Vdash x^{\gamma'+1}_i< z_i^{\alpha}$, then by freshness of channels and their generations we have $i \neq p(t)$, $\Omega' \Vdash x^{\gamma'+1}_i=x^{\gamma'}_i$ and $\Omega'' \Vdash x^{\gamma'}_i< z_i^{\alpha}$. By the induction hypothesis $i \in \mathtt{c}(\omega)$ and a $\mu L$ rule with priority $i$ is applied on the path.
    \item[(c)] For every $c \le n$ with $\epsilon(c)= \nu$ and $c\neq p(t)$, if $\Omega' \Vdash x^{\gamma'+1}_c \le z^\alpha_c$, then $\Omega'\Vdash x^{\gamma'+1}_c = x^{\gamma'}_c$ and $\Omega'' \Vdash x^{\gamma'} \le z^\alpha_c$. Therefore, by induction hypothesis, no $\nu L$ rule with priority $c$ is applied on the path. Note that $\Omega' \not \Vdash x^{\gamma'+1}_p(t) \le z^{\alpha}_p(t)$. 
\end{enumerate}

\item[Case] \[\infer[\& R]{x^\gamma: \omega \vdash_{\Omega'} \mathbf{case}\, Ry^\delta \ (\ell \Ra Q_{\ell}) :: (y^{\delta}:\&\{\ell:A_{\ell}\}_{\ell \in L})}{  x^\gamma: \omega \vdash_{\Omega'} Q_k :: y^{\delta}: A_k & \forall k \in L  } \]
  None of the conditions in the conclusion are different from the premise. Therefore, by the induction hypothesis, statements (a)-(c) hold.

\item[Case] \[ \infer[\& L]{x^\gamma: \&\{\ell:A_{\ell}\}_{\ell \in L} \vdash_{\Omega'} Lx^\gamma.k; Q:: (y^{\delta}:C')}{  x^{\gamma}: A_k \vdash_{\Omega'} Q :: (y^{\delta}:C') }\]
By definition of $\mathtt{c}(x)$, we have $\mathtt{c}(A_k) \subseteq \mathtt{c}(\&\{\ell:A_{\ell}\}_{\ell \in L})$. Therefore, statements (a)-(c) follow from the induction hypothesis.

\item[Cases] The statements are trivially true if the last step of the proof is either $1R$ or $\msc{Id}$ rules. \qedhere
\end{description}
\end{proof}

\begin{lemma}[corresponding to Lemma~\ref{dordtrace}]\label{lem:dordtrace-app}
Consider a path $\mathbb{P}$ in the (infinite) typing derivation of a program $\mathcal{Q}=\langle V, S\rangle$ defined on a Signature $\Sigma$,
\[{\infer{\bar{z}^\alpha: \omega \vdash_{\Omega} P ::(w^{\beta}: C)}{ \infer{\vdots}{{\bar{x}^\gamma: \omega' \vdash_{\Omega'} P' :: ( y^{\delta}:C') }}}  }\]
with $n$ the maximum priority in $\Sigma$.
\begin{enumerate}
    \item[(a)] For every $i \in \mathtt{c}(\omega')$ with $\epsilon(i)=\nu$, if $\Omega' \Vdash y^{\delta}_i \le w^{\beta}_i$, then $y=w$ and $i \in \mathtt{c}(\omega)$.
    \item[(b)] If $\Omega' \Vdash y^{\delta}_i< w_i^{\beta}$, then $i \in \mathtt{c}(\omega) $ and a $\nu L$ rule with priority $i$ is applied on $\mathbb{P}$ .
    \item[(c)] For every $c\le n$ with $\epsilon(c)=\mu$, if $\Omega'\Vdash y^{\delta}_c \le w^{\beta}_c$, then no $\mu R$ rule with priority $c$ is applied on $\mathbb{P}$ .
\end{enumerate}
\end{lemma}
\begin{proof}
Dual to the proof of Lemma~\ref{lem:ordtrace-app}.
\end{proof}

\begin{lemma}\label{basics}
Consider a path $\mathbb{P}$ on a program  $\mathcal{Q}=\langle V, S\rangle$ defined on a Signature $\Sigma$, with $n$ the maximum priority in $\Sigma$.
\[{\infer{\bar{z}^\alpha: \omega \vdash_{\Omega} P ::(w^{\beta}: C)}{ \infer{\vdots}{{\bar{x}^\gamma: \omega' \vdash_{\Omega'} P' :: ( y^{\delta}:C') }}}  }\]
$\Omega'$ preserves the (in)equalities in $\Omega$. In other words,
for channels $u,v$, generations $\eta, \eta' \in \mathbb{N}$ and type priorities $i, j\le n$,
\begin{enumerate}
    \item[(a)]  If $\Omega \Vdash u_i^\eta < v_j^{\eta'}$, then $\Omega' \Vdash u_i^\eta <v_j^{\eta'}$.
    \item[(b)] If $\Omega \Vdash u_i^\eta \le v_j^{\eta'}$, then $\Omega' \Vdash u_i^\eta \le v_j^{\eta'}$.
    \item[(c)] If $\Omega \Vdash u_i^\eta = v_j^{\eta'}$, then $\Omega' \Vdash u_i^\eta = v_j^{\eta'}$.
\end{enumerate}
\end{lemma}
\begin{proof}
Proof is by induction on the structure of $\mathbb{P}$. We consider each case for topmost step in $\mathbb{P}$. Here, we only give one non-trivial case. The proof of other cases is similar.
\begin{description}
    \item[Case] \[ \infer[\mu R]{ x^\alpha:  \omega' \vdash_{\Omega''} Ry^{\delta'}.\mu_t; P :: (y^{\delta'}:t)}{x^\alpha:  \omega' \vdash_{\Omega'} P :: (y^{\delta'+1}:C') & t=_{\mu}C' & \Omega'= \Omega'' \cup \{(y^{\delta'})_{p(s)} =(y^{\delta'+1})_{p(s)} \mid p(s)\neq p(t)\}} \]
    \begin{itemize}
        \item[(a)] If $\Omega \Vdash u_i^\eta < v_j^{\eta'}$, then by the inductive hypothesis, $\Omega'' \Vdash u_i^\eta < v_j^{\eta'}$. By freshness of channels and their generations, we know that $y^{\delta'+1}$ does not occur in any (in)equalities in $\Omega''$ and thus $y^{\delta'+1} \neq u^\eta, v^{\eta'}$. Therefore $\Omega' \Vdash u_i^\eta < v_j^{\eta'}.$
    \end{itemize}
   Following the same reasoning, we can prove statements (b) and (c). \qedhere
   
\end{description}
\end{proof}
\begin{lemma}\label{lem:coinduction}
Consider a finitary derivation (Figure~\ref{fig:validity}) for \[\langle\bar{u}, X, v\rangle; \bar{x}^\alpha: \omega \vdash_{\Omega,\subset} P ::(y^\beta:C),\] on a locally valid program $\mathcal{Q}= \langle V, S \rangle$ defined on signature $\Sigma$ and order $\subset$. Consider a process variable $X'$ and a list $list(\bar{u}',v')$ formed using channels $\bar{u}'$ and $v'$ such that 
either $X,list(\bar{u},v) \mathrel{(\subset, <_{\Omega})} X',list(\bar{u}',v'),$ or $X=X'$ and $list(\bar{u},v)=list(\bar{u'},v')$.
There is a (potentially infinite) derivation $\mathbb{D}$ for 
\[\bar{x}^\alpha: \omega \vdash_{\Omega} P ::(y^\beta:C),\]
based in the infinitary rule system of Figure~\ref{fig:stp-order}.

Moreover, for every $\bar{w}^{\gamma}:\omega' \vdash_{\Omega'} z^{\delta} \leftarrow Y \leftarrow \bar{w}^{\gamma} :: (z^{\delta}:C')$ in $\mathbb{D}$, we have  \[\star \; Y,list(\bar{w}^{\gamma},z^{\delta}) \mathrel{(\subset, <_{\Omega'})} X',list(\bar{u'},v').\]
\end{lemma}
\begin{proof}
The proof is by coinduction, producing the derivation of $\bar{x}^\alpha: \omega \vdash_{\Omega} P ::(y^\beta:C)$ one sequent at a time. Moreover, for every produced sequent, we make sure that, if applicable, the $\star$ property holds. The correctness of the coinduction proof follows from the fact that $\star$ is a property about each finite portion of the infinite derivation and is checked after the portion is entirely produced. Lemma~\ref{algtoinf} is a special case of this lemma, in which $X',list(\bar{u}',v')$ is instantiated as $X,list(\bar{u},v)$. Here, we generalize the coinduction hypothesis of Lemma~\ref{algtoinf} for any $X',list(\bar{u}',v')$ to make the proof work.

We proceed by case analysis of the first rule applied on 
$\langle \bar{u}, X, v\rangle; \bar{x}^\alpha: \omega \vdash_{\Omega, \subset} P ::(y^\beta:C),$
in its finite derivation. 
\begin{description}
    \item[Case]  \[\infer[\msc{Call}(Y)]{\langle \bar{u}^\gamma, X , v^\delta \rangle; \bar{x}^{\alpha}:\omega \vdash_{\Omega, \subset} y^{\beta} \leftarrow Y \leftarrow \bar{x}^{\alpha} ::(y^{\beta}:C)}{Y, list(\bar{x}^\alpha, y^\beta) \mathrel{(\subset, <_{\Omega})} X, list(\bar{u}^\gamma, v^\delta) & \bar{x}:\omega \vdash Y=P'_{\bar{x}, y} :: (y:C) \in V }\]
    By validity of the program there is a finitary derivation for \[\langle\bar{x}^0, Y, y^0\rangle;\bar{x}^0: \omega \vdash_{\emptyset, \subset} P'_{\bar{x}^0, y^0} ::(y^0:C).\]
    
    Having Proposition~\ref{subuv} and freshness of future generations of channels in $\Omega$, there is also a finitary derivation for 
     \[\langle\bar{x}^{\alpha}, Y, y^{\beta}\rangle;\bar{x}^{\alpha}: \omega \vdash_{\Omega, \subset} P'_{\bar{x}^{\alpha}, y^{\beta}} ::(y^{\beta}:C).\]
    By transitivity of $(\subset, <_{\Omega}),$ we get \[Y, list(\bar{x}^\alpha, y^\beta) \mathrel{(\subset, <_{\Omega})}  X', list(\bar{u}', v').\]
    We apply the coinductive hypothesis to get an infinitary derivation $\mathbb{D'}$ for
     \[\bar{x}^{\alpha}: \omega \vdash_{\Omega, \subset} P'_{\bar{x}^{\alpha}, y^{\beta}} ::(y^{\beta}:C),\]
    and then produce the last step of derivation
    \[\infer[\msc{Def}(Y)]{ \bar{x}^{\alpha}:\omega \vdash_{\Omega} y^{\beta} \leftarrow Y \leftarrow \bar{x}^{\alpha} ::(y^{\beta}:C)}{\deduce{\bar{x}^{\alpha}:\omega \vdash_{\Omega} P'_{\bar{x}^{\alpha}, y^{\beta}} ::(y^{\beta}:C)}{\mathbb{D'}} & \bar{x}:\omega \vdash Y=P'_{\bar{x}, y} :: (y:C) \in V }\]
    in the infinitary rule system.
    
    Moreover, by the coinductive hypothesis, we know that for every \[\bar{w}^{\gamma'}:\omega' \vdash_{\Omega'} z^{\delta'} \leftarrow W \leftarrow w^{\gamma'} :: (z^{\delta'}:C')\] on $\mathbb{D'}$, we have \[W,list(w^{\gamma'},z^{\delta'}) \mathrel{(\subset, <_{\Omega'})} X',list(\bar{u}',v').\] 
   
    This completes the proof of this case as we already know  $Y, list(\bar{x}^\alpha, y^\beta) \mathrel{(\subset, <_{\Omega})} X', list(\bar{u}', v')$.
    \item[Case] {\small\[\infer[\msc{Cut}^{z}]{\langle \bar{u}^\gamma, X , v^\delta \rangle; \bar{x}^\alpha: \omega \vdash_{\Omega, \subset}  (z \leftarrow Q_z ; Q'_z) :: (y^{\beta}:C)}{ \langle \bar{u}^\gamma, X , v^\delta \rangle; \bar{x}^\alpha: \omega \vdash_{\Omega \cup \mathtt{r}(y^\beta), \subset} Q_{z^0} :: (z^0: C') & \langle \bar{u}^\gamma, X , v^\delta \rangle;z^0: C' \vdash_{\Omega\cup \mathtt{r}(x^\alpha), \subset} Q'_{z^0} ::(y^{\beta}: C)},\]}
     where $\mathtt{r}(w)= \{z^{0}_{j}=w_{j}\mid j \not \in \mathtt{c}(A)\, \mbox{and}\, j \le n\}.$
     
    By Lemma~\ref{basics}, we conclude from
    $X, list(\bar{u}^\gamma, v^\delta) \mathrel{(\subset, <_{\Omega})} X', list(\bar{u}', v')$ that
    \[X, list(\bar{u}^\gamma, v^\delta) \mathrel{(\subset, <_{\Omega \cup \mathtt{r}(y^\beta)})} X', list(\bar{u}', v'),\] and \[X, list(\bar{u}^\gamma, v^\delta) \mathrel{(\subset, <_{\Omega \cup \mathtt{r}(x^\alpha)})} X', list(\bar{u}', v').\]

    By coinductive hypothesis, we have infinitary derivations $\mathbb{D'}$ and $\mathbb{D''}$ for $\bar{x}^\alpha: \omega \vdash_{\Omega \cup \mathtt{r}(y^\beta)} Q_{z^0} :: (z^0: C')$
    and 
    $z^0: C' \vdash_{\Omega\cup \mathtt{r}(x^\alpha)} Q'_{z^0} ::(y^{\beta}: C),$ respectively. We can produce the last step of the derivation as
     \[\infer[\msc{Cut}^{z}]{ \bar{x}^\alpha: \omega \vdash_{\Omega}  (z \leftarrow Q_z ; Q'_z) :: (y^{\beta}:C)}{ \deduce{ \bar{x}^\alpha: \omega \vdash_{\Omega \cup \mathtt{r}(y^\beta)} Q_{z^0} :: (z^0: C')}{\mathbb{D'}} & \deduce{z^0: C' \vdash_{\Omega\cup \mathtt{r}(x^\alpha)} Q'_{z^0} ::(y^{\beta}: C)}{\mathbb{D''}}}\]
     Moreover, by the coinductive hypothesis, we know that for every \[\bar{w}^{\gamma}:\omega' \vdash_{\Omega'} z^{\delta} \leftarrow W \leftarrow \bar{w}^{\gamma} :: (z^{\delta}:C')\] on $\mathbb{D'}$ and $\mathbb{D''}$, and thus $\mathbb{D}$, we have  \[W,list(\bar{w}^{\gamma},z^{\delta}) \mathrel{(\subset, <_{\Omega'})} X',list(\bar{u}',v').\]
     \item[Cases] The proof of the other cases are similar by applying the coinductive hypothesis and the infinitary system rules. \qedhere
\end{description}
\end{proof}
\end{document}